\documentclass[11pt,reqno]{amsart}
\usepackage[a4paper]{anysize}\marginsize{2.5cm}{2.5cm}{2cm}{2cm}
\pdfpagewidth=\paperwidth \pdfpageheight=\paperheight
\allowdisplaybreaks
\usepackage{mlmodern}
\usepackage[english]{babel}
\usepackage[latin1]{inputenc}
\usepackage[T1]{fontenc}
\usepackage{amssymb,amsmath,amstext,amsfonts,amsthm,amscd,latexsym, mathrsfs,amsbsy}
\usepackage{mathtools}
\usepackage{verbatim}
\usepackage{relsize}
\usepackage[colorlinks=true, urlcolor=blue, linkcolor=blue, citecolor=blue]{hyperref}
\usepackage{tikz-cd}
\usepackage[normalem]{ulem} % for striking out
\usepackage{caption}
\usepackage{graphicx}
\usepackage{array}
\usepackage{arydshln}
\usepackage{xcolor}
\usepackage[matrix,arrow]{xy}
\usepackage{booktabs}
\usepackage{cleveref}
\usepackage{hhline}
\usepackage{cprotect} % for using \verb|| within a caption
\usepackage{enumitem,kantlipsum}

\numberwithin{equation}{section}
\theoremstyle{plain}
\newtheorem*{theorem*}{Theorem}
\newtheorem{theorem}{Theorem}
\numberwithin{theorem}{section}
\newtheorem{proposition}[theorem]{Proposition}
\newtheorem{lemma}[theorem]{Lemma}

\newtheorem{problem}[theorem]{Problem}

\newtheorem{remark}[theorem]{Remark}

\theoremstyle{definition}
\newtheorem{definition}[theorem]{Definition}
\newtheorem{example}[theorem]{Example}

\newcommand*{\mge}[1][]{\mathsmaller{\ge #1}}

\newcommand{\C}{{\mathbb C}}
\newcommand{\Q}{{\mathbb Q}}

\newcommand{\R}{{\mathbb R}}
\newcommand{\Z}{{\mathbb Z}}

\newcommand{\PP}{{\mathbb P}}

\def\bv{{\boldsymbol{v}}}

\def\bj{{\boldsymbol{j}}}

\def\bx{{\boldsymbol{x}}}
\def\ba{{\boldsymbol{a}}}

\def\bd{{\boldsymbol{d}}}
\def\balpha{{\boldsymbol{\alpha}}}

\def\bpi{{\boldsymbol{\pi}}}
\def\bone{{\boldsymbol{1}}}
\def\bzero{{\boldsymbol{0}}}
\def\codim{{\rm codim}}

\newcommand{\mT}{\mathsmaller{\mathsf{T}}}
\newcommand{\mR}{\mathsmaller{\mathbb{R}}}

\newcommand{\kc}{{\mathcal C}}
\newcommand{\kd}{{\mathcal D}}
\newcommand{\kl}{{\mathcal L}}

\newcommand{\kn}{{\mathcal N}}
\newcommand{\ko}{{\mathcal O}}

\newcommand{\kr}{{\mathcal R}}
\newcommand{\ks}{{\mathcal S}}

\newcommand{\ku}{{\mathcal U}}

\newcommand{\ky}{{\mathcal Y}}
\newcommand{\kz}{{\mathcal Z}}

\makeatletter
\renewcommand*\env@matrix[1][*\c@MaxMatrixCols c]{%
    \hskip -\arraycolsep
    \let\@ifnextchar\new@ifnextchar
    \array{#1}}
\makeatother

\makeatletter
\@namedef{subjclassname@2020}{%
	\textup{2020} Mathematics Subject Classification}
\makeatother

\author{Hirotachi Abo}
\address{Department of Mathematics, University of Idaho, Moscow, Idaho 83844-1103, United States of America}
\email{abo@uidaho.edu}

\author{Irem Portakal}
\address{Max Planck Institute for Mathematics in the Sciences, Leipzig, Germany}
\email{mail@irem-portakal.de}

\author{Luca Sodomaco}
\address{Max Planck Institute for Mathematics in the Sciences, Leipzig, Germany}
\email{luca.sodomaco@mis.mpg.de}



\title[A vector bundle approach to Nash equilibria]{A vector bundle approach to Nash equilibria}

\date{}

\begin{document}

\begin{abstract}
    We use vector bundles to study the locus of totally mixed Nash equilibria of an $n$-player game in normal form, which we call the Nash equilibrium scheme. When the payoff tensor format is balanced, we study the Nash discriminant variety, i.e., the algebraic variety of games whose Nash equilibrium scheme is nonreduced or has a positive dimensional component. We prove that this variety has codimension one. We classify all possible components of the Nash equilibrium scheme for a binary three-player game. We prove that if the payoff tensor is of boundary format, then the Nash discriminant variety has two components: an irreducible hypersurface and a larger-codimensional component. A generic game with an unbalanced payoff tensor format does not admit totally mixed Nash equilibria. We define the Nash resultant variety of games admitting a positive number of totally mixed Nash equilibria. We prove that it is irreducible and determine its codimension and degree.
\end{abstract}

\maketitle

\section{Introduction}

The study of Nash equilibria has impacted many areas beyond mathematics, including economics, computer science, evolutionary biology, quantum mechanics, and social science. It is well-known that every finite game has at least one Nash equilibrium \cite{nash1950equilibrium}. However, finding Nash equilibria is known to be PPAD-complete \cite{daskalakis2009complexity}, and even NP-complete in specific cases, such as finding a second Nash equilibrium, one that maximizes the sum of players' utilities, or one that uses a given strategy with positive probability \cite{conitzer2008complexity,gilboa1989nash}. Despite these computational challenges, the set of equilibria can still be described using algebro-geometric methods. The set of totally mixed Nash equilibria is a semialgebraic set that is defined as the intersection of the \emph{Nash equilibrium scheme} (Definition~\ref{def: Nash equilibria scheme}) with the probability simplices that correspond to the mixed strategies of the players. By applying the Bernstein-Khovanskii-Kushnirenko (BKK) theorem, a classical result from algebraic geometry, McKelvey and McLennan provided an upper bound on the number of totally mixed Nash equilibria of \emph{generic} games \cite{mckelvey1996computation,mckelvey1997maximal}. Since any Nash equilibrium gives rise to a totally mixed equilibrium of the smaller game obtained by eliminating all unused pure strategies, this upper bound is a lower bound for the maximal number of Nash equilibria of generic games.

Under the growing field of algebraic game theory, these developments have been extended to various types of equilibria and trace their inspiration back to the algebro-geometric study of Nash equilibria \cite{wilson1971computing}. For instance, the study of correlated equilibria has significantly benefited from convex geometry \cite{brandenburg2024combinatorics}, and strong connections have been discovered between oriented matroids, elliptic curves, rational varieties, and dependency equilibria via the use of Spohn varieties \cite{kidambi2025elliptic,portakal2022geometry}. In particular, it is proven that any Nash equilibrium lies on the Spohn variety of the game \cite[Theorem 3.18]{portakal2024dependency}. Our paper aims to highlight the profound and valuable connection between algebraic geometry and Nash equilibria. In particular, we study totally mixed Nash equilibria of $n$-player games in normal form using vector bundles over a product of $n$ projective spaces. This allows us to study the space of nongeneric games and their totally mixed Nash equilibria.

The notion of generic games has appeared under different names in previous game theory literature. Wilson \cite{wilson1971computing} refers to these as {\em nondegenerate games}, where he gives a first proof of the oddness and finiteness of Nash equilibria. Harsanyi \cite[Section 5]{harsanyi1973oddness} describes them as {\em almost all games}, i.e., for all games except for a closed set of measure zero, repeating Wilson's definition. McKelvey and McLennan \cite{mckelvey1997maximal} follow these definitions and use the term {\em generic games}. In this work, we call a finite game in normal form {\em generic} if it belongs to the complement of the {\emph{Nash discriminant variety}} or the {\emph{Nash resultant variety}} (Section~\ref{sec: Nash discriminants}). Our definition is thus motivated by algebraic geometry but aligns with those established in earlier game theory papers. Indeed, the study of nongeneric games, i.e., games whose Nash equilibrium schemes are nonreduced or positive dimensional, is not uncommon. The idea of considering discriminants appeared in \cite[Section 6]{mckelvey1997maximal} for three-player games with binary strategies and elaborated further in \cite{emiris2016compact}. This is, in particular, an example of a game with a nonreduced Nash equilibrium scheme.  For the same classes of games, in \cite[Section 2]{chin1974structure}, one finds a game whose Nash equilibrium scheme is a line. The concept of resultants was also discussed in the context of Nash equilibria in \cite{mckelvey1996computation}. A symbolic method for obtaining a parametric representation of totally mixed Nash equilibria is presented in \cite{jeronimo2009parametric}. This method is based on a symbolic procedure for calculating multihomogeneous resultants. Moreover, Datta \cite{datta2003universality} showed that for any real algebraic variety, it is possible to construct a three-player or an $n$-player game with binary strategies whose Nash equilibrium scheme is isomorphic to that variety. This further highlights the critical role algebraic geometry plays in advancing the study of Nash equilibria. 

We let $X$ be an $n$-player game $\bd=(d_1,\ldots,d_n)\in\Z_{\mge[2]}^n$ with $d_1\le\cdots\le d_n$. This means each player $i \in [n]$ can choose from $d_i$ pure strategies. In Section~\ref{subsec: vector bundles}, we formulate the Nash equilibrium scheme as the zero scheme of a global section of the vector bundle $E$ on the multiprojective space $\PP^\bd$ of rank $D=\dim\PP^\bd$. In the following, we revisit a well-known result \cite{emiris2016compact, mckelvey1997maximal, sturmfels2002solving} for generic games with vector bundles.

\begin{theorem*}[Theorem~\ref{thm: number tmNe generic game}]
If $X$ is a generic game, then the following three conditions are equivalent:
\begin{itemize}
    \item[(1)] The Nash equilibrium scheme is empty.
    \item[(2)] The degree of the top Chern class $c(\bd)$ of the vector bundle $E$ is zero.
    \item[(3)] $d_n-1>\sum_{i=1}^{n-1} (d_i-1)$. 
\end{itemize}
\end{theorem*}

The property of $c(\bd)=0$, called \emph{beyond boundary format} brings us the following problem: {\em given a beyond boundary format $\bd$, can we describe the set of (nongeneric) $n$-player games of format $\bd$ admitting at least one Nash equilibrium?}

\noindent The answer is positive, and we can say more: this locus is an algebraic variety, and we call it the \emph{Nash resultant variety} $\kr(\bd)$ (Definition~\ref{def: Nash resultant variety}). 

\begin{theorem*}[Proposition~\ref{prop: Nash resultant variety two players}, Theorem~\ref{thm: codim degree Nash resultant variety}]
The Nash resultant variety $\kr(\bd)$ is irreducible and has codimension 
\[
\codim\,\kr(\bd) = d_n-1-\sum_{i=1}^{n-1}(d_i-1)\,.
\]
Its degree is 
\begin{align*}
\begin{split}
\deg\kr(\bd) = \frac{(d_n-1)!}{(d_1-1)! \cdots (d_{n-1}-1)! (d_n-1-\sum_{i=1}^{n-1}(d_i-1))!}\,.
\end{split}
\end{align*}
\end{theorem*}
\noindent In particular, for two-player games, the ideal of the Nash resultant variety is the maximal minors of the $(d_1 +1)\times d_2$ matrix obtained by adding the row consisting of ones to the payoff matrix of the second player. 

We consider a similar question for nongeneric games \emph{within boundary format} where $d_n-1\le \sum_{i=1}^{n-1} (d_i-1)$: describe the set of nongeneric games with an ``unexpected'' number of totally mixed Nash games. 
This within boundary format condition, in particular, was also earlier used to give a necessary and sufficient condition for the existence of a game with a unique totally mixed Nash equilibrium \cite{kreps1981finite}. For this, we define the Nash discriminant variety $\Delta(\bd)$ (Definition~\ref{def: Nash discriminant variety}). 

\begin{theorem*}[Proposition~\ref{prop: discriminant_two_players}, Theorem~\ref{thm: codimension real part discriminant of E is 1}]
For $\bd = (d, d)$ games, $\Delta(\bd)$ has two irreducible components of codimension two and degree $\binom{d}{2}$. Otherwise, the real part of the Nash discriminant variety $\Delta(\bd)$ has codimension one.
\end{theorem*}
We present specific computations for this variety for $(2,2,2)$ and $(2,2,3)$ games in Example~\ref{ex: computation degree hypersurface component 222} and Example~\ref{ex: computation degree hypersurface component 223}. 
Section~\ref{subsec: 2x2x2 game} is dedicated to the study of the Nash discriminant variety of $(2, 2, 2)$ games.
\begin{theorem*}[Theorem~\ref{thm: degree Nash discriminant variety 2x2x2}, Proposition~\ref{prop:curve_component}]
The Nash discriminant variety $\Delta(\bd)$ for $\bd = (2,2,2)$ is an irreducible hypersurface of degree~6 in $\mathbb P^{11}$. A nonsingular point of $\Delta(\bd)$ corresponds to a game $X$ whose Nash equilibrium scheme $\kz_X$ is a nonreduced point of multiplicity two. A game $X$ lies in the singular locus of $\Delta(\bd)$ if and only if its Nash equilibrium scheme $\kz_X$ contains a line, a nonsingular conic, or a nonsingular cubic. 
\end{theorem*}

We further compute all the irreducible components of the singular strata of this Nash discriminant variety $\Delta(2,2,2)$. The computations are done with \verb|Macaulay2| \cite{grayson1997macaulay2}, and the details are explained on the Mathrepo repository \cite{mathrepo} with illustrative images. A summary of these results can also be found in Table~\ref{tab: description varieties singular strata} and Table~\ref{tab: singular strata}. Additionally, we analyze the Nash discriminant variety for these games in the \emph{boundary format} (Section~\ref{subsec: Nash discriminant boundary format}). In this case, the Nash discriminant variety can be written as the union of an irreducible hypersurface (with a known degree) and a component of codimension bigger than or equal to two (Theorem~\ref{thm: discriminant of E for boundary format}).

\section{Games, tensors, and Nash}\label{sec: preliminaries}\label{subsec: Nash}

\subsection{General notations}\label{sec: notations}
We start by setting up the main notations used throughout the paper.
For a positive integer $n$, we denote $\{1,\ldots,n\}$ by $[n]$.
Let $\bd\coloneqq(d_1,\ldots,d_n)\in\Z_{\mge[2]}^n$.
Unless otherwise stated, we assume that $d_1\le\cdots\le d_n$. For every $i\in[n]$, let $V_i\coloneqq\C^{d_i}$ with the standard basis $\{e_1^{(i)},\ldots,e_{d_i}^{(i)}\}$.
We denote by $\bone$ a tuple that consists of ones and by $\bone_i$ a tuple that has a single zero in the $i$th element and one everywhere else. For the sake of simplicity, we use the same notations for the corresponding column vectors. For each $i \in [n]$, let  $(\pi_1^{(i)},\ldots, \pi_{d_i}^{(i)})$ be coordinates on $V_i$, which form the dual basis for the dual space $V_i^*$ of $V_i$ with respect to $\{e_1^{(i)},\ldots,e_{d_i}^{(i)}\}$. We denote by $\pi^{(i)}$ the column vector $(\pi_1^{(i)},\ldots, \pi_{d_i}^{(i)})^\mT$, and let $\bpi\coloneqq(\pi^{(1)},\ldots,\pi^{(n)})\in \prod_{i=1}^n V_i$.

Define $I \coloneqq \prod_{i=1}^n[d_i]$. For each $i \in [n]$, set $I_{-i} \coloneqq \prod_{j \neq i} [d_j]$. If $\bj\coloneqq (j_1, \ldots, j_n) \in I$, then we write $\bj_{-i}$ for the element of $I_{-i}$ obtained from $\bj$ by removing its $i$th element:  
\[
\bj_{-i} \coloneqq (j_1, \ldots, j_{i-1}, j_{i+1}, \ldots, j_n)\in I_{-i}, 
\]
and $(k,\bj_{-i})$ denotes the element of $I$ obtained from $\bj_{-i} \in I_{-i}$ by inserting $k$ in its $i$th position: 
\[
(k,\bj_{-i}) \coloneqq (j_1, \ldots, j_{i-1},k, j_{i+1}, \ldots, j_{n}) \in I. 
\]
Furthermore, we write
\[
\bpi_{-i}\coloneqq(\pi^{(1)},\ldots,\pi^{(i-1)},\pi^{(i+1)},\ldots, \pi^{(n)})\in \prod_{j\neq i}V_j\,.
\]
If $\bj= (j_1, \ldots, j_n) \in I$, then we set $\pi_{\bj_{-i}} \coloneqq \prod_{k \neq i} \pi_{j_k}^{(k)}$. 

Let $V\coloneqq\bigotimes_{i=1}^n V_i$. 
An element $T\in V$ can be described in coordinates as a tensor $T=(t_{\bj})_{\bj \in I}$ of format $\bd$. We denote by $T(\bpi_{-i})$ the vector-valued function on $\bigoplus_{j \neq i} V_j$ whose $k$th component function is 
\[
\sum_{\bj_{-i} \in I_{-i}} t_{(k,\bj_{-i})} \, \pi_{\bj_{-i}}\,. 
\]
The function $T(\bpi_{-i})$ corresponds to the operation of {\em tensor contraction} $T\cdot \bigotimes_{k \neq i}\pi^{(k)}$ between the tensor $T$ and the rank-one tensor $\bigotimes_{k \neq i} \pi^{(k)}\in \bigotimes_{k\neq i}V_k$.

For each $i \in [n]$, let $\PP^{d_i-1}\coloneqq \PP(V_i)$ be the projective space of one-dimensional subspaces of $V_i$ and let
\[
R_i \coloneqq \C[\pi_1^{(i)},  \ldots, \pi_{d_i}^{(i)}] =  \bigoplus_{k \ge 0} \mathrm{Sym}^k(V_i^*)
\] 
be its homogeneous coordinate ring. The multihomogeneous coordinate ring of $\PP^\bd\coloneqq\prod_{i=1}^n\PP^{d_i-1}$ is the multigraded polynomial ring $R \coloneqq R_1 \otimes_\C \cdots \otimes_\C R_n$ with multigrading defined by $\deg (\pi_j^{(i)})=\bone-\bone_i$.
For each $i \in [n]$, let $\mathrm{pr}_i\colon\PP^\bd \to \PP^{d_i-1}$ be the projection from $\PP^\bd$ to its $i$th factor. If $A(\PP^\bd)$ denotes the Chow ring of $\PP^\bd$, then 
\begin{equation}\label{eq: coordinate ring P^dd}
A(\PP^\bd) = \frac{\Z[h_1,\ldots, h_n]}{\langle h_1^{d_1},\ldots, h_n^{d_n}\rangle}\,,  
\end{equation}
where $h_i$ denotes the pullback of the hyperplane class on the $i$th factor $\PP^{d_i-1}$ of $\PP^\bd$ via the projection map $\mathrm{pr}_i$. Throughout the paper, we denote by $\ko$ the structure sheaf $\ko_{\PP^\bd}$ on $\PP^\bd$ and by $\ko(\balpha)$ the line bundle $\ko_{\PP^\bd}(\balpha)$ on $\PP^\bd$ for each $\balpha\in\Z^n$. We refer to \cite[Chapter 8]{fulton1998intersection} for more details.

\subsection{Preliminaries}\label{ssec:prelim}
In this section, we define an $n$-player game in normal form and the notion of Nash equilibria. We show how globally generated vector bundles \eqref{eq: Nash vector bundle} over the product of projective spaces $\PP^\bd$ provide an elegant description of the totally mixed Nash equilibria of an $n$-player game. This allows us to give an alternative proof of \cite[Theorem 3.3]{mckelvey1997maximal} for the maximal number of totally mixed Nash equilibria of generic games, see Theorem~\ref{thm: number tmNe generic game}. Any Nash equilibrium gives rise to a totally mixed Nash equilibrium of a smaller game obtained by eliminating unused strategies. Thus, it is still interesting to consider totally mixed Nash equilibria, as, for example, in the generic case, the maximal number of them gives a lower bound for the number of Nash equilibria of the smaller game. The completely (totally) mixed games, i.e., the games with only totally mixed Nash equilibria, are also of interest to game theorists (e.g., \cite{bubelis1979equilibria}, \cite{chin1974structure}).

\begin{definition}\label{def: game}
For a given positive integer $n \ge 2$, let $[n]$ be the set of players. If $\bd = (d_1,\ldots, d_n)$ with $d_i \ge 2$, then we interpret $[d_i]$ as the set of strategies for player $i$, and let $I = \prod_{i=1}^n [d_i]$. For the selected strategies $\bj = (j_1, \ldots, j_n) \in I$ of the players, we write $x_{\bj}^{(i)}$ for the payoff for player~$i$, and let $X^{(i)}$ be the {\em payoff tensor} $(x_{\bj}^{(i)})_{\bj\in I}$. If $X$ denotes the collection $(X^{(1)}, \ldots, X^{(n)})\in V^{\oplus n}$ of payoff tensors, then the triplet $([n], I, X)$ is called an {\em $n$-player game in normal form of format $\bd$.} For the sake of simplicity, we use $X$ to represent the game $([n],I,X)$.
\end{definition}
\begin{remark}
This paper addresses games with an arbitrary number of players. However, most examples concern three-player games. Thus, we discuss how we represent a tensor of order $3$.

If $X = (X^{(1)}, X^{(2)}, X^{(3)})$ is a three-player game of format $(d_1, d_2, d_3)$, then each payoff tensor $X^{(i)}=(x_\bj^{(i)})_{\bj\in I}$ is a three-dimensional array with real entries. 
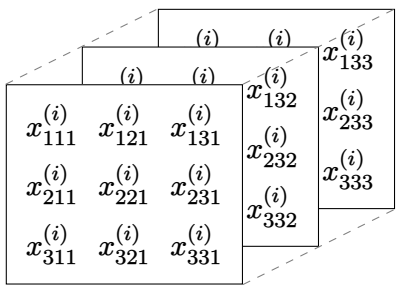
\begin{figure}[ht]
\begin{tikzpicture}
\def\xs{1} %shift in x direction
\def\ys{0.5} %shift in y direction
\foreach \x in {3,2,1}
{
\matrix [draw, % for the rectangle border
         fill=white, % so that it is not transparent
         ampersand replacement=\&] %see explanation
(mm\x)%give the matrix a name
at(\x * \xs, \x * \ys) %shift the matrix
{
    \node {$x_{11\x}^{(i)}$}; \& \node {$x_{12\x}^{(i)}$}; \& \node {$x_{13\x}^{(i)}$};\\
    \node {$x_{21\x}^{(i)}$}; \& \node {$x_{22\x}^{(i)}$};  \& \node {$x_{23\x}^{(i)}$};\\
    \node {$x_{31\x}^{(i)}$}; \& \node {$x_{32\x}^{(i)}$};  \& \node {$x_{33\x}^{(i)}$};\\
};
}
\draw [dashed,gray](mm1.north west) -- (mm3.north west);
\draw [dashed,gray](mm1.north east) -- (mm3.north east);
\draw [dashed,gray](mm1.south east) -- (mm3.south east);
\end{tikzpicture}
\caption{$3$-player game of format $(3,3,3)$}
\end{figure}
We express each payoff tensor~$X^{(i)}$ as the $d_1\times d_2d_3$ matrix obtained by concatenating the horizontal slices of $X^{(i)}$. For example, when $(d_1,d_2,d_3)=(3,3,3)$, the $3 \times 9$ matrix
\[
\begin{bmatrix}[ccc:ccc:ccc]
x_{111}^{(i)} & x_{121}^{(i)} & x_{131}^{(i)} & x_{112}^{(i)} & x_{122}^{(i)} & x_{132}^{(i)} & x_{113}^{(i)} & x_{123}^{(i)} & x_{133}^{(i)}\\[2pt]
x_{211}^{(i)} & x_{221}^{(i)} & x_{231}^{(i)} & x_{212}^{(i)} & x_{222}^{(i)} & x_{232}^{(i)} & x_{213}^{(i)} & x_{223}^{(i)} & x_{233}^{(i)}\\[2pt]
x_{311}^{(i)} & x_{321}^{(i)} & x_{331}^{(i)} & x_{312}^{(i)} & x_{322}^{(i)} & x_{332}^{(i)} & x_{313}^{(i)} & x_{323}^{(i)} & x_{333}^{(i)}
\end{bmatrix}\,
\]
represents $X^{(i)}$. 
\end{remark}

For all $i\in[n]$, let $\Delta_{d_i-1}$ be the $(d_i-1)$-dimensional probability simplex:
\[
\Delta_i \coloneqq \left\{\pi^{(i)}\in \R^{d_i}\,  \left|\, \mbox{$\sum_{j=1}^{d_i}\pi_{j}^{(i)}=1$,  $\pi_{j}^{(i)} \ge 0$ for all $j$} \right. \right\}\,.
\]
Given $\pi^{(i)}\in \Delta_{d_i-1}$, we interpret each component $\pi_{j}^{(i)}$ as the probability (or mixed strategy) that player $i$ unilaterally selects the pure strategy $j \in [d_i]$. The $n$ players choose a joint probability distribution $\pi^{(1)}\otimes\cdots\otimes\pi^{(n)}=(\pi_{j_1}^{(1)}\cdots\pi_{j_n}^{(n)})_{\bj\in I}$. The {\em expected payoff} for player~$i$ is the standard inner product of the tensors $\pi^{(1)}\otimes\cdots\otimes\pi^{(n)}$ and $X^{(i)}$, namely, 
\begin{equation}\label{eq: expected payoff}
\pi^{(1)}\otimes\cdots\otimes\pi^{(n)} \cdot X^{(i)} = \sum_{\bj\in I}\pi_{j_1}^{(1)}\cdots\pi_{j_n}^{(n)}x_{\bj}^{(i)}\,.
\end{equation}

A {\em Nash equilibrium} $\bpi\in\prod_{i=1}^n \Delta_{d_i-1}$ for a game $X$ is obtained if no player can increase their expected payoff \eqref {eq: expected payoff} by changing their mixed strategy $\pi^{(i)}$ while the other players keep their mixed strategies fixed. The semialgebraic set defining Nash equilibria was studied in, e.g., \cite{mckelvey1996computation}, \cite[Chapter 6]{sturmfels2002solving}.
Our interest is in the case where each player's strategy is \emph{totally mixed}, i.e., strictly positive. We give an algebraic definition of a totally mixed Nash equilibrium.

\begin{definition}\label{def: tmNe}
Let $X=(X^{(1)},\ldots,X^{(n)})$ be an $n$-player game. For all $i\in[n]$, we write $\Delta_{d_i-1}^\circ$ for the interior of $\Delta_{d_i-1}$.
A tuple $\bpi=(\pi^{(1)},\ldots,\pi^{(n)})\in \prod_{i=1}^n \Delta_{d_i-1}^\circ$ is a {\em totally mixed (completely mixed) Nash equilibrium} of $X$ if, for all $i\in[n]$, the contraction $X^{(i)}(\bpi_{-i})$ is a scalar multiple of~$\bone$.
\end{definition}

\begin{remark}\label{rmk: equations tmNe}
With the same notation used in Definition~\ref{def: tmNe}, the linear dependence of $X^{(i)}(\bpi_{-i})$ and $\bone$ is equivalent to the vanishing of the wedge product $X^{(i)}(\bpi_{-i}) \wedge \bone$, or the vanishing of the $2 \times 2$ minors of the $d_i \times 2$ matrix
\[
[X^{(i)}(\bpi_{-i})\mid\bone]\,.
\]
Since $X^{(i)}(\bpi_{-i})$ is the column vector of multihomogeneous polynomials $f_1^{(i)}, \ldots, f_{d_i}^{(i)}$, where 
\[
f_k^{(i)}\coloneqq\sum_{\bj_{-i} \in I_{-i}} x_{(k,\bj_{-i})}^{(i)}\,\pi_{\bj_{-i}}, 
\]
the $2 \times 2$ minors of $[X^{(i)}(\bpi_{-i})\mid\bone]$ are 
\[
\varDelta f_{k,\ell}^{(i)} \coloneqq f_k^{(i)}-f_\ell^{(i)}=\sum_{\bj_{-i} \in I_{-i}} (x_{(k,\bj_{-i})}^{(i)}-x_{(\ell,\bj_{-i})}^{(i)}) \, \pi_{\bj_{-i}}
\]
for all $1 \le k < \ell \le d_i$. As $\varDelta f_{k,\ell}^{(i)} = \varDelta f_{1,\ell}^{(i)}-\varDelta f_{1,k}^{(i)}$, the necessary and sufficient conditions for an element $\bpi=(\pi^{(1)},\ldots,\pi^{(n)})$ of $\prod_{i=1}^n \Delta_{d_i-1}^\circ$ to be a totally mixed Nash equilibrium of $X$ are expressed as the system of the following polynomial equalities and polynomial inequalities:     
\begin{equation}\label{eq: semi_algebraic_set}
\begin{cases}
    \varDelta f_{1,k}^{(i)} = 0 & \text{for all $i \in [n]$ and $k\in[d_i]\setminus\{1\}$,}\\
    \pi_k^{(i)} > 0 & \text{for all $i \in [n]$ and $k\in[d_i]$,}\\
    \sum_{k=1}^{d_i} \pi_k^{(i)} = 1 & \text{for all $i \in [n]$.}
\end{cases} 
\end{equation}
In other words, a totally mixed Nash equilibrium of $X$ is a point of the semialgebraic set defined by \eqref{eq: semi_algebraic_set}.    
\end{remark}

\begin{example}\label{ex: 3x3x3}
Let $X = (X^{(1)}, X^{(2)}, X^{(3)})$ be the three-player game of format $\bd=(3,3,3)$ with
\begin{align*}
X^{(1)} &= \begin{bmatrix}[rrr:rrr:rrr]
-20 & -4 & -12 & -16 & 8 & 4 & 12 & 8 & -4\\
8 & 20 & 20 & -20 & -20 & 16 & 12 & -4 & -23\\
-76 & 4 & -4 & 8 & 4 & -4 & 12 & -4 & 8 
\end{bmatrix},\\
X^{(2)} &= \begin{bmatrix}[rrr:rrr:rrr]
6 & -2 & 4 & -8 & 4 & -6 & -10 & -2 & 0 \\
-2 & 0 & -10 & -2 & 6 & 4 & 10 & 6 & -2 \\
2 & -10 & 4 & 4 & -8 & -2 & -2 & -5 & 2 
\end{bmatrix},\\
X^{(3)} &= \begin{bmatrix}[rrr:rrr:rrr]
-8 & -4 & 6 & -8 & -4 & 0 & -10 & 10 & 10 \\
-6 & -10 & 6 & 0 & -10 & 2 & 2 & 4 & -10 \\
4 & 10 & -2 & -4 & 0 & 14 & 0 & -6 & 3 
\end{bmatrix}\,.
\end{align*}

A triple $\bpi = (\pi^{(1)},\pi^{(2)}, \pi^{(3)}) \in \Delta_2^\circ \times \Delta_2^\circ  \times \Delta_2^\circ$ is a totally mixed Nash equilibrium of $X$ if and only if it satisfies $\varDelta f_{1,k}^{(i)} = 0$ for all $i \in [3]$ and $k \in [3] \setminus\{1\}$, where  
\begin{align*}
\begin{smallmatrix*}[l]
f_1^{(1)} & = & -20\,\pi_1^{(2)}\pi_1^{(3)}-4\,\pi_2^{(2)}\pi_1^{(3)}-12\,\pi_3^{(2)}\pi_1^{(3)}-16\,\pi_1^{(2)}\pi_2^{(3)}+8\,\pi_2^{(2)}\pi_2^{(3)}+4\,\pi_3^{(2)}\pi_2^{(3)}+12\,\pi_1^{(2)}\pi_3^{(3)}+8\,\pi_2^{(2)}\pi_3^{(3)}-4\,\pi_3^{(2)}\pi_3^{(3)}\\
f_2^{(1)} & = & 8\,\pi_1^{(2)}\pi_1^{(3)}+20\,\pi_2^{(2)}\pi_1^{(3)}+20\,\pi_3^{(2)}\pi_1^{(3)}-20\,\pi_1^{(2)}\pi_2^{(3)}-20\,\pi_2^{(2)}\pi_2^{(3)}+16\,\pi_3^{(2)}\pi_2^{(3)}+12\,\pi_1^{(2)}\pi_3^{(3)}-4\,\pi_2^{(2)}\pi_3^{(3)}-23\,\pi_3^{(2)}\pi_3^{(3)}\\
f_3^{(1)} & = & -76\,\pi_1^{(2)}\pi_1^{(3)}+4\,\pi_2^{(2)}\pi_1^{(3)}-4\,\pi_3^{(2)}\pi_1^{(3)}+8\,\pi_1^{(2)}\pi_2^{(3)}+4\,\pi_2^{(2)}\pi_2^{(3)}-4\,\pi_3^{(2)}\pi_2^{(3)}+12\,\pi_1^{(2)}\pi_3^{(3)}-4\,\pi_2^{(2)}\pi_3^{(3)}+8\,\pi_3^{(2)}\pi_3^{(3)}\vspace{5pt}\\
f_1^{(2)} & = & 6\,\pi_1^{(1)}\pi_1^{(3)}-2\,\pi_2^{(1)}\pi_1^{(3)}+2\,\pi_3^{(1)}\pi_1^{(3)}-8\,\pi_1^{(1)}\pi_2^{(3)}-2\,\pi_2^{(1)}\pi_2^{(3)}+4\,\pi_3^{(1)}\pi_2^{(3)}-10\,\pi_1^{(1)}\pi_3^{(3)}+10\,\pi_2^{(1)}\pi_3^{(3)}-2\,\pi_3^{(1)}\pi_3^{(3)}\\
f_2^{(2)} & = & -2\,\pi_1^{(1)}\pi_1^{(3)}-10\,\pi_3^{(1)}\pi_1^{(3)}+4\,\pi_1^{(1)}\pi_2^{(3)}+6\,\pi_2^{(1)}\pi_2^{(3)}-8\,\pi_3^{(1)}\pi_2^{(3)}-2\,\pi_1^{(1)}\pi_3^{(3)}+6\,\pi_2^{(1)}\pi_3^{(3)}-5\,\pi_3^{(1)}\pi_3^{(3)}\\
f_3^{(2)} & = & 4\,\pi_1^{(1)}\pi_1^{(3)}-10\,\pi_2^{(1)}\pi_1^{(3)}+4\,\pi_3^{(1)}\pi_1^{(3)}-6\,\pi_1^{(1)}\pi_2^{(3)}+4\,\pi_2^{(1)}\pi_2^{(3)}-2\,\pi_3^{(1)}\pi_2^{(3)}-2\,\pi_2^{(1)}\pi_3^{(3)}+2\,\pi_3^{(1)}\pi_3^{(3)}\vspace{5pt}\\
f_1^{(3)} & = & -8\,\pi_1^{(1)}\pi_1^{(2)}-6\,\pi_2^{(1)}\pi_1^{(2)}+4\,\pi_3^{(1)}\pi_1^{(2)}-4\,\pi_1^{(1)}\pi_2^{(2)}-10\,\pi_2^{(1)}\pi_2^{(2)}+10\,\pi_3^{(1)}\pi_2^{(2)}+6\,\pi_1^{(1)}\pi_3^{(2)}+6\,\pi_2^{(1)}\pi_3^{(2)}-2\,\pi_3^{(1)}\pi_3^{(2)}\\
f_2^{(3)} & = & -8\,\pi_1^{(1)}\pi_1^{(2)}-4\,\pi_3^{(1)}\pi_1^{(2)}-4\,\pi_1^{(1)}\pi_2^{(2)}-10\,\pi_2^{(1)}\pi_2^{(2)}+2\,\pi_2^{(1)}\pi_3^{(2)}+14\,\pi_3^{(1)}\pi_3^{(2)}\\
f_3^{(3)} & = & -10\,\pi_1^{(1)}\pi_1^{(2)}+2\,\pi_2^{(1)}\pi_1^{(2)}+10\,\pi_1^{(1)}\pi_2^{(2)}+4\,\pi_2^{(1)}\pi_2^{(2)}-6\,\pi_3^{(1)}\pi_2^{(2)}+10\,\pi_1^{(1)}\pi_3^{(2)}-10\,\pi_2^{(1)}\pi_3^{(2)}+3\,\pi_3^{(1)}\pi_3^{(2)}\,.
\end{smallmatrix*}
\end{align*}
We verified with the \verb|Macaulay2| package \verb|RealRoots|~\cite{lopez2024real} that the system \eqref{eq: semi_algebraic_set} admits four real solutions. More precisely, our code provides a list of intervals of  $(a_k^{(i)},b_k^{(i)}]$ of $\R$ containing a unique root $\pi_k^{(i)}$ such that $b_k^{(i)}-a_k^{(i)}$ is less than the tolerance determined by the user. One can directly check that
\[
\bpi=\left(\left(\frac{1}{3},\frac{1}{3},\frac{1}{3}\right),\left(\frac{1}{4},\frac{1}{4},\frac{1}{2}\right),\left(\frac{1}{5},\frac{2}{5},\frac{2}{5}\right)\right)
\]
is one of the solutions of \eqref{eq: semi_algebraic_set}.\hfill$\diamondsuit$
\end{example}

If $n \ge 3$, then the polynomials $\varDelta f_{1,k}^{(i)}$ in \eqref{eq: semi_algebraic_set} are multihomogeneous in the $n$ vectors of variables $\pi^{(i)}=(\pi_1^{(i)},\ldots,\pi_{d_i}^{(i)})$ for all $i\in[n]$. Thus, we can regard a totally mixed Nash equilibrium as a point of a multiprojective space $\PP^\bd$. This leads to the following definition.

\begin{definition}\label{def: Nash equilibria scheme}
Let $X=(X^{(1)},\ldots,X^{(n)})\in V^{\oplus n}$. For all $i \in [n]$, let $J_i$ be the multihomogenous ideal of $R$ generated by $\varDelta f_{1,k}^{(i)}$ in \eqref{eq: semi_algebraic_set}.
The {\em Nash equilibrium scheme} of $X$ is the subscheme $\kz_X$ of $\PP^\bd$ defined by the multihomogeneous ideal~$J\coloneqq \sum_{i=1}^n J_i$.
\end{definition}

It follows immediately from Definitions~\ref{def: tmNe} and~\ref{def: Nash equilibria scheme} that $\bpi=(\pi^{(1)},\ldots,\pi^{(n)}) \in \prod_{i=1}^{d_i} \Delta_{d_i-1}^\circ$ is a totally mixed Nash equilibrium of the game $X=(X^{(1)},\ldots,X^{(n)})$ if the corresponding point~$[\bpi]\coloneqq([\pi^{(1)}],\ldots,[\pi^{(n)}])$ of $\PP^\bd$ lies in the Nash equilibrium scheme $\kz_X$ of $X$. 

\subsection{Vector bundles}\label{subsec: vector bundles}
If $X = (X^{(1)},\ldots, X^{(n)}) \in V^{\oplus n}$, then the $2 \times 2$ minors $\varDelta f_{k,\ell}^{(i)}$ of the matrix $[X^{(i)}(\pi_{-i})\mid\bone]$ are multihomogeneous of multidegree $\bone_i$. Thus, we view them as elements of the cohomology group $H^0(\PP^\bd, \ko(\bone_i))$ of the line bundle $\ko(\bone_i)$ on $\PP^\bd$ and $r_X^{(i)} \coloneqq (\varDelta f_{1,2}^{(i)}, \ldots, \varDelta f_{1,d_i}^{(i)})$ as an element of the cohomology group $H^0(\PP^\bd, \ko(\bone_i)^{\oplus(d_i-1)})$ of the direct sum $\ko(\bone_i)^{\oplus(d_i-1)}$ of $d_i-1$ copies of~$\ko(\bone_i)$. Therefore, $r_X \coloneqq (r_X^{(1)}, \ldots, r_X^{(n)})$ can be thought of as a global section of the following vector bundle of rank $D \coloneqq \sum_{i=1}^n (d_i-1)$ on $\PP^\bd$: 
\begin{equation}\label{eq: Nash vector bundle}
    E \coloneqq \bigoplus_{i=1}^n \ko(\bone_i)^{\oplus(d_i-1)}\,.  
\end{equation}
Consequently, the Nash equilibrium scheme of $X$ is obtained as the zero scheme of a global section of~$E$. 

The vector bundle $E$ is globally generated, as so are its direct summands. Thus, the degree $\int_{\PP^{\bd}} c_D(E)$ of the top Chern class $c_D(E)$ of $E$ is nonnegative \cite[Proposition 10]{geertsen2002degeneracy}. Moreover, the global sections $f$ of $E$ whose zero schemes $\kz(f)$ are nonsingular of the expected codimension (meaning that it is of codimension $D+1$ if $\int_{\PP^{\bd}}c_D(E) = 0$; it is of codimension $D$ otherwise) form an open subset $U(E)$ of $H^0(\PP^{\bd},E)$ \cite[Lemma 2.5]{ein1982some}. 
\begin{proposition}\label{prop: open subset generic global sections E}
    The map $\Phi\colon V^{\oplus n} \to H^0(\PP^{\bd},E)$ defined by $\Phi(X) = r_X$ is an onto linear transformation. In particular, for each $X \in V^{\oplus n}$ there exists a global section $f$ of the vector bundle $E$ of rank $D$ such that $\kz_X = \kz(f)$. Furthermore, the subset of elements of $V^{\oplus n}$ whose Nash equilibrium scheme is nonsingular of the expected codimension forms an open subset. 
\end{proposition}
\begin{proof}
    For each $i \in \{1,\ldots, n\}$, let $V_i$ denote the $i$th direct summand of $V^{\oplus n}$ and let $\Phi_i\colon V_i \to H^0(\PP^{\bd},\ko(\bone_i)^{\oplus(d_i-1)})$ be the map defined by $\Phi_i(X^{(i)})=r_X^{(i)}=(\varDelta f_{1,2}^{(i)}, \ldots, \varDelta f_{1,d_i}^{(i)})$. The linearity of $\Phi_i$ follows from that of the coefficient $x_{(1,\bj_{-i})}^{(i)}-x_{(\ell,\bj_{-i})}^{(i)}$ of $\pi_{\bj_{-i}}$ in $\varDelta f_{1,\ell}^{(i)}$. These coefficients are linearly independent, which implies the surjectivity of $\Phi_i$. Note that $\Phi$ is the direct sum of $\Phi_1, \ldots, \Phi_n$. Thus, it is an onto linear transformation. Therefore, for each $X^{(i)} \in V_i$, there exists a global section $f^{(i)}$ of $\ko(\bone_i)^{\oplus(d_i-1)}$ such that $r_X^{(i)} = f^{(i)}$. Since $r_X^{(i)}$ generated the ideal $J_i$ and since $J = \sum_{i=1}^n J_i$ defines the Nash equilibrium scheme $\kz_X$ of $X$, the zero scheme of $f = (f^{(1)}, \ldots, f^{(n)}) \in H^0(\PP^{\bd},E)$ coincides with $\kz_X$. 

    The subset of elements of $V^{\oplus n}$ whose Nash equilibrium schemes are nonsingular of the expected codimension is the inverse image of $U(E)$ under $\Phi$. Therefore, it is open in~$V^{\oplus n}$. 
\end{proof}

We say that a game $X\in V^{\oplus n}$ is {\em generic} if $X$ belongs to the open subset of $V^{\oplus n}$ defined in Proposition \ref{prop: open subset generic global sections E}.
As formalized in the upcoming Theorem \ref{thm: number tmNe generic game}, there is a dichotomy in the previous notion of genericity, which depends only on the format $\bd$. We summarize this in Table~\ref{tab:boundary_formats}.
\begingroup
\renewcommand{\arraystretch}{1.4}
\begin{table}[ht]
\centering
\resizebox{\textwidth}{!}{
\begin{tabular}{l||c|c}
 & Within boundary format  & Beyond boundary format \\ \hhline{===}
 Formulation & $d_n - 1 \leq \sum_{i=1}^{n-1} (d_i - 1)$ & $d_n - 1 > \sum_{i=1}^{n-1} (d_i - 1)$ \\ \hline
 Generic $X$ & $\dim\kz_X=0$ and $\kz_X$ reduced of degree $c(\bd)$ & $\kz_X=\emptyset$ \\ \hline
 Nongeneric $X$ & Nash discriminant variety, Definition \ref{def: Nash discriminant variety} & Nash resultant variety, Definition \ref{def: Nash resultant variety}
\end{tabular}
}
\cprotect\caption{The dictionary for format separation, the generic behavior of the Nash equilibrium scheme $\kz_X$, and nongeneric behavior of games, depending on the format $\bd=(d_1,\ldots,d_n)\in\Z_{\mge[2]}^n$ with $d_1\le\cdots\le d_n$.}\label{tab:boundary_formats}
\end{table}
\endgroup
We use the above description and Proposition \ref{prop: open subset generic global sections E} to give an alternative formulation and proof to the well-known result on the maximal number of totally mixed Nash equilibria of a generic game in terms of the degree of the top Chern class of the vector bundle $E$. 
Indeed, Mckelvey and McLennan \cite{mckelvey1997maximal} compute this maximal number as the mixed volume of the Newton polytopes associated to the generators of the multihomogeneous ideal $J$ defining the Nash equilibrium scheme $\kz_X$ (Definition~\ref{def: Nash equilibria scheme}). Equivalently, this number can be computed combinatorially via so-called block derangements \cite[Theorem 6.8]{sturmfels2002solving}.

\begin{theorem}\label{thm: number tmNe generic game}
Let $n\in\Z_{\mge[2]}$, let $\bd=(d_1,\ldots,d_n)\in\Z_{\mge[2]}^n$ with $d_1\le\cdots\le d_n$, and let $c(\bd)$ be the coefficient of the monomial $\prod_{i=1}^n h_i^{d_i-1}$ in $\prod_{i=1}^n \hat{h}_i^{d_i-1}$ with $\hat{h}_i\coloneqq \sum_{j\neq i}h_j$.
If $X \in V^{\oplus n}$ is generic, then the following three conditions are equivalent: 
\begin{itemize}
    \item[(1)] $\kz_X = \emptyset$. 
    \item[(2)] $c(\bd) = 0$.
    \item[(3)] $d_n-1>\sum_{i=1}^{n-1} (d_i-1)$. 
\end{itemize}
Furthermore, if $\kz_X \neq\emptyset$, then $\kz_X$ has dimension $0$, and its degree is equal to $c(\bd)$. In particular, the number of distinct totally mixed Nash equilibria of $X$ is bounded above by $c(\bd)$. 
\end{theorem}
\begin{proof}
The second and third assertions are immediate consequences of (1), (2), and (3). Thus, we focus on proving the equivalence of (1), (2), and (3). 

If $c_t(\ko(\bone_i))$ denotes the Chern polynomial of $\ko(\bone_i)$, then $c_t(\ko(\bone_i)) = 1+\hat{h}_it$, and hence $c_t(\ko(\bone_i)^{\oplus(d_i-1)}) = (1+\hat{h}_it)^{d_i-1}$. So, we have $c_t(E) = \prod_{i=1}^n (1+\hat{h}_it)^{d_i-1}$. Since the $D$th Chern class of $E$ is the coefficient of $t^D$ in $c_t(E)$, the degree $\int_{\PP^\bd} c_D(E)$ of $c_D(E)$ is equal to the coefficient~$c(\bd)$ of $\prod_{i=1}^n h_i^{d_i-1}$ in $\prod_{i=1}^n \hat{h}_i$. This establishes the equivalence between (1) and (2); it remains to show the equivalence between (2) and (3). 

For each $i \in [n]$, the Newton polytope corresponding to the torus-invariant Cartier divisor associated with $\ko(\bone_i)$ is 
\[
\Delta^{(i)} \coloneqq \Delta_{d_1-1} \times \cdots \times \Delta_{d_{i-1}-1} \times \{0\} \times  \Delta_{d_{i+1}-1} \times \cdots \times \Delta_{d_n-1}\,,
\]
where $\Delta_{d_j-1}$ is the $(d_j-1)$-dimensional probability simplex. Define
\[
\Delta[\bd] \coloneqq \{
\underbrace{\Delta^{(1)}, \ldots, \Delta^{(1)}}_{d_1-1}, 
\underbrace{\Delta^{(2)}, \ldots, \Delta^{(2)}}_{d_2-1}, \ldots, 
\underbrace{\Delta^{(n)}, \ldots, \Delta^{(n)}}_{d_n-1}
\}\,.
\]
The mixed volume $\mathrm{MV}(\Delta[\bd])$ of $\Delta[\bd]$ is the coefficient of the monomial $\prod_{i=1}^n \prod_{j=1}^{d_i-1} \lambda_i^{(j)}$ in the volume of $\sum_{i=1}^n \sum_{j=1}^{d_i-1} \lambda_i^{(j)} \Delta^{(i)}$, and it coincides with $c(\bd)$ by the BKK theorem \cite{bernstein1975number, kouchnirenko1976polyedres}. Thus, to complete the proof, we show that~$\mathrm{MV}(\Delta[\bd])=0$ if and only if $d_n-1>\sum_{i=1}^{n-1} (d_j-1)$. 

Note that $\mathrm{MV}(\Delta[\bd]) = 0$ if and only if $\Delta[\bd]$ is ``dependent,'' i.e., there exist $B_1, B_2, \ldots, B_\delta \in \Delta[\bd]$ such that $\sum_{j=1}^\delta B_j < \delta$. See, for example, \cite[Theorem~5.1.8]{schneider2014convex}. The definition of the polytopes $\Delta^{(1)},\ldots,\Delta^{(n)}$ implies that the latter happens precisely when the following inequality holds:  
\begin{equation}\label{eq: inequality dim sum}
\dim (\underbrace{\Delta^{(n)}+\Delta^{(n)} + \cdots + \Delta^{(n)}}_{d_n-1}) < d_n-1\,.
\end{equation}
Since the left-hand side of \eqref{eq: inequality dim sum} is $\sum_{i=1}^{n-1} (d_i-1)$, the last part of the statement follows. 
\end{proof}

\begin{remark}\label{rmk: combinatorial formulas Nash special formats}
If each player of an $n$-player game has $2$ strategies, then
\[
c(2,\ldots,2) = n!\sum_{i=2}^n\frac{(-1)^j}{j!}\,,
\]
which is the number of derangements of $[n]$ or the integer sequence \cite[\href{http://oeis.org/A000166}{A000166}]{oeis} (see also \cite[Theorem 3.3]{mckelvey1997maximal}, \cite[Eq. (3.2)]{vidunas2017counting}). The values of $c\,(2,\ldots,2)$ for $n \in \{2, \ldots, 10\}$ are 
\[
1,2,9,44,265,1854,14833,133496,1334961
\]
respectively. 

Another interesting case is when $n=3$. If each player has $d$ strategies, then 
\[
c(d,d,d)=\sum_{j=0}^{d-1}\binom{d-1}{j}^3.  
\]
This is also known as a {\em Franel number}, see also \cite[\href{http://oeis.org/A000172}{A000172}]{oeis}.
\end{remark}

\begin{example}\label{ex: 3x3x3 number complex solutions}
Let $X$ be the three-player game of Example~\ref{ex: 3x3x3} and let $J$ be the ideal of Definition~\ref{def: Nash equilibria scheme}. We symbolically verified in \verb|Macaulay2| that $J$ defines a reduced zero-dimensional scheme $\kz_X$ of degree~$10$ in $\PP^2\times\PP^2\times\PP^2$. The degree of $\kz_X$ coincides with $c(3,3,3)$ of Theorem~\ref{thm: number tmNe generic game}:
\[
c(3,3,3) = \binom{2}{0}^3+\binom{2}{1}^3+\binom{2}{2}^3 = 1 + 8 + 1 = 10\,.
\]
As was discussed in Example \ref{ex: 3x3x3}, four of the ten points in $\kz_X$ are real.
The remaining six nonreal solutions are grouped into three pairs of complex-conjugated solutions because~$J$ is generated by polynomials with real coefficients.\hfill$\diamondsuit$
\end{example}

A natural question is whether $c(\bd)$ can be uniquely described via a multivariate generating function. The next result goes in this direction and is reported in \cite[Eq. (1.5)]{vidunas2017counting}. Its proof employs MacMahon's Master Theorem \cite[Section 3, Chapter 2, 66]{MacMahon1915combinatory}.

\begin{theorem}\label{thm: generating function Nash}
The generating function of the coefficients $c(\bd)$ is given by
\[
\sum_{\bd\in\Z_{\mge[2]}^n}c(\bd)\,{\bx}^d=\frac{x_1\cdots x_n}{\sum_{i=0}^{n}(1-i)\,e_i({\bx})}\,,
\]
where ${\bx}=(x_1,\ldots,x_n)$, $\bx^\bd=x_1^{d_1}\cdots x_n^{d_n}$, and $e_i({\bx})$ is the $i$th elementary symmetric polynomial in the entries of ${\bx}$.
\end{theorem}

\begin{remark}
We compare the generating function of the coefficients $c(\bd)$ given in Theorem \ref{thm: generating function Nash} with other relevant generating functions.  We list three such examples below. 
\begin{enumerate}
    \item The degrees of hyperdeterminants of tensors of format $\bd$, see \cite[Theorem 2.4, Chapter 14]{GKZ}. 
    \item The number of singular vector tuples (see Definition \ref{def: singular tuples}) of a generic tensor of format $\bd$, see \cite[Proposition 1]{ekhad2016number}. 
    \item The degrees of Kalman varieties of tensors \cite[Theorem 2]{shahidi2021degrees}.
\end{enumerate}
In each case, the generating function admits a rational expression $F(\bx)/G(\bx)$ for some holomorphic functions $F$ and $G$. Thanks to \cite[Theorem 3.2]{raichev2008asymptotics}, it is possible to study the asymptotic behavior of the coefficients of the chosen generating function as long as the vector $\bd$ diverges along a certain direction. The most natural one is the asymptotic behavior along the main diagonal, namely when $d_1=\cdots=d_n=d$ and as $d\to \infty$. In the case of totally mixed Nash equilibria, Vidunas studied the asymptotic behavior of $c(d,\ldots,d)$ as $d\to\infty$. In particular \cite[Theorem 6.1]{vidunas2017counting} in our notation is expressed as follows:
\begin{equation}\label{eq: asymptotics Nash main diagonal}
    c\,(d,\ldots,d) = \frac{\sqrt{n}(n-1)^{nd-1}}{(2n(n-2)\pi d)^{\frac{n-1}{2}}}\left(1+O\left(\frac{1}{d}\right)\right)\,.
\end{equation}
It may be interesting to view \eqref{eq: asymptotics Nash main diagonal} alongside the asymptotics of the degrees of hyperdeterminants or the number of singular vector tuples of a hypercubic tensor. We refer to \cite[Remark 3.9]{ottaviani2021asymptotics} for more details.
\end{remark}

\subsection{Intermezzo: totally mixed Nash equilibria vs singular vector tuples}

We compare totally mixed Nash equilibria of an $n$-player game and singular vector tuples of tensors. In~2005, Lim introduced the concept of a singular vector tuple, which generalizes the concept of a singular vector pair of a rectangular matrix \cite{lim2005singular}.

\begin{definition}\label{def: singular tuples}
Let $n\in\Z_{\mge[2]}$ and let $\bd = (d_1,\ldots,d_n)\in\Z_{\mge[2]}^n$. For each $i \in [n]$, let $V_i$ denote a $d_i$-dimensional vector space over $\R$.  We denote by $V$ the tensor product of $V_1,\ldots,V_n$.
An $n$-tuple $\bv\coloneqq(v^{(1)},\ldots,v^{(n)}) \in \prod_{i=1}^n(V_i\setminus\{0\})$ is called a {\em singular vector tuple} of a tensor $T \in V$ if $\mathrm{rank}\,[T(\bv_{-i})\mid v^{(i)}]\le 1$ for each $i\in[n]$. 
\end{definition}

\noindent The same definition holds for complex tensors, but in this remark, we restrict ourselves to real ones. 

Definition~\ref{def: singular tuples} indicates that an $n$-tuple $(v^{(1)}, \ldots, v^{(n)})$ is a singular vector tuple of $T$ if and only if $(\mu_1 v^{(1)}, \ldots, \mu_n v^{(n)})$ is also a singular vector tuple of $T$ for any nonzero $\mu_1, \ldots, \mu_n \in \R$. Thus, we can unambiguously define a singular vector tuple of $T$ as the point $([v^{(1)}],\ldots,[v^{(n)}]) \in \PP^\bd$. The locus of singular vector tuples of $T$ is the closed subscheme in $\PP^\bd$ defined by all the $2 \times 2$ minors of the $d_i \times 2$ matrix $[T(\bv_{-i})\mid v^{(i)}]$ for every $i\in[n]$.

In their paper~\cite{FO}, Friedland and Ottaviani showed that the scheme of singular vector tuples of the tensor can be expressed as the zero scheme of a global section of a certain vector bundle on $\PP^\bd$ of rank $D=\dim\PP^\bd$. Furthermore, they derived a formula for the number of distinct singular vector tuples of a generic tensor by using the top Chern class of the vector bundle \cite[Theorem 1]{FO}. It is also worth mentioning that singular vector tuples of tensors are a special case of {\em singular vector tuples of hyperquiver representations} introduced and studied in \cite{muller2025multilinear}.

We show that both finding a totally mixed Nash equilibrium of an $n$-player game $X=(X^{(1)},\ldots, X^{(n)})\in V^{\oplus n}$ and computing a singular vector tuple of a tensor $T\in V$ can be formulated as optimization problems:

(1) To compute the totally mixed Nash equilibria of $X=(X^{(1)},\ldots, X^{(n)})\in V^{\oplus n}$, one maximizes the expected payoff of each player of $X$ subject to the law of total probability. Recalling Definition \ref{def: tmNe}, a vector of probability distributions $\bpi = (\pi^{(1)},\ldots,\pi^{(n)}) \in \prod_{i=1}^n \Delta_{d_i-1}^\circ$ is a totally mixed Nash equilibrium of $X$ when, for every $i \in [n]$, there exists a $\lambda_i \in \R$ such that $X^{(i)}(\bpi_{-i})-\lambda_i \cdot\bone = \bzero$.
Since $X^{(i)}(\bpi_{-i})$ and $\bone$ are the gradients of player $i$'s expected payoff $\pi^{(1)}\otimes\cdots\otimes\pi^{(n)} \cdot X^{(i)}$ introduced in \eqref{eq: expected payoff} and the linear function $\sum_{k=1}^{d_i} \pi_k^{(i)} - 1$ respectively, the problem of finding a totally mixed Nash equilibrium of $X$ is the same as the problem of finding local extreme values of $\pi^{(1)}\otimes\cdots\otimes\pi^{(n)} \cdot X^{(i)}$ subject to $\sum_{k=1}^{d_i} \pi_k^{(i)} - 1=0$, or the problem of finding the simultaneous critical points of the functions
\[
\pi^{(1)}\otimes\cdots\otimes\pi^{(n)} \cdot X^{(i)} - \lambda_i\left(\sum_{k=1}^{d_i} \pi_k^{(i)} - 1\right)\colon \R^{d_i} \times \R \to \R
\]
with Lagrangian multipliers $\lambda_i$.

(2) For every $i\in [n]$, let $\ks^i\subset V_i$ be the sphere of equation $\sum_{k=1}^{d_i} (v_k^{(i)})^2 - 1=0$, where $v^{(i)}=(v_1^{(i)},\ldots,v_{d_i}^{(i)})$. The image of the map $\prod_{i=1}^n\ks^i\times\R\to V$ defined by $(v^{(1)},\ldots,v^{(n)},\lambda)\mapsto\lambda\,v^{(1)}\otimes\cdots\otimes v^{(n)}\in V$ is the affine cone $\kc$ over the Segre embedding of the real multiprojective space $\PP^\bd=\prod_{i=1}^n\PP(V_i)$. We equip $V$ with the Frobenius inner product induced by the spheres $\ks^1,\ldots,\ks^n$. A {\em best rank-one approximation} of a real tensor $T\in V$ is a global minimizer of the Frobenius distance function from $T$, restricted to $\kc$.
As described in \cite[Section 7]{FO}, computing a best rank-one approximation of $T$ is equivalent to maximizing the objective function $v^{(1)}\otimes\cdots\otimes v^{(n)} \cdot T$ over $\prod_{i=1}^n\ks^i$. The latter problem is solved by computing the critical points of the function
\begin{equation}\label{eq: optimization problem singular vector tuples}
v^{(1)}\otimes\cdots\otimes v^{(n)} \cdot T - \sum_{i=1}^n \lambda_i\left(\sum_{k=1}^{d_i} (v_k^{(i)})^2 - 1\right) \colon \left(\prod_{i=1}^n \R^{d_i}\right) \times \R^n \to \R
\end{equation}
with Lagrange multipliers $\lambda_i$.
Since $T(\bv_{-i})$ and $v^{(i)}$ are the gradients of $v^{(1)}\otimes\cdots\otimes v^{(n)} \cdot T$ and $\sum_{k=1}^{d_i} (v_k^{(i)})^2 - 1$ respectively, the critical points of \eqref{eq: optimization problem singular vector tuples} correspond to the singular vector tuples of $T$.

To conclude this remark, we highlight a relevant difference between totally mixed Nash equilibria and singular vector tuples. A consequence of \cite[Theorem 1]{FO} is that a generic tensor of format $\bd=(d_1,\ldots,d_n)$ always has at least one singular vector tuple. Furthermore, if $d_1\le\cdots\le d_n$, their number is nondecreasing in $d_n$ and stabilizes when $d_n-1=\sum_{i=1}^{n-1} (d_i-1)$, or when $\bd$ is a boundary format. A geometric description of this phenomenon is described in \cite{ottaviani2021asymptotics}.
Instead, Theorem \ref{thm: number tmNe generic game} implies that a generic game of format $\bd$ with $d_n-1>\sum_{i=1}^{n-1} (d_i-1)$ does not admit totally mixed Nash equilibria.

\section{Nash discriminants and Nash resultants}\label{sec: Nash discriminants}

The primary focus of the previous section was on the expected numbers of totally mixed Nash equilibria of games and their asymptotic behaviors. This section is devoted to studying games with unexpected numbers of totally mixed Nash equilibria and the sets formed by such games. In doing so, we consider the following two cases separately: when the formats of games are balanced and when they are unbalanced. 

If $\bd = (d_1, \ldots, d_n)$ with $d_1 \leq \cdots \leq d_n$ is balanced, i.e., if it satisfies $d_n-1\le\sum_{i=1}^{n-1}(d_i-1)$, then a game of format $\bd$ is anticipated to have an unexpected number of totally mixed Nash equilibria when its Nash equilibrium scheme has an unexpected number of points, or equivalently if it is nonreduced of dimension $0$ (so that its reduced scheme consists of less than the expected number of points) or has a positive dimensional component (so that it contains infinitely many points). We call the variety parameterizing the games whose Nash equilibrium schemes have unexpected numbers of points the {\em Nash discriminant variety}. The primary purpose of Section~\ref{subsec: Nash discriminant} is to show that the real part of the Nash discriminant variety has codimension one. A detailed study of the geometry of the Nash discriminant variety of games for a specific format will be presented in Section~\ref{subsec: 2x2x2 game}. 

The Nash discriminant variety is not always irreducible. In Section~\ref{subsec: Nash discriminant boundary format}, we show that it is reducible and consists of an irreducible hyperplane and a variety of higher codimension if the format is at the boundary between the balanced and unbalanced cases.  

If $\bd$ is unbalanced, then most games of format $\bd$ have no totally mixed Nash equilibria. For the game to have a positive number of totally mixed Nash equilibria, it is necessary for its Nash equilibrium scheme not to be empty. The variety parameterizing games with nonempty Nash equilibrium schemes is called the {\em Nash resultant variety}. In Section~\ref{subsec: Nash resultant}, we prove that the Nash resultant variety is irreducible and give formulas for its dimension and degree.

\subsection{The Nash discriminant variety}\label{subsec: Nash discriminant}

As was discussed in Section \ref{subsec: Nash}, the expected number of totally mixed Nash equilibria of an $n$-player game of format $\bd=(d_1,\ldots, d_n)$ with $d_n-1\le\sum_{i=1}^{n-1}(d_i-1)$ is nonzero and is computed in Theorem \ref{thm: number tmNe generic game}. However, some games have unexpected numbers of totally mixed Nash equilibria. For example, McKelvey and McLennan discussed examples of three-player games with two pure strategies for each player that admit only one totally mixed Nash equilibrium, while the expected number of totally mixed Nash equilibria of such games is two \cite[Section 6]{mckelvey1997maximal}.
There are also games with infinitely many totally mixed Nash equilibria. One example is given in \cite[Section 4]{bubelis1979equilibria} for a six-player binary game where the Nash equilibrium scheme is a one-dimensional manifold. Below, we present yet another such example in detail. 

\begin{example}\label{ex: 222 with infinitely many tmNes}
Keep the same notation as in Section~\ref{ssec:prelim}. Let $X=(X^{(1)},X^{(2)},X^{(3)})$ be the three-player game with the payoff tensors 
\[
X^{(1)} =  
\begin{bmatrix}[cc:cc]
1 & 3 & 2 & 1\\
3 & 2 & 2 & 3
\end{bmatrix}\,,\quad
X^{(2)} =
\begin{bmatrix}[cc:cc]
3 & 2 & 2 & 1\\
3 & 3 & 2 & 4
\end{bmatrix}\,,\quad
X^{(3)} = 
\begin{bmatrix}[cc:cc]
3 & 4 & 5 & 1\\
5 & 1 & 1 & 3
\end{bmatrix}\,.
\]
The Nash equilibrium scheme $\kz_X$ associated with $X$ is defined by the following system: 
\begin{equation}\label{eq: example system 222 game}
\begin{cases}
    0 = \varDelta f_{1,2}^{(1)} &= -2\pi_1^{(2)}\pi_1^{(3)}+\pi_2^{(2)}\pi_1^{(3)}-2\pi_2^{(2)}\pi_2^{(3)}\\ 
    0 = \varDelta f_{1,2}^{(2)} &= \pi_1^{(1)}\pi_1^{(3)}+\pi_1^{(1)}\pi_2^{(3)}-2\pi_2^{(1)}\pi_2^{(3)}\\
    0 = \varDelta f_{1,2}^{(3)} &= -2\pi_1^{(1)}\pi_1^{(2)}+3\pi_1^{(1)}\pi_2^{(2)}+4\pi_2^{(1)}\pi_1^{(2)}-2\pi_2^{(1)}\pi_2^{(2)}\,.
\end{cases} 
\end{equation}
It is straightforward to check that the system \eqref{eq: example system 222 game} has a solution 
\[
([a:b],[-3a+2b:-2a+4b],[-a+2b:a])\in \PP^1 \times \PP^1 \times \PP^1 
\]
for every $[a:b] \in \PP^1$. Therefore, the Nash equilibrium scheme $\kz_X$ contains a curve isomorphic to $\PP^1$ as an irreducible component. In particular, the triples of the vectors
\[
\left(
\left(
\frac{a}{a+b},\, \frac{b}{a+b}
\right), \, 
\left(
\frac{-3a+2b}{-5a+6b},\, \frac{-2a+4b}{-5a+6b}
\right), \, 
\left(
\frac{-a+2b}{2b},\, \frac{a}{2b} 
\right)
\right)\in \Delta_1 \times \Delta_1 \times \Delta_1
\]
are totally mixed Nash equilibria, and hence $X$ admits infinitely many totally mixed Nash equilibria.\hfill$\diamondsuit$
\end{example}

\begin{definition}\label{def: Nash discriminant variety}
Let $n\in\Z_{\mge[2]}$, and let $\bd=(d_1,\ldots,d_n)\in\Z_{\mge[2]}^n$ satisfying the inequalities $d_1\le\cdots\le d_n$ and $d_n-1\le\sum_{i=1}^{n-1}(d_i-1)$. For each $i \in [n]$, we write $V_i$ for a $d_i$-dimensional vector space over $\C$, and $V$ denotes the tensor product of $V_1,\ldots,V_n$. Let $\ku$ be the Zariski open set of $\PP V^{\oplus n}$ consisting of the elements $[X]$ whose Nash equilibrium schemes are reduced of codimension $D=\sum_{i=1}^n(d_i-1)$. We call $\Delta(\bd)\coloneqq \PP V^{\oplus n}\setminus\ku$ the {\em Nash discriminant variety of games of format $\bd$}. 
\end{definition}

\begin{proposition}\label{prop: discriminant_two_players}
If $d \ge 2$, then the Nash discriminant variety $\Delta(d,d)$ has two irreducible components of codimension two and degree $\binom{d}{2}$.
\end{proposition}
\begin{proof}
If $X = (X^{(1)},X^{(2)}) \in V^{\oplus 2}$, then its Nash equilibrium scheme $\kz_X$ is defined by the following two linear systems of $d-1$ equations in $d$ variables:
\begin{equation}\label{eq: linear system 2 players}
\begin{cases}
    0 = \varDelta f_{1,2}^{(1)} = \sum_{j=1}^{d}(x_{1j}^{(1)}- x_{2j}^{(1)})\,\pi_{j}^{(2)}\\
    \quad\vdots\\
    0 = \varDelta f_{1,d}^{(1)} = \sum_{j=1}^{d}(x_{1j}^{(1)}- x_{dj}^{(1)})\,\pi_{j}^{(2)}
\end{cases}
\begin{cases}
    0 = \varDelta f_{1,2}^{(2)} = \sum_{j=1}^{d}(x_{j1}^{(2)}- x_{j2}^{(2)})\,\pi_{j}^{(1)}\\
    \quad\vdots\\
    0 = \varDelta f_{1,d}^{(2)} = \sum_{j=1}^{d}(x_{j1}^{(2)}- x_{jd}^{(2)})\,\pi_{j}^{(1)}\,.
\end{cases}
\end{equation}
By Theorem \ref{thm: number tmNe generic game}, a generic game $X$ admits exactly one totally mixed Nash equilibrium, and hence $X \in \Delta(d,d)$ if and only if $\kz_X$ has a positive dimensional component, or equivalently the $(d-1) \times (d-1)$ minors of the coefficient matrix of at least one of the linear systems in \eqref{eq: linear system 2 players} vanish. Therefore, the Nash discriminant variety $\Delta (d,d)$ is the union of the determinantal varieties defined by the maximal minors of the coefficient matrices of the linear systems in (\ref{eq: linear system 2 players})
\[
\begin{pmatrix}
x_{11}^{(1)}-x_{21}^{(1)} & \cdots & x_{1d}^{(1)}-x_{2d}^{(1)} \\
\vdots & & \vdots \\
x_{11}^{(1)}-x_{d1}^{(1)} & \cdots & x_{1d}^{(1)}-x_{dd}^{(1)} 
\end{pmatrix}, \, 
\begin{pmatrix}
x_{11}^{(1)}-x_{21}^{(1)} & \cdots & x_{1d}^{(1)}-x_{2d}^{(1)} \\
\vdots & & \vdots \\
x_{11}^{(1)}-x_{d1}^{(1)} & \cdots & x_{1d}^{(1)}-x_{dd}^{(1)} 
\end{pmatrix}. 
\]
Note that the linear entries of each of the coefficient matrices are linearly independent, as explained in the proof of Proposition \ref{prop: open subset generic global sections E}. Hence, after a suitable linear change of coordinates, the determinantal variety defined by the maximal minors of each coefficient matrix can be regarded as a cone over the projective determinantal variety in $\PP(\C^{(d-1)\times d})$ of the generic $(d-1)\times d$ matrix of rank at most $d-2$. 
Therefore, Proposition~\ref{prop: discriminant_two_players} follows from \cite[Chapter II.5]{arbarello1985geometry}.
\end{proof}

\begin{remark}\label{rmk: Nash discriminant cone over discriminant E} Let $E$ be the vector bundle on $\PP^\bd$ defined in \eqref{eq: Nash vector bundle}. As explained in Proposition \ref{prop: open subset generic global sections E}, there exists a surjective linear map from $V^{\oplus n}$ to $H^0(\PP^\bd,E)$. This implies that the Nash discriminant variety $\Delta(\bd)$ is a cone over the discriminant variety of~$E$
\[
    \Delta(E)\coloneqq \PP H^0(\PP^\bd,E) \setminus \left\{\left. [f]\in \PP H^0(\PP^\bd,E) \, \right| \, \mbox{$Z(f)$ is reduced of dimension $0$}\right\}.
\]
Therefore, their invariants, such as codimension and degree, are the same. Furthermore, one can derive the properties of $\Delta(\bd)$, such as irreducibility, from $\Delta(E)$. Thus, we focus on studying  $\Delta(E)$ for the rest of the section.
\end{remark}

We show that if $n\ge 3$, then $\Delta(E)$ has a hypersurface component.
The following result is essentially \cite[Theorem 4.1]{mckelvey1997maximal}. We provide its proof because it will be used in Lemma \ref{lem: section with almost all real zeros}.

\begin{lemma}\label{lem: section with all real zeros}
Let $E$ be the vector bundle on $\PP^\bd$ introduced in \eqref{eq: Nash vector bundle}.
There exists a real global section $f$ of $E$ whose zero scheme is reduced of dimension~$0$ and consists only of real points. 
\end{lemma}
\begin{proof}
For each $i \in [n]$, $k \in [n]\setminus \{i\}$, and $j \in [d_i-1]$, let $\ell^{(i)}_{j,k}$ be a generic linear form in $\pi^{(k)}$ with real coefficients, in the following sense: For every $k \in [n]$, any subset $S$ of $\{\ell_{j,k}^{(i)}\mid\text{$i\in[n]\setminus\{k\}$, $j\in[d_i-1]$}\}$ defines a linear subspace of $\PP^{d_k-1}$ of codimension $|S|$. If $f^{(i)}_j \coloneqq \prod_{k \in [n]\setminus\{i\}} \ell^{(i)}_{j,k} \in H^0(\PP^\bd,\ko(\bone_i))$, then define $f^{(i)}\coloneqq (f_1^{(i)},\ldots,f_{d_i-1}^{(i)})\in H^0(\PP^\bd,\ko(\bone_i)^{\oplus(d_i-1)})$ and $f\coloneqq (f^{(1)},\ldots,f^{(n)})\in H^0(\PP^\bd,E)$.
We prove that $f$ satisfies the desired property. 

Let
\begin{equation}\label{eq: sets F_i and F}
F_i\coloneqq\{(i,j)\mid j\in[d_i-1]\}\text{ for every $i\in[n]$,}\quad F\coloneqq \bigcup_{i=1}^n F_i\,.
\end{equation}
Given a $[\bpi] = ([\pi^{(1)}], \ldots, [\pi^{(n)}]) \in \PP^\bd$, define the subsets $A_1, \ldots, A_n$ of $F$ by $(i,j) \in A_k$ if and only if $\ell^{(i)}_{j,k} (\pi^{(k)}) = 0$. Because of the genericity of the linear forms $\ell^{(i)}_{j,k}$, the point $[\bpi]$ is in $\kz(f)$ precisely when the following two conditions are satisfied: (1) For every $k \in [n]$, there exists a unique subset of $\{\ell_{j,k}^{(i)}\mid\text{$i\in[n]\setminus\{k\}$, $j\in[d_i-1]$}\}$ with $d_k-1$ elements that defines $[\pi^{(k)}]$, and (2) for each $(i,j) \in F$, there exists a unique integer $k \in [n]\setminus \{i\}$ such that $(i,j) \in A_k$. In other words, $[\bpi]\in\kz(f)$ if and only if the associated subsets $A_1, \ldots, A_n$ are unique and form a partition of $F$ such that $|A_i|=d_i-1$ and $F_i\cap A_i=\emptyset$ for all $i\in[n]$. The restricted partition $\{A_1,\ldots,A_n\}$ of $F$ (with respect to the sets $F_i$) is called a {\em block derangement} in \cite[Section 1]{vidunas2017counting}. As explained in \cite[Section 2]{vidunas2017counting} or in \cite[Section 3]{mckelvey1997maximal}, the number of block derangements of $F$ is known to be equal to the integer $c(\bd)$ introduced in Theorem~\ref{thm: number tmNe generic game}. We conclude that the zero scheme of $f$ consists of $c(\bd)$ distinct points, and hence it is reduced of dimension $0$. Additionally, these points are all real because so are the coefficients of the linear forms $\ell^{(i)}_{j,k}$.
\end{proof}

\begin{lemma}\label{lem: section with almost all real zeros}
There exists a real global section $f$ of $E$ whose zero scheme is reduced of dimension~$0$ and contains at least two nonreal points. 
\end{lemma}
\begin{proof}
We utilize block derangements to construct a real global section of $E$ whose zero scheme consists of $c(\bd)$ distinct points and has at least one pair of complex points.
Let $F_i$ and $F$ be the sets defined in \eqref{eq: sets F_i and F}. 
Fix two block derangements $\{A_1,\ldots,A_n\}$ and $\{B_1,\ldots,B_n\}$ of $F$ such that $A_k=B_k$ for all $k\in\{3,\ldots,n\}$, while $|A_k\cap B_k|=d_k-2$ if $k\in [2]$.
Notice that, under the assumptions $n\ge 3$ and $d_n-1\le\sum_{i=1}^{n-1}(d_i-1)$, there exist two such block derangements unless $\bd = (2,2,2)$. So, we treat this case separately. In this specific format, one verifies that the zero scheme of the global section
\[
f = (f^{(1)},f^{(2)},f^{(3)}) = (\pi_1^{(2)}\pi_2^{(3)}-\pi_2^{(2)}\pi_1^{(3)}, \pi_1^{(1)}\pi_2^{(3)}-\pi_2^{(1)}\pi_1^{(3)},\pi_1^{(1)}\pi_1^{(2)}+\pi_2^{(1)}\pi_2^{(2)})
\]
of $E$ is reduced of dimension zero and consists of two distinct nonreal points.

Next, assume that $\bd\neq(2,2,2)$. We define a real global section $f$ of $E$ as follows:

$(i)$ If $k\in\{3,\ldots,n\}$ and $(i,j)\in A_k=B_k$, define $f_j^{(i)}\coloneqq\prod_{r\neq i}\ell_{j,r}^{(i)}$ where $\ell^{(i)}_{j,r}$ is a generic linear form in $\pi^{(r)}$ such that $\{\ell^{(i)}_{j,k}\mid (i,j)\in A_k\}$ defines the point $[e_1^{(k)}]\in\PP^{d_k-1}$, where $\{e_1^{(k)},\ldots,e_{d_k}^{(k)}\}$ is the standard basis for $V_k$.

$(ii)$ If $k\in[2]$ and $(i,j)\in A_k\cap B_k$, define $f_j^{(i)}\coloneqq\prod_{r\neq i}\ell_{j,r}^{(i)}$ where $\ell^{(i)}_{j,r}$ is a generic linear form in $\pi^{(r)}$, with the property that $\{\ell^{(i)}_{j,k}\mid (i,j)\in A_k\cap B_k\}$ defines the line $\kl_k\subset\PP^{d_k-1}$ spanned by $[e_1^{(k)}]$ and $[e_2^{(k)}]$.

$(iii)$ Finally, consider $k\in[2]$ and the set $(A_1\cap B_2)\cup(A_2\cap B_1)$ consisting of two elements. For all $(i,j)\in (A_1\cap B_2)\cup(A_2\cap B_1)$, let $b_j^{(i)}$ be a bilinear form in $\pi^{(1)},\pi^{(2)}$, and define $\tilde{b}_j^{(i)}$ to be the restriction of $b_j^{(i)}$ to $\kl_1\times\kl_2$. We claim that there exists a choice of two bilinear forms $b_j^{(i)}$ such that the discriminant of the bilinear system $\widetilde{S}$ in the two pairs of variables $(\pi_1^{(1)},\pi_2^{(1)})$, $(\pi_1^{(2)},\pi_2^{(2)})$ defined by $\{\tilde{b}_j^{(i)}=0 \mid (i,j)\in (A_1\cap B_2)\cup(A_2\cap B_1)\}$ is negative. More explicitly, for all $(i,j)\in (A_1\cap B_2)\cup(A_2\cap B_1)$ we write $b_j^{(i)} = \sum_{\alpha=1}^{d_1}\sum_{\beta=1}^{d_2} c_{\alpha\beta}^{(i,j)}\,\pi_{\alpha}^{(1)}\pi_{\beta}^{(2)}$ for some real coefficients $c_{\alpha\beta}^{(i,j)}$. If $(A_1\cap B_2)\cup(A_2\cap B_1)=\{(\lambda,\mu),(\xi,\eta)\}$, then
\[
\widetilde{S}\colon
\begin{cases}
    \tilde{b}_\mu^{(\lambda)} = c_{11}^{(\lambda,\mu)}\,\pi_1^{(1)}\pi_1^{(2)}+c_{12}^{(\lambda,\mu)}\,\pi_1^{(1)}\pi_2^{(2)}+c_{21}^{(\lambda,\mu)}\,\pi_2^{(1)}\pi_1^{(2)}+c_{22}^{(\lambda,\mu)}\,\pi_2^{(1)}\pi_2^{(2)} = 0\\[2pt]
    \tilde{b}_\eta^{(\xi)} = c_{11}^{(\xi,\eta)}\,\pi_1^{(1)}\pi_1^{(2)}+c_{12}^{(\xi,\eta)}\,\pi_1^{(1)}\pi_2^{(2)}+c_{21}^{(\xi,\eta)}\,\pi_2^{(1)}\pi_1^{(2)}+c_{22}^{(\xi,\eta)}\,\pi_2^{(1)}\pi_2^{(2)} = 0\,.
\end{cases}
\]
Solving $\tilde{b}_\mu^{(\lambda)}=0$ for $(\pi_1^{(2)},\pi_2^{(2)})$ and substituting into $\tilde{b}_\eta^{(\xi)}=0$ yields a quadratic form in $(\pi_1^{(1)},\pi_2^{(1)})$ whose discriminant is the following irreducible polynomial of degree $4$:
\begin{small}
\[
(c_{11}^{(\lambda,\mu)}c_{22}^{(\xi,\eta)}+c_{12}^{(\lambda,\mu)}c_{21}^{(\xi,\eta)}-c_{21}^{(\lambda,\mu)}c_{12}^{(\xi,\eta)}-c_{22}^{(\lambda,\mu)}c_{11}^{(\xi,\eta)})^2-4(c_{11}^{(\lambda,\mu)}c_{21}^{(\xi,\eta)}-c_{21}^{(\lambda,\mu)}c_{11}^{(\xi,\eta)})(c_{12}^{(\xi,\eta)}c_{22}^{(\lambda,\mu)}-c_{22}^{(\xi,\eta)}c_{12}^{(\lambda,\mu)})\,.
\]
\end{small}

Observe that there exists no algebraic relation among the coefficients $c_{\alpha\beta}^{(i,j)}$. In particular, if the coefficients $c_{\alpha\beta}^{(i,j)}$ are generic, then the previous discriminant does not vanish. If for a choice of real coefficients $c_{\alpha\beta}^{(i,j)}$ we obtain a negative value of the discriminant, then the system $\widetilde{S}$ has two pairwise conjugate distinct solutions. Furthermore, the value of the discriminant also remains negative for a small perturbation of the coefficients $c_{\alpha\beta}^{(i,j)}$. For example, if
\[
(c_{11}^{(\lambda,\mu)},c_{12}^{(\lambda,\mu)},c_{21}^{(\lambda,\mu)},c_{22}^{(\lambda,\mu)},c_{11}^{(\xi,\eta)},c_{12}^{(\xi,\eta)},c_{21}^{(\xi,\eta)},c_{22}^{(\xi,\eta)}) = (0,1,-1,0,1,0,0,-1)\,,
\]
then the discriminant of $\widetilde{S}$ takes the negative value $-4$. Using the coefficients above and picking generic values for the remaining coefficients $c_{\alpha\beta}^{(i,j)}$ in $b_j^{(i)}$, define $f_j^{(i)} \coloneqq b_j^{(i)}\,\prod_{r\in\{3,\ldots,n\}\setminus\{i\}}\ell_{j,r}^{(i)}$ for every $(i,j)\in (A_1\cap B_2)\cup(A_2\cap B_1)$, where $\ell_{j,r}^{(i)}$ is a generic linear form in $\pi^{(r)}$.

We set $f^{(i)}\coloneqq (f_1^{(i)},\ldots,f_{d_i-1}^{(i)})\in H^0(\PP^\bd,\ko(\bone_i)^{\oplus(d_i-1)})$ using the components $f_j^{(i)}$ defined in steps $(i)-(iii)$.
We show that every block derangement $C=\{C_1,\ldots,C_n\}$ of $F$ can be associated with a point in the zero scheme of $f\coloneqq (f^{(1)},\ldots,f^{(n)})\in H^0(\PP^\bd,E)$:

$(1)$ Suppose that $C$ is either $A$ or $B$. Consider the system
\[
S_1\colon
\begin{cases}
    b_j^{(i)}=0 & \text{if $(i,j)\in (A_1\cap B_2)\cup(A_2\cap B_1)$}\\[2pt]
    \ell_{j,k}^{(i)}=0 & \text{for every $k\in[n]$ and $(i,j)\in A_k\cap B_k$.}
\end{cases}
\]
The last set of the $\dim\PP^\bd-2=D-2$ linear equations of $S_1$ defines the subscheme $\kl_1\times\kl_2\times \prod_{k=3}^n\{[e_1^{(k)}]\} \subset \PP^\bd$. Restricting the first subsystem of $S_1$, consisting of the two bilinear equations, to this subscheme yields the bilinear system $\widetilde{S}$ previously studied. In particular, the solution set of $S_1$ is the union of two pairwise conjugate points of $\PP^\bd$.

$(2)$ Suppose that $C$ is different from $A$ and $B$, and that $C_k\cap[(A_1\cap B_2)\cup(A_2\cap B_1)]\neq\emptyset$ for every $k\in[2]$. Consider the system
\[
S_2\colon
\begin{cases}
    b_j^{(i)}=0 & \text{for every $k\in[2]$ and $(i,j)\in C_k\cap[(A_1\cap B_2)\cup(A_2\cap B_1)]$}\\[2pt]
    \ell_{j,k}^{(i)}=0 & \text{for every $k\in[n]$ and $(i,j)\in C_k\setminus [(A_1\cap B_2)\cup(A_2\cap B_1)]$.}
\end{cases}
\]
The last set of the $D-2$ linear equations of $S_2$ defines the subscheme $\kl_1'\times\kl_2'\times \prod_{k=3}^n \{p^{(k)}\} \subset \PP^\bd$, where $\kl_k'\subset\PP^{d_k-1}$ is the line defined by $\{\ell^{(i)}_{j,k}\mid (i,j)\in C_k\setminus [(A_1\cap B_2)\cup(A_2\cap B_1)]\}$ for all $k\in[2]$ and $p^{(k)}\in \PP^{d_k-1}$ is the point defined by $\{\ell^{(i)}_{j,k}\mid (i,j)\in C_k\setminus [(A_1\cap B_2)\cup(A_2\cap B_1)]\}$ for all $k\in\{3,\ldots,n\}$. The intersection between $\kl_1'\times\kl_2'\times \prod_{k=3}^n \{p^{(k)}\}$ and the subset cut out by the first two bilinear equations of $S_2$ consists of two distinct points, both different from the solutions of $S_1$. One of the solutions of $S_2$ is the point of $\PP^d$ associated with the block derangement $C$.

$(3)$ Suppose that $[(A_1\cap B_2)\cup(A_2\cap B_1)]\subset C_k$ for some $k\in[2]$. This case holds only if $d_k\ge 3$. We consider the case $k=1$ for the sake of simplicity. Consider the system 
\[
S_3\colon
\begin{cases}
    b_j^{(i)}=0 & \text{if $k=1$ and $(i,j)\in (A_1\cap B_2)\cup(A_2\cap B_1)\subset C_1$}\\[2pt]
    \ell_{j,k}^{(i)}=0 & \text{for every $k\in[n]$ and $(i,j)\in C_k\setminus [(A_1\cap B_2)\cup(A_2\cap B_1)]$.}
\end{cases}
\]
In this case, the last set of the $D-2$ linear equations of $S_3$ defines the subscheme $\kn\times\prod_{k=2}^n\{q^{(k)}\}\subset \PP^\bd$, where $\kn$ is a plane in $\PP^{d_1-1}$ and $q^{(k)}$ is a point in $\PP^{d_k-1}$ for all $k\in\{2,\ldots,n\}$.
The intersection between $\kn\times\prod_{k=2}^n\{q^{(k)}\}$ and the subset cut out by the first two bilinear equations of $S_3$ is the point in $\PP^\bd$ associated with the block derangement $C$.

$(4)$ Suppose that $|C_1\cap[(A_1\cap B_2)\cup(A_2\cap B_1)]|=1$ and $C_2\cap[(A_1\cap B_2)\cup(A_2\cap B_1)]=\emptyset$. Similarly, one considers the case where the sets $C_1$ and $C_2$ are swapped. Consider the system 
\[
S_4\colon
\begin{cases}
    b_j^{(i)}=0 & \text{if $k=1$ and $(i,j)\in C_1\cap[(A_1\cap B_2)\cup(A_2\cap B_1)]$}\\[2pt]
    \ell_{j,k}^{(i)}=0 & \text{for every $k\in[n]$ and $(i,j)\in C_k\setminus[(A_1\cap B_2)\cup(A_2\cap B_1)]$.}
\end{cases}
\]
In this case, the last set of the $D-1$ linear equations of $S_4$ defines the subscheme $\kl\times\prod_{k=2}^n\{r^{(k)}\} \subset \PP^\bd$, where $\kl$ is a line in $\PP^{d_1-1}$ and $r^{(k)}$ is a point in $\PP^{d_k-1}$ for all $k\in\{2,\ldots,n\}$.
The intersection between $\kl\times\prod_{k=2}^n\{r^{(k)}\}$ and the subset cut out by the first bilinear equation of $S_4$ is the point in $\PP^\bd$ associated with the block derangement $C$.

$(5)$ The last case is when $C_k\cap[(A_1\cap B_2)\cup(A_2\cap B_1)]=\emptyset$ for every $k\in[2]$. Then the point of $\PP^\bd$ associated with the block derangement $C$ is the solution of the linear system $S_5$ consisting of $\ell_{j,k}^{(i)}=0$ for every $k\in[n]$ and $(i,j)\in C_k$.

As discussed in the proof of Lemma \ref{lem: section with all real zeros}, the number of block derangements of $F$ equals $c(\bd)$. The genericity of the linear forms $\ell_{j,k}^{(i)}$ and of the coefficients of the bilinear forms $b_j^{(i)}$ implies that all the solutions computed via the systems $S_1,\ldots,S_5$ are pairwise distinct. All these considerations imply that the zero scheme of the real global section $f$ of $E$ consists of $c(\bd)$ distinct points of $\PP^\bd$, and hence it is reduced and zero-dimensional. Furthermore, two points of $\kz(f)$ have nonreal coordinates. This concludes the proof.
\end{proof}

The two preceding lemmas serve as a preparation for Theorem \ref{thm: codimension real part discriminant of E is 1}. The main idea of its proof is as follow. Consider two distinct sections $f_1,f_2$ in $H^0(\PP^\bd,E)$ with real coefficients in the complement of the real part of $\Delta(E)$, hence $\kz(f_i)$ is a zero-dimensional reduced complex scheme of $\PP^\bd$ such that $|\kz(f_i)\cap\PP^\bd_\mR|<\infty$ for all $i\in\{1,2\}$. If $f_1$ and $f_2$ are connected by a path $\{f_t\mid t\in[1,2]\}\subseteq H^0(\PP^\bd,E)$ of sections with real coefficients, such that $0\in\R^D$ is a regular value of $f_t\colon\PP^\bd_\mR\to\R^D$ for all $t\in[1,2]$, then Thom's Isotopy Lemma ensures that the topology of $\kz(f_t)\cap\PP^\bd_\mR$ is preserved along the parameter space $[1,2]$. This implies that if additionally $|\kz(f_1)\cap\PP^\bd_\mR|=|\kz(f_2)\cap\PP^\bd_\mR|$, then  $|\kz(f_t)\cap\PP^\bd_\mR|=|\kz(f_1)\cap\PP^\bd_\mR|$ for all $t\in[1,2]$. Therefore, the path $\{f_t\mid t\in[1,2]\}$ does not meet the real part of $\Delta(E)$. In the upcoming proof, we instead start with two distinct sections $f_1,f_2$ in $H^0(\PP^\bd,E)$ such that $|\kz(f_1)\cap\PP^\bd_\mR|\neq|\kz(f_2)\cap\PP^\bd_\mR|$. We apply Thom's Isotopy Lemma to show that every path $\{f_t\mid t\in[1,2]\}$ must intersect the real part of $\Delta(E)$. In particular, the complement of the latter subset is not path-connected, which implies that the real part of $\Delta(E)$ has codimension one.

\begin{theorem}\label{thm: codimension real part discriminant of E is 1}
Let $E$ be the vector bundle on $\PP^\bd$ defined in \eqref{eq: Nash vector bundle} and let $\Delta(E)$ be its discriminant. The real part of $\Delta(E)$ has codimension one in the vector space of global sections $f$ of $E$ with real coefficients. In particular, $\codim\,\Delta(E)=1$.
\end{theorem}
\begin{proof}
Let $\PP^\bd_\mR$ be the multiprojective space over $\R$.
By Lemmas \ref{lem: section with all real zeros} and \ref{lem: section with almost all real zeros}, there exist two sections $f_1,f_2$ in $H^0(\PP^\bd,E)$ with real coefficients such that $\kz(f_1)$ and $\kz(f_2)$ are both zero-dimensional reduced complex schemes of $\PP^\bd$ such that
\begin{equation}\label{eq: different cardinalities}
    |\kz(f_1)\cap\PP^\bd_\mR|=c(\bd) \text{ and } |\kz(f_2)\cap\PP^\bd_\mR|\le c(\bd)-2\,.
\end{equation}

Consider any smooth map $F\colon\PP^\bd_\mR\times[1,2]\to\R^D$ and define $f_t\coloneqq F(\cdot,t)\colon\PP^\bd_\mR\to\R^D$ for all $t\in[1,2]$. Using \eqref{eq: different cardinalities} and applying Thom's Isotopy Lemma (for the version used, see Theorem 2.6 in the lecture notes \cite{lerario2023topology} with $M=\PP^\bd$, $N=\R^D$, and $A=\{0\}\subset\R^D$) show that there exist $t^*\in [1,2]$ such that $f_{t^*}$ is not transversal to $\{0\}$, or equivalently $\kz(f_{t^*})\cap\PP^\bd_\mR$ is nonreduced. This means that the path $\{f_t\mid t\in[1,2]\}$ in $H^0(\PP^\bd,E)$ between $f_1$ and $f_2$ must intersect the real part of the discriminant $\Delta(E)$. This completes the proof.
\end{proof}
\begin{problem}
    Study the irreducible components and the degree of the Nash discriminant variety.
\end{problem}
The computation of the codimension of $\Delta(\bd)$, or of $\Delta(E)$, is a fundamental step towards a complete description of these discriminants. Secondly, one is interested in showing whether these discriminants are irreducible and possibly computing the degrees of their irreducible components. In Section \ref{subsec: Nash discriminant boundary format}, we will show that, when $d_n-1 = \sum_{i=1}^{n-1}(d_i-1)$, then $\Delta(\bd)$ contains at least two components, one of which is an irreducible hypersurface, while the other component has codimension at least two. Therefore, a complete study of $\Delta(E)$ is an interesting and difficult problem. The main obstruction comes from the fact that $E$ is not very ample. If we considered a very ample vector bundle instead, then the study of its discriminant could be addressed by applying \cite[Corollary 2.7]{abo2022ramification}.

In the remainder of this section, we restrict to the hypersurface component of $\Delta(\bd)$, and we compute its degree for $\bd\in\{(2,2,2),(2,2,3)\}$. In general, thanks to Remark \ref{rmk: Nash discriminant cone over discriminant E}, we study the discriminant $\Delta(E)$ instead of $\Delta(\bd)$.

\begin{example}\label{ex: computation degree hypersurface component 222}
Let $\bd = (2,2,2)$. If $E$ denotes the vector bundle on $\PP^\bd$ defined in \eqref{eq: Nash vector bundle}, then $E = \ko(0,1,1)\oplus\ko(1,0,1)\oplus\ko(1,1,0)$. By Theorem \ref{thm: codimension real part discriminant of E is 1}, the discriminant variety $\Delta(E)$ of $E$ contains a hypersurface. In this example, we show that the degree of this hypersurface component is $6$. This is done by calculating the number of its intersection points with the pencil $\kl$ of generic global sections $f_1$ and $f_2$ of $E$. We use an approach similar to \cite{abo2020discriminant}. 

The proof of Lemma~\ref{lem: section with all real zeros} shows the existence of a global section of $E$ whose zero scheme is a nonreduced point of multiplicity $2$. This means that the global sections of $E$, whose zero schemes are such nonreduced points, form an open subset of $\Delta(E)$. Thus, we may assume that the global section of $E$ corresponding to any intersection point of $\Delta(E)$ and $\kl$ defines a nonreduced point of multiplicity $2$. 

Since $E$ is globally generated as explained at the beginning of Section \ref{subsec: vector bundles}, the dependency locus $\kc$ of $f_1$ and $f_2$ is a nonsingular curve \cite[Lemma 2.5]{ein1982some}. We show the irreducibility of $\kc$ and find the genus $g(\kc)$ of $\kc$ by utilizing the Eagon-Northcott complex of the sheaf morphism $(f_1,f_2) \colon\ko^{\oplus 2} \to E$ 
\begin{equation}\label{eq: Eagon-Northcott}
0 \to \left(\bigwedge^3 E^*\right)^{\oplus 2} \longrightarrow \bigwedge^2E^* \longrightarrow \ko \longrightarrow \ko_\kc \to 0\,, 
\end{equation}
which is a locally free resolution of the structure sheaf $\ko_\kc$ of $\kc$.

Note that 
\[
\bigwedge^i  E^* = 
\begin{cases}
\ko(-2,-1,-1)\oplus\ko(-1,-2,-1)\oplus\ko(-1,-1,-2) & \text{if $i=2$,}\\ 
\ko(-2,-2,-2) & \text{if $i=3$,} 
\end{cases}
\]
from which it follows that $H^i(\PP^\bd,\bigwedge^2 E^*) = 0$ for each $i \in \{0,1,2,3\}$ and 
\[
\dim H^i(\PP^\bd,\bigwedge^3 E^*) = 
\begin{cases}
2 & \text{if $i = 3$,} \\
0 & \text{otherwise.}
\end{cases}
\]
Therefore, decomposing \eqref{eq: Eagon-Northcott} into two short exact sequences and taking cohomology show that $\dim H^0(\kc,\ko_\kc) = \dim H^0(\PP^\bd,\ko) = 1$ and $\dim H^1(\kc,\ko_\kc) = \dim H^3 (\PP^\bd, \bigwedge^3 E^*) = 2$, and hence the curve $\kc$ is irreducible, and its genus $g(\kc)$ is $2$.

Define the morphism $\psi\colon \kc \to \kl$ by $\psi(p) \coloneqq [f]$ if and only if $p \in Z(f)$. This is a finite morphism of degree $c(2,2,2) = 2$. If $[f] \in \Delta(E) \cap \kl$, then, by assumption, the zero scheme $Z(f)$ of $f$ is a nonreduced point of multiplicity $2$ supported at a point $p \in \kc$. To put it another way, this means that the finite morphism $\psi$ is ramified at $p$ with ramification index $2$. Therefore, finding the intersection number $|\Delta(E) \cap \kl|$ is equivalent to counting the number of branch points of $\psi$. Furthermore, since the ramification index of $\psi$ at each ramification point is $2$, the number of branch points of $\psi$ coincides with the degree of the ramification divisor of $\psi$. It follows from the Hurwitz-Riemann formula \cite[Chapter IV, Section 2] {hartshorne1977algebraic} that the degree of the ramification divisor of $\psi$, and hence the degree of $\Delta(E)$, is $(2-2\,g(\kl))\deg\psi+2\,g(\kc)-2 = 6$.\hfill$\diamondsuit$
\end{example}

\begin{example}\label{ex: computation degree hypersurface component 223}
Let $\bd = (2,2,3)$. Since $c(2,2,3)=2$, one can use the same strategy as in Example~\ref{ex: computation degree hypersurface component 222} to compute the degree of the hypersurface component of the discriminant variety of the vector bundle $E = \ko(0,1,1)\oplus\ko(1,0,1)\oplus\ko(1,1,0)^{\oplus 2}$. 

It can be verified, as in Example~\ref{ex: computation degree hypersurface component 222}, that the dependency locus $\kc$ of two generic global sections $f_1$ and $f_2$ of $E$ is a nonsingular irreducible curve of genus $1$ in $\PP^\bd$ by calculating the Eagon-Northcott complex of the sheaf morphism $(f_1,f_2)\colon \ko^{\oplus 2} \to E$. The map $\psi\colon\kc\to\kl$ defined in the same way as in Example~\ref{ex: computation degree hypersurface component 222} is a finite morphism of degree $c(2,2,3)=2$. Thus, the Hurwitz-Riemann formula shows that the degree of the ramification divisor of $\psi$, and hence the degree of $\Delta(E)$, is $4$.\hfill$\diamondsuit$
\end{example}

\noindent We suspect that the same approach holds for any format $\bd$ such that $d_1\le\cdots\le d_n$ and $d_n-1\le\sum_{i=1}^{n-1}(d_i-1)$. In particular, the degree of the one-codimensional part of $\Delta(E)$ is equal to $2(c(\bd)+g(\kc)-1)$, where $\kc$ is the dependency locus of two generic global sections of $E$.

\subsection{The Nash discriminant variety of a three-player binary game}\label{subsec: 2x2x2 game}

In this section, we focus on three-player games of format $\bd=(2,2,2)$. Thanks to Remark \ref{rmk: Nash discriminant cone over discriminant E}, the study of $\Delta(2,2,2)$ is equivalent to the study of $\Delta(E)$, where in this format the vector bundle $E$ on $\PP^\bd$ defined in \eqref{eq: Nash vector bundle} is $E = \ko(0,1,1)\oplus\ko(1,0,1)\oplus\ko(1,1,0)$ . Throughout this section, we use $\Delta$ to denote $\Delta(E)$ for simplicity. Firstly, in Theorem \ref{thm: degree Nash discriminant variety 2x2x2}, we show that $\Delta\subset\PP H^0(\PP^\bd,E)$ is an irreducible hypersurface of degree $6$. Secondly, since a generic global section $f$ of $E$ is such that $\kz(f)$ is a nonreduced point of multiplicity $2$, we investigate the loci of global sections such that $\dim(\kz(f))>0$. In particular, in Propositions \ref{prop:curve_component} and \ref{prop:surface_component}, we characterize all zero schemes $\kz(f)$ of global sections containing either a one-dimensional or a two-dimensional component. Lastly, we compute the codimensions, degrees, and equations and study the irreducibility of the loci of global sections $f$ of $E$ such that $\kz(f)$ is positive-dimensional and of a specific type. With the aid of \texttt{Macaulay2}, we utilize these results to determine all singular strata of $\Delta$ (Table~\ref{tab: singular strata}, \cite{mathrepo}).

In this section, a global section in $H^0(\PP^\bd,E)$ is a triple $f=(f^{(1)},f^{(2)},f^{(3)})$, where
\begin{align}\label{eq: components global section E 222}
\begin{split}
    f^{(1)} &= a_{11}^{(1)}\pi_1^{(2)}\pi_1^{(3)}+a_{12}^{(1)}\pi_1^{(2)}\pi_2^{(3)}+ a_{21}^{(1)}\pi_2^{(2)}\pi_1^{(3)}+a_{22}^{(1)}\pi_2^{(2)}\pi_2^{(3)}\\
    f^{(2)} &= a_{11}^{(2)}\pi_1^{(1)}\pi_1^{(3)}+a_{12}^{(2)}\pi_1^{(1)}\pi_2^{(3)}+ a_{21}^{(2)}\pi_2^{(1)}\pi_1^{(3)}+a_{22}^{(2)}\pi_2^{(1)}\pi_2^{(3)}\\
    f^{(3)} &= a_{11}^{(3)}\pi_1^{(1)}\pi_1^{(2)}+a_{12}^{(3)}\pi_1^{(1)}\pi_2^{(2)}+ a_{21}^{(3)}\pi_2^{(1)}\pi_1^{(2)}+a_{22}^{(3)}\pi_2^{(1)}\pi_2^{(2)}\,.
\end{split}
\end{align}

We note a similarity between our approach and the universality theorem on Nash equilibria. The theorem proves that every real algebraic variety is isomorphic to the set of totally mixed Nash equilibria of a three-player game or an $n$-player game in which each player has two strategies. However, even for a real algebraic curve, more than two strategies per player are needed in a three-player game, or more than three players in a binary game (see \cite[Theorem 5, Theorem 6]{datta2003universality}). While the universality result begins with a given variety and constructs a corresponding game, our approach follows a similar line of thought but takes a slightly different focus: we fix the game and examine all possible zero schemes of global sections $f$ of $E$.

\begin{theorem}\label{thm: degree Nash discriminant variety 2x2x2}
The discriminant locus $\Delta$ of $E$ is an irreducible hypersurface of degree $6$.
\end{theorem}
\begin{proof}
Let $\Psi\colon \PP H^0(\PP^\bd,E) \dashrightarrow P \coloneqq \PP(H^0(\PP^\bd, \ko(\bone_1)) \oplus H^0(\PP^\bd, \ko(\bone_2))$ be the projection from $\PP H^0(\PP^\bd, \ko(\bone_3))$ to $P$, and define the open subset $\ku$ of $P$ to be
\[
\ku\coloneqq \left\{\left. \left[\left(f^{(1)},f^{(2)}\right)\right] \in P  \, \right| \, \text{$f^{(1)}$ and $f^{(2)}$ are irreducible}\right\}\,. 
\]  
The restriction of $\Psi$ to $\Delta$, for which we also write~$\Psi$, is onto. We show that there exists an open subset $\ku_0$ of $\ku$ such that the fiber of $\Psi\colon \Psi^{-1}(\ku_0) \to \ku_0$ over any point of $\ku_0$ is irreducible of dimension~$3$. 

Let $f = (f^{(1)}, f^{(2)}, f^{(3)})\in H^0(\PP^\bd,E)$ and  $([\pi^{(1)}], [\pi^{(2)}], [\pi^{(3)})]) \in \kz(f)$. Since $f^{(1)}(\pi^{(2)}, \pi^{(3)}) = f^{(2)}(\pi^{(1)}, \pi^{(3)}) = 0$, it is immediate that $(\pi^{(i)}_1, \pi^{(i)}_2)$ is a scalar multiple of 
\[
\pi^{(i)}\left(\pi^{(3)}_1,\pi^{(3)}_2\right) \coloneqq \left(a^{(i)}_{21}\pi^{(3)}_1+a^{(i)}_{22} \pi^{(3)}_2, -\left(a^{(i)}_{11}\pi^{(3)}_1+a^{(i)}_{11} \pi^{(3)}_2\right)\right) 
\]
for each $i \in \{2,3\}$. Thus, $[f] \in \Delta$ if and only if the following binary quadratic form in $\pi^{(3)}_1$ and $\pi^{(3)}_2$ defines a nonreduced point of multiplicity $2$ in $\PP(V_3)$: 
\[
f^{(3)}\left(\pi^{(2)}\left(\pi^{(3)}_1,\pi^{(3)}_2\right),\pi^{(3)}\left(\pi^{(3)}_1,\pi^{(3)}_2\right)\right)\,.  
\]
Its discriminant $\Delta(f^{(3)})$ is a quadratic form in $a^{(3)}_{11}, \ldots, a^{(3)}_{22}$ whose coefficients are bihomogeneous polynomials of bidegree $(2,2)$ in $a^{(1)}_{11}, \ldots, a^{(1)}_{22}$ and $ a^{(2)}_{11},\ldots, a^{(2)}_{22}$. This quadratic form defines a nonsingular quadric surface in $\PP H^0(\PP^\bd, \ko(\bone_3))$ if and only if the $4 \times 4$ symmetric matrix of the quadratic form has rank $4$. The complement $\ku_0$ of the hypersurface defined by the $4 \times 4$ symmetric matrix in $\ku$ is a nonempty open subset (for example, the symmetric matrix has rank $4$ if $f^{(1)} = \pi^{(2)}_1\pi^{(3)}_1-\pi^{(2)}_1\pi^{(3)}_2-\pi^{(2)}_2\pi^{(3)}_1-\pi^{(2)}_2\pi^{(3)}_2$ and $f^{(2)} = -\pi^{(1)}_1\pi^{(3)}_1+\pi^{(1)}_1\pi^{(3)}_2+\pi^{(1)}_2\pi^{(3)}_1+\pi^{(1)}_2\pi^{(3)}_2$).  

The fiber of $\Psi$ over any point $[(f^{(1)},f^{(2)})]$ of $\ku_0$,  
\[
\Psi^{-1}\left(\left[\left(f^{(1)},f^{(2)}\right)\right]\right) = 
\left\{\left. 
\left[\left(f^{(1)},f^{(2)},f^{(3)}\right)\right] \, \right| \, 
\Delta\left(f^{(3)}\right) = 0
\right\}\,, 
\]
is identified with the affine cone over a nonsingular quadric surface, and hence it is irreducible and has dimension $3$. 
\end{proof}

\begin{remark}
Let $f = (f^{(1)}, f^{(2)}, f^{(3)})$ be a global section of $E$, and let $\Theta_1\colon V_2 \to V_3^*$, $\Theta_2\colon V_1\to V_3^*$, and $\Theta_3\colon V_1 \to V_2^*$ be the linear transformations corresponding to $f^{(1)}$, $f^{(2)}$, and $f^{(3)}$ respectively.  
Define the linear transformation $\Theta_f$ from $V_1\oplus V_2 \oplus V_3 $ to $V_1^*\oplus V_2^* \oplus V_3^*$ by 
\[
    \Theta_f \coloneqq \left(\Theta_2+\Theta_3, \Theta_1+\Theta_3^T, \Theta_1^T+\Theta_2^T\right)\,.
\]
Emiris and Vidunas in \cite[Section 5]{emiris2016discriminants} show that the rank of $\Theta_f$ is less than or equal to $5$ if and only if the zero scheme of $f$ is singular at the point of $\PP^\bd$ corresponding to the kernel of~$\Theta_f$. 
A straightforward calculation shows that the matrix of $\Theta_f$ relative to the standard bases $\{e^{(i)}_1, e^{(i)}_2\}_{i \in [3]}$ for $V_1\oplus V_2 \oplus V_3$ and $\{\pi^{(i)}_1,\pi^{(i)}_2\}_{i \in [3]}$ for $V_1^*\oplus V_2^* \oplus V_3^*$ is 
\begin{equation}\label{eq: 6x6 matrix discriminant}
    \begin{pmatrix}
    0 & 0 & a^{(3)}_{11} & a^{(3)}_{12} & a^{(2)}_{11} & a^{(2)}_{12} \\[2pt]
    0 & 0 & a^{(3)}_{21} & a^{(3)}_{22} & a^{(2)}_{21} & a^{(2)}_{22} \\[2pt]
    a^{(3)}_{11} & a^{(3)}_{21} & 0 & 0 & a^{(1)}_{11} & a^{(1)}_{12} \\[2pt]
    a^{(3)}_{12} & a^{(3)}_{22} & 0 & 0 & a^{(1)}_{21} & a^{(1)}_{22} \\[2pt]
    a^{(2)}_{11} & a^{(2)}_{21} & a^{(1)}_{11} & a^{(1)}_{21} & 0 & 0 \\[2pt]
    a^{(2)}_{12} & a^{(2)}_{22} & a^{(1)}_{12} & a^{(1)}_{22} & 0 & 0 
    \end{pmatrix}\,,    
\end{equation}
and hence its determinant is the defining equation for $\Delta$. 
\end{remark}

\begin{example}[Selten's Horse]\label{ex: Selten}
We discuss a famous example from Selten \cite{selten1975reexamination}, known as ``Selten's Horse'', see also \cite[Section 4.3]{jahani2022automated}.
Let $X = (X^{(1)}, X^{(2)}, X^{(3)})$ be the three-player game of format $\bd=(2,2,2)$ with
\[
X^{(1)} =
\begin{bmatrix}[rr:rr]
3 & 3 & 0 & 0 \\
4 & 1 & 0 & 1
\end{bmatrix},\quad
X^{(2)} =
\begin{bmatrix}[rr:rr]
2 & 2 & 0 & 0 \\
4 & 1 & 0 & 1
\end{bmatrix},\quad
X^{(3)} = 
\begin{bmatrix}[rr:rr]
2 & 2 & 0 & 0 \\
0 & 1 & 1 & 1
\end{bmatrix}\,.
\]
The global section $f\in H^0(\PP^\bd,E)$ associated with the game $X$ is
\[
f = (-\pi_1^{(2)}\pi_1^{(3)}+2\pi_2^{(2)}\pi_1^{(3)}-\pi_2^{(2)}\pi_2^{(3)},3\pi_2^{(1)}\pi_1^{(3)}-\pi_2^{(1)}\pi_2^{(3)},2\pi_1^{(1)}\pi_1^{(2)}+2\pi_1^{(1)}\pi_2^{(2)}-\pi_2^{(1)}\pi_1^{(2)})\,.
\]
One verifies by direct computations that $\kz(f)$ is a nonreduced point of multiplicity $2$ supported at $[\bpi] = ([0:1],[1:-1]:[1:3]) \in \PP^\bd$.
This can be verified by observing that the determinant of the matrix obtained by evaluating \eqref{eq: 6x6 matrix discriminant} at the coefficients of $f$ is zero. This implies $[f] \in \Delta$. The point $[\bpi]$ cannot be associated with a triple of probability distributions. Hence, the game $X$ does not admit totally mixed Nash equilibria. However, it has mixed Nash equilibria, which is the union of two line segments containing two pure Nash equilibria. This prompts the following natural question.\hfill$\diamondsuit$
\end{example}

\begin{problem}\label{prob: mixed NE}
Can an $n$-player game in the Nash discriminant variety have a finite, nonempty set of mixed (but not totally mixed) Nash equilibria? Also, characterize the set of games whose Nash equilibria are nonreduced or form a positive-dimensional set.
\end{problem}

The existence of completely mixed games, i.e., games where all Nash equilibria are totally mixed, has been well studied. Kaplansky first analyzed such equilibria in zero-sum two-player games \cite{kaplansky1945contribution}, and this was later generalized to general two-player games in \cite{bohnenblust1950solutions, raghavan1970completely}. Chin, Parthasarathy, and Raghavan extended this study to $n$-player games \cite{chin1974structure}.
The approach introduced in this paper, building on extensive existing literature, may offer a new framework for addressing Problem~\ref{prob: mixed NE}.

\begin{proposition}\label{prop:curve_component}
If the zero scheme of a global section of $E$ contains a curve as a component (but no surface component), then the zero scheme of the global section is isomorphic to a (possibly degenerate) twisted cubic, a (possibly degenerate) plane conic, or a line. 
\end{proposition}
\begin{proof}
Let $f \in H^0(\PP^\bd,E)$, and for each $i \in [3]$, let $f^{(i)} \in  H^0(\PP^\bd,\ko(\bone_i))$ be the $i$th coordinate function of $f$. 

If any two of $f^{(1)}$, $f^{(2)}$, and $f^{(3)}$ share a common factor, then we may assume, up to scaling, that there exist linear forms~$L$, $M$, and $N$ in $\pi^{(1)}, \pi^{(2)}$, and $\pi^{(3)}$ respectively such that $f^{(1)} = MN$, $f^{(2)} = LM$, and $f^{(3)} = LM$. The zero scheme $\kz(f)$ of $f$ is the union of three nonplanar, concurrent lines (the lines are defined by $M=N=0$, $L=N=0$, and $L=M=0$, and they intersect in the point defined by $L=M=N=0$). Therefore, it is a degenerate twisted cubic. Hence, we may assume that a pair of the coordinate functions of $f$ exists without a common factor. We may also assume, without loss of generality, that $(f^{(2)}, f^{(3)})$ is such a pair. 
    
If $\kc$ denotes the subscheme defined by $f^{(2)}$ and $f^{(3)}$, then $\kz(f)$ is the intersection of $\kc$ with the surface defined by~$f^{(1)}$. In particular, if $f^{(1)} = 0$, then $\kc=\kz(f)$. Thus, we prove that $\kc$ is a possibly degenerate twisted cubic and discuss which irreducible component of $\kc$ appears as a component of $\kz(f)$. 
    
The locally free resolution of $\ko_\kc$, 
\[
\begin{CD}
    0 @>>> \ko(-2,-1,-1) @>{\left(\begin{array}{c}-f^{(3)} \\ f^{(2)}\end{array}\right)}>>{\begin{array}{c} \ko(-1,0,-1) \\ \oplus \\ \ko(-1,-1,0) \end{array}} @>{\left(\begin{array}{cc}f^{(2)} & f^{(3)} \end{array}\right)}>> \ko\,,   
\end{CD}
\]
gives rise to $\dim H^0(\kc,\ko_\kc(1))=4$ and implies that the Hilbert polynomial of $\kc$ is $3t+1$. Therefore, it follows that $\kc$ is a space curve of dimension~$1$, degree~$3$, and arithmetic genus~$0$, and hence it is a (possibly degenerate) twisted cubic in the Segre embedding. 
    
We show that if $\kc$ contains a line, then one of $f^{(2)}$ and $f^{(3)}$ is reducible, which would imply that if $f^{(2)}$ and $f^{(3)}$ are irreducible, then the subscheme they define is a nonsingular twisted cubic. Assume that $f^{(2)}$ is irreducible. If the line is defined by a linear form $L$ in $\pi^{(1)}$ and a linear form $N$ in $\pi^{(3)}$, then there exist a linear form $L'$ in $\pi^{(1)}$ and a linear form $N'$ in $\pi^{(3)}$ such that $f^{(2)} = LN'+L'N$. Since $f^{(3)}$ is bihomogeneous in $\pi^{(1)}$ and $\pi^{(2)}$, the surface defined by $f^{(3)}$ contains the line precisely when there exists a linear form $M$ in $\pi^{(2)}$ such that $f^{(3)}=LM$. Hence, it is reducible, and the subscheme defined by $f^{(2)}$ and $f^{(3)}$ consists of the line defined by $f^{(2)} = L=0$ and the nonsingular plane conic defined by $f^{(2)} = M = 0$. If both $f^{(2)}$ and $f^{(3)}$ are reducible but do not share any factor, say $f^{(2)}=LN$ and $f^{(3)}=L'M$, where $L$ and $L'$ are linearly independent linear forms in $\pi^{(1)}$, $M$ is a linear form in $\pi^{(2)}$, and $N$ is a linear form $\pi^{(3)}$, then the subscheme defined by $f^{(2)}$ and $f^{(3)}$ is the union of three lines defined by $M=N=0$, $L=M=0$, and $L'=N=0$, respectively. So, we showed that $\kc$ is a nonsingular twisted cubic, the union of a line and a nonsingular conic, or the union of three lines. 

Assume that $\kc$ is the union of a nonsingular plane conic and a line. Let $f^{(2)} = LN'+L'N$ and $f^{(3)}=LM$ be as described in the previous paragraph. If $f^{(1)}$ is irreducible, then the scheme defined by $f^{(1)}$ and $f^{(2)}$ is a nonsingular twisted cubic $\kc'$. Thus, $\kz(f)$ cannot have a positive dimensional component because $\kc$ and $\kc'$ do not share any common irreducible component. Therefore, we may assume that $f^{(1)}$ is reducible. Let $M'$ be a linear form in $\pi^{(2)}$ and $N''$ a linear form in $\pi^{(3)}$ such that $f^{(1)} = M'N''$. If $\{M,M'\}$ and $\{N,N''\}$ are linearly independent, then $\kz(f)$ consists of two points; one is defined by $M'$, $f^{(2)}$, and $f^{(3)}$, and the other is defined by $N''$, $f^{(2)}$, and $f^{(3)}$. If $M'$ is a multiple of $M$ but $N$ and $N''$ are linearly independent, then $\kz(f)$ is the nonsingular plane conic defined by $M$ and $f^{(2)}$. Similarly, if $N''$ is a multiple of $N$ but $M$ and $M'$ are linearly independent, then $\kz(f)$ is the line defined by $L$ and $N$. In particular, if $\{M,M'\}$ and $\{N,N''\}$ are both linearly dependent, then $\kz(f)$ is the union of three lines defined by $M=N=0$, $L=N=0$, and $L'=M=0$, respectively.

If $\kc$ is the union of three lines, then $f^{(2)} = LN$ and $f^{(3)}=L'M$, as was described in a previous paragraph. If $f^{(1)}$ is irreducible and $\dim\kz(f)>0$, then $f^{(1)}=MN'+M'N$ with $\{M, M'\}$ and $\{N, N'\}$ both linearly independent. Therefore,  the zero scheme $\kz(f)$ of $f$ is the line defined by $M=N=0$, and hence we may assume that $f^{(1)}$ is reducible. Let $M'$ be a linear form in $\pi^{(2)}$ and $N'$ a linear form in $\pi^{(3)}$ such that $f^{(1)} = M'N'$. If $\{M, M'\}$ and $\{N, N'\}$ are both linearly independent, then $\kz(f)$ consists of the point defined by $M'$, $N$, and $L'$ and the point defined by $N'$, $L$, and $M$. If $M'$ is a multiple of $M$ but $N$ and $N'$ are linearly independent, then $\kz(f)$ is a singular conic defined by $LN$ and $M$. Similarly, if $M$ and $M'$ are linearly independent, and if $N'$ is a multiple of $N$, then $\kz(f)$ is the singular conic defined by $L'M$ and $N$. In particular, if $\{M,M'\}$ and $\{N,N'\}$ are both linearly dependent, then $\kz(f)$ is the union of three lines defined by $M=N=0$, $L=M=0$, and $L'=N=0$, respectively.
\end{proof}

\begin{proposition}\label{prop:surface_component}
If the zero scheme of a global section of $E$ contains a surface as a component, then the zero scheme of $f$ is isomorphic to a rational normal scroll of degree $4$, the union of two quadric surfaces, or the union of a quadric surface and a line intersecting in a point. 
\end{proposition}
\begin{proof}
If the zero scheme $\kz(f)$ of a global section $f$ of $E$ has a surface $\ky$ as a component, then there exists a triple $\ba$ such that $\ko(\ba)$ is the line bundle associated with $\ky$. Furthermore, the global section $f$ of $E$ is the image of a nonzero global section of $H^0(\PP^\bd,E(-\ba))$ under the multiplication map by a nonzero element of $H^0(\PP^\bd,\ko(\ba))$. 

By assumption, $H^0(\PP^\bd,E(-\ba)) \not= 0$. Thus, either (1) two of the integers in $\ba$ are $1$, and the remaining one is $0$, or (2) one of the integers in $\ba$ is $1$, and the remaining two are both $0$. 

(1) Without loss of generality, we may assume that $\ba = (0,1,1)$. The surface $\ky$ is defined by a trihomogeneous polynomial of tridegree $(0,1,1)$. If the trihomogeneous polynomial is irreducible, then $\ky$ is a nonsingular surface of degree $4$ in the Segre embedding, which is either a Veronese surface or a rational normal scroll \cite[p. 525]{griffiths1994principles}. However, since $\ky$ contains lines, it must be a rational normal scroll. If the trihomogeneous polynomial is reducible, then it is the product of two linear forms, one in $\pi^{(2)}$ and the other in $\pi^{(3)}$. Thus, the surface $\ky$ is the union of two subvarieties of $\PP^\bd$, each of which is the product of the first factor of $\PP^\bd$ and a line from one of the two rulings of the biprojective space of the second and third factors of $\PP^\bd$. Thus, it is the union of two quadric surfaces in the Segre embedding. 

(2) Without loss of generality, we may assume that $\ba = (0,0,1)$. The zero scheme $\kz(f)$ is defined by two trihomogeneous polynomials of tridegree $(1,0,1)$ and tridegree~$(0,1,1)$, which share a linear form in $\pi^{(3)}$ as a common factor, and the surface component $\ky$ of $\kz(f)$ is defined by this linear factor (the remaining component of $\kz(f)$ is the line defined by the other linear factors of the two trihomogeneous polynomials). Thus, the surface~$\ky$ is the biprojective space consisting of the first two factors of $\PP^\bd$, and hence it is a quadric surface in the Segre embedding. Furthermore, the zero scheme of $f$ is the union of $\ky$ and the line defined by the remaining linear factors of the two trihomogeneous polynomials. The surface and the line intersect at the point defined by the three linear forms appearing in the decompositions of the trihomogeneous polynomials into linear forms.
\end{proof}

We define the {\em type} of a global section $f$ of $E$ to be the largest component of the zero scheme of $f$. According to Propositions~\ref{prop:curve_component} and~\ref{prop:surface_component}, the possible types of $f$ are a nonsingular twisted cubic (cub), a nonsingular plane conic (con), a line (lin), a singular plane conic (ll), the union of a nonsingular plane conic and a line (cl), the union of three nonplanar nonconcurrent lines (lll), the union of three nonplanar concurrent lines (3l), the union of a quadric surface and a line (ql), a rational normal scroll (scr), and the union of two quadric surfaces (qq). 

If $\star \in \{\mbox{cub, con, lin, ll, cl, lll, 3l, ql, scr, qq}\}$, then we denote by $\Delta^\star$ the Zariski closure of the set of global sections of $E$ whose zero schemes are of type $\star$: 
\[
\Delta^\star \coloneqq \overline{\left\{[f] \in \Delta \mid \text{$f$ is of type $\star$}\right\}}\,. 
\]
Table~\ref{tab: description varieties singular strata} summarizes descriptions of $\Delta^\star$ for each type $\star$, including its dimension, degree, and number of irreducible components. The degrees with an asterisk in Table~\ref{tab: description varieties singular strata} were calculated using \verb|Macaulay2|.

\begingroup
\renewcommand{\arraystretch}{1.2}
\begin{table}[ht]
\centering
\resizebox{\textwidth}{!}{
\begin{tabular}{c|c|c|c|p{7.5cm}}
Type $\star$ & $\dim\Delta^\star$ & \#\{components\} & $\deg\Delta^\star$ & description of $\kz(f)$, $f$ generic on $\Delta^\star$\\ \hhline{=====}
$\mathrm{cub}$ & $8$ & $1$ & $11^*$ & Nonsingular twisted cubic, Proposition~\ref{prop: codim Delta_3(E)} \\\hline
$\mathrm{con}$ & $8$ & $3$ & $12=4\times 3$ & Nonsingular conic, Proposition~\ref{prop:conic_component} \\\hline
$\mathrm{lin}$ & $8$ & $3$ & $24=8\times 3$ & Line, Proposition~\ref{prop:line_component} \\\hline
$\mathrm{ll}$ & $7$ & $3$ & $24=8\times 3$ & Singular plane conic, \cite[Proposition 1]{mathrepo} \\\hline
$\mathrm{cl}$ & $7$ & $3$ & $30=10^*\times 3$ & Space connected cubic, union of a nonsingular conic and a line, \cite[Proposition 2]{mathrepo} \\\hline
$\mathrm{lll}$ & $6$ & $3$ & $30=10^*\times 3$ & Space connected cubic, union of three nonconcurrent lines, \cite[Proposition 3]{mathrepo} \\\hline
$\mathrm{3l}$ & $5$ & $1$ & $14^*$ & Union of three nonplanar, concurrent lines, \cite[Proposition 4]{mathrepo}\\\hline
$\mathrm{ql}$ & $4$ & $3$ & $12=4\times 3$ & Union of a quadric surface and a line,\newline Proposition~\ref{prop:Delta_1_2}(2) \\\hline
$\mathrm{scr}$ & $3$ & $3$ & $3=1\times 3$ & Rational normal scroll, Proposition~\ref{prop:Delta_1_2}(1) \\\hline
$\mathrm{qq}$ & $2$ & $3$ & $6=2\times 3$ & Union of two nonsingular quadrics,\newline 
\cite[Proposition 6]{mathrepo}\\
\end{tabular}
}
\cprotect\caption{Dimensions, number of components, and degrees of $\Delta^\star$.}\label{tab: description varieties singular strata}
\end{table}
\endgroup

For a fixed $f = (f^{(1)},f^{(2)},f^{(3)}) \in H^0(\PP^\bd,E)$, define a linear transformation 
\[
\Phi_f \colon \prod_{i=1}^3 V_i^* \to \bigotimes_{i=1}^3 V_i^*
\]
by $\Phi(L,M,N) =  Lf^{(1)}+Mf^{(2)}+Nf^{(3)}$. 

The linear subspace of the projective space of $\bigotimes_{i=1}^3 V_i^*$, the dual of the image of~$\Phi_f$, coincides with the linear span of the zero scheme of $f$ in the Segre embedding. Therefore, Propositions~\ref{prop:curve_component} and~\ref{prop:surface_component} characterize the elements of $\Delta$ whose zero schemes contain positive dimensional components in terms of the rank of $\Phi_f$ as in Table \ref{tab: components type zero schemes}.

\begingroup
\renewcommand{\arraystretch}{1.2}
\begin{table}[ht]
\centering
\begin{tabular}{c|c|c}
Type $\star$ & Largest component of $\kz(f)$ & $\mathrm{rank}\,\Phi_f$ \\ \hhline{===}
$\mathrm{scr}$ & Rational normal scroll/union of two quadric surfaces & $2$ \\ \hline 
$\mathrm{ql}$ & Quadric surface & $3$ \\ \hline 
$\mathrm{cub}$ & Possibly degenerate twisted cubic & $4$ \\ \hline 
$\mathrm{con}$ & Possibly degenerate plane conic & $5$ \\ \hline 
$\mathrm{lin}$ & Line & $6$
\end{tabular}
\caption{Correspondence between the rank of $\Phi_f$ and the type of the largest dimensional component of $\kz(f)$.}\label{tab: components type zero schemes}
\end{table}
\endgroup

The bulk of the remaining subsection is devoted to verifying Table \ref{tab: description varieties singular strata}.

\begin{proposition}\label{prop:Delta_1_2}
\begin{itemize}
    \item[(1)] The locus $\Delta^{\mathrm{scr}}$ consists of three $3$-planes.
    \item[(2)] The locus $\Delta^{\mathrm{ql}}$ consists of three irreducible components. Each irreducible component is the Segre embedding of $\PP^1 \times \PP^3$, and its dimension and degree are both $4$.  
\end{itemize}
\end{proposition}
\begin{proof}
(1) The proof of Proposition~\ref{prop:surface_component} indicates that the zero scheme of a global section $f$ of $E$ is a rational normal scroll or the union of two quadric surfaces if and only if one of the coordinate functions of $f$ is nonzero, but the rest are zero. Thus, the locus $\Delta^{\mathrm{scr}}$ consists of the $3$-planes $\PP H^0(\PP^\bd,\ko(\bone_1)) = \PP(V_2^* \otimes V_3^*)$, $\PP H^0(\PP^\bd,\ko(\bone_2)) = \PP(V_1^* \otimes V_3^*)$, and $\PP H^0(\PP^\bd,\ko(\bone_3)) = \PP(V_1^* \otimes V_2^*)$.  

(2) The proof of Proposition~\ref{prop:surface_component} implies that a global section $f$ of $E$ defines the union of a quadric surface and a line precisely when one of the coordinate functions is zero,  but the rest are products of linear forms and share a linear factor. Hence, there exist three components, each determined by which coordinate function is zero. 
One of the loci is  
\[
\{[(0,LN,LM)] \in \Delta \mid  L \in V_1^*, M \in V_2^*, N \in V_3^*\}\,,
\]
which is the Segre embedding of $\PP(V_1^*) \times \PP(V_2^* \oplus V_3^*)$ in $\PP(V_1^* \otimes (V_2^* \oplus V_3^*))$ defined by sending $([L],[(M,N)])$ to $[(LN,LM)]$, and similarly, one can show that the remaining loci are also the Segre embedding of $\PP^1 \times \PP^3$. 
\end{proof}

\begin{proposition}\label{prop: codim Delta_3(E)}
The locus $\Delta^{\mathrm{cub}}$ is irreducible of dimension $8$. 
\end{proposition}
\begin{proof}
The locus $\Delta^{\mathrm{cub}}$ is the closure of the open subset
\[
    \ku \coloneqq \left\{[f] \in \Delta \mid \text{$\kz(f)$ is a nonsingular twisted cubic}\right\}\,.
\]
We prove that $\dim\,\ku=8$. Define the rational map $\Psi\colon\PP H^0(\PP^\bd,E) \dashrightarrow \PP^7 = \PP((V_1^* \oplus V_3^*)\otimes (V_1^* \oplus V_2^*))$ by $\Psi([f]) = [(f^{(2)},f^{(3)})]$ if~$f = (f^{(1)},f^{(2)},f^{(3)})$. The image $\ku$ under $\Psi$, denoted $\ku'$, is an open subset of $\PP^7$, and hence it has dimension $7$. Thus, it suffices to show that for any point $[f] \in \ku$, the fiber of the morphism obtained from $\Psi$ by restricting to $\ku$ over $\Psi([f])$ has dimension $1$. 

If the element $[f]$ lies in $\ku$, then $\dim\ker\,\Phi_f = 2$ (see Table~\ref{tab: components type zero schemes}), which is equivalent to the condition that the linear syzygies of $f$ are generated precisely by two linear relations of~$f$, say 
\begin{equation}\label{eq: linear_relations}
    Lf^{(1)}+Mf^{(2)}+Nf^{(3)} = L'f^{(1)}+M'f^{(2)}+N'f^{(3)} = 0\,. 
\end{equation}
    
Note that since $f^{(1)}$, $f^{(2)}$, and $f^{(3)}$ are irreducible, if $Lf^{(1)}+Mf^{(2)}+Nf^{(3)} = 0$ is a nontrivial linear relation of $f$, then none of $L$, $M$, and $N$ are zero. 

Let $\Pi$ be the biprojective space 
\[
\left\{\left.\left(\lambda_1 \pi^{(2)}_1+\lambda_2 \pi^{(2)}_2\right)f^{(2)} + \left(\mu_1 \pi^{(3)}_1+\mu_2 \pi^{(3)}_2\right)f^{(3)} \, \right| \, [\lambda_1:\lambda_2] \in \PP^1, [\mu_1:\mu_2] \in \PP^1\right\}\,.
\]
Because of the linear relations~\eqref{eq: linear_relations} of $f$, the pencil of $Lf^{(1)}$ and $L'f^{(1)}$ is the same as that of $Mf^{(2)}+Nf^{(3)}$ and $M'f^{(2)}+N'f^{(3)}$, which it can be interpreted as the diagonal embedding of the pencil of $Lf^{(1)}$ and $L'f^{(1)}$ in $\Pi$. We call the pencil of $Lf^{(1)}$ and $L'f^{(1)}$ the diagonal embedding of the line associated with $f$. 
    
Suppose that $[f'] \in \Psi^{-1}(\Psi([f]))$. Let $(f')^{(1)}$ be the first coordinate function of $f'$. The diagonal embeddings of the lines in $\Pi$ associated with $f$ and $f'$ intersect nontrivially, so they must share elements. Thus, there exist linear forms $L$ and~$L'$ in $\pi^{(1)}$ such that~$Lf^{(1)} = L'(f')^{(1)}$, from which it follows that $f^{(1)}$ and $(f')^{(1)}$ (and hence $L$ and $L'$) differ only by a scalar multiple. Therefore, we have 
\[
    \Psi^{-1}\left(\Psi([f])\right) = \Psi^{-1}\left(\left[\left(f^{(2)},f^{(3)}\right)\right]\right) = \left\{\left.\left(\alpha f^{(1)},f^{(2)},f^{(3)}\right) \, \right| \, \alpha \in \C\right\}\,.
\]
Thus, we proved that $\ku$ has dimension $8$. Furthermore, each fiber of $\Psi$ is irreducible of the same dimension. Therefore, $\ku$, and hence its closure, is irreducible, and hence we completed the proof.
\end{proof}

\begin{remark}
Let $f \in H^0(\PP^\bd,E)$. If $[f]$ is in $\Delta^{\mathrm{scr}}$ or $\Delta^{\mathrm{ql}}$, then the zero scheme of $f$ contains a possibly degenerate twisted cubic, and hence $[f]$ appears as a limit point of~$\Delta^{\mathrm{cub}}$. This implies that $\Delta^{\mathrm{cub}}$ coincides with the set of elements of $\Delta$ such that their associated linear transformations from $\bigoplus_{i=1}^3 V_i^*$ to $\bigotimes_{i=1}^3 V_i^*$ have rank $4$ or less.
Using the representation of the components of $f=(f^{(1)}, f^{(2)}, f^{(3)})$ in \eqref{eq: components global section E 222}, one can show that the matrix of the linear transformation associated with $f$ relative to the basis $\{\pi^{(i)}_1,\pi^{(i)}_2\}_{i \in [3]}$ for $\bigoplus_{i=1}^3 V_i^*$ and the basis $\{\pi^{(1)}_i \otimes \pi^{(2)}_j \otimes \pi^{(3)}_k\}_{i,j,k \in [2]}$ for $\bigotimes_{i=1}^3 V_i^*$ is the transpose of the following $6\times 8$ matrix: 
\[
    \begin{pmatrix}
    a_{11}^{(1)}&0&a_{12}^{(1)}&0&a_{21}^{(1)}&0&a_{22}^{(1)}&0\\[2pt]
    0&a_{11}^{(1)}&0&a_{12}^{(1)}&0&a_{21}^{(1)}&0&a_{22}^{(1)}\\[2pt]
    a_{11}^{(2)}&a_{12}^{(2)}&0&0&a_{21}^{(2)}&a_{22}^{(2)}&0&0\\[2pt]
    0&0&a_{11}^{(2)}&a_{12}^{(2)}&0&0&a_{21}^{(2)}&a_{22}^{(2)}\\[2pt]
    a_{11}^{(3)}&a_{12}^{(3)}&a_{21}^{(3)}&a_{22}^{(3)}&0&0&0&0\\[2pt]
    0&0&0&0&a_{11}^{(3)}&a_{12}^{(3)}&a_{21}^{(3)}&x_{22}^{(3)}
    \end{pmatrix}\,, 
\]
whose $5 \times 5$ minors give rise to the set-theoretic equations for $\Delta^{\mathrm{cub}}$. 
\end{remark}

\begin{proposition}\label{prop:conic_component}
The locus $\Delta^{\mathrm{con}}$ consists of three irreducible components, each with dimension $8$ and degree $4$. 
\end{proposition}
\begin{proof}
Let $f=(f^{(1)},f^{(2)},f^{(3)}) \in H^0(\PP^\bd,E)$. If the zero scheme of $f$ is a nonsingular plane conic, then one of the coordinate functions of $f$ is irreducible, and the others are reducible and share a linear factor. Thus, three different components exist, each corresponding to which coordinate function is irreducible. Since the locus $\Delta^{\mathrm{con}}$ is invariant under the permutation on the three variable sets $\pi^{(1)}$, $\pi^{(2)}$, and $\pi^{(3)}$, it suffices to prove the proposition for one of the components.  

Let $P \coloneqq  \PP (V_2^* \otimes V_3^*) \times \PP(V_1^*) \times \PP(V_2^* \oplus V_3^*)$, and let 
\[
    \ku \coloneqq \left\{\left.\left(\left[f^{(1)}\right],\left[L\right],\left[(M,N)\right]\right) \in P\ \right|\ \text{$f^{(1)}$ is irreducible}\right\}\,. 
\]
Define the map $\Psi\colon P \to \PP H^0(\PP^\bd,E)$ by 
\[
    \Psi\left(\left[f^{(1)}\right],\left[L\right],\left[(M,N)\right]\right) = \left[\left(f^{(1)}, LM, LN\right)\right]\,.
\]
The image of $\ku$ under $\Psi$ forms an open subset of the image of $\Psi$. It also lies in $\Delta^{\mathrm{con}}$, and hence so does its Zariski closure (which coincides with the image of $\Psi)$.  

The map from $\PP(V_1^*) \times \PP(V_2^* \oplus V_3^*)$ to $\PP((V_1^* \otimes V_2^*) \oplus (V_1^* \otimes V_3^*))$ defined by sending $([L], [(M,N)])$ to $[(LM,LN)]$ is the Segre embedding of $\PP(V_1^*)\times \PP(V_2^* \oplus V_3^*)$, and hence the image of $\Psi$ is the cone over this Segre embedding with the vertex $\PP (V_2^* \otimes V_3^*)$. In particular, it is irreducible of dimension $8$ and degree $4$. 
\end{proof}

\begin{remark}\label{rmk: equations of conic component}
For each $i \in [3]$, denote by $\Delta_i^{\mathrm{con}}$ the irreducible component of $\Delta^{\mathrm{con}}$ which contains a global section $f = (f^{(1)}, f^{(2)}, f^{(3)})$ of $E$ such that $f^{(i)}$ is irreducible, but the others are reducible and share a linear factor. This remark concerns a determinantal expression for the equations for $\Delta_i^{\mathrm{con}}$. We discuss only the case where $i=1$ as the remaining cases can be treated similarly. 
    
As was shown in the proof of Proposition~\ref{prop:conic_component}, the locus $\Delta_1^{\mathrm{con}}$ is the cone over the image of the Segre map from $\PP(V_1^*) \times \PP(V_2^* \oplus V_3^*)$ to $\PP(V_1^* \otimes (V_2^* \oplus V_3^*))$ with vertex $\PP(V_2^* \otimes V_3^*)$. Thus, it is enough to find the equations for the image of this Segre map. 

Every element of $V_1^* \otimes (V_2^* \oplus V_3^*)$ can be identified with a pair of an $f^{(3)} \in H^0(\PP^\bd,\ko(\bone_3))$ and an $f^{(2)} \in H^0(\PP^\bd,\ko(\bone_2))$, because 
\[
    V_1^* \otimes (V_2^* \oplus V_3^*) =  (V_1^* \otimes V_2^*) \oplus (V_1^* \otimes V_3^*) = H^0(\PP^\bd,\ko(\bone_3))\oplus H^0(\PP^\bd,\ko(\bone_2))\,.
\]
The element $[(f^{(3)},f^{(2)})] \in \PP (V_1^* \otimes (V_2^* \oplus V_3^*))$ is contained in the image of the Segre map if and only if $(f^{(3)},f^{(2)})$ has rank $1$ when considering it as a linear transformation from $V_1$ to $V_3^* \oplus V_2^*$. The latter is equivalent to the condition that the $2 \times 2$ minors of the matrix representation of the linear transformation relative to given bases for $V_1$ and $V_3^* \oplus V_2^*$ are all zero. More specifically, if $f^{(2)} = \sum_{j,k=0}^1 a_{jk}^{(2)} \pi^{(1)}_j \pi^{(3)}_k$ and let $f^{(3)} = \sum_{j,k=0}^1 a_{jk}^{(3)} \pi^{(1)}_j \pi^{(2)}_k$, then the matrix of the linear transformation $(f^{(3)}, f^{(2)})\colon V_1 \to V_3^* \oplus V_2^*$ relative to $\{e^{(1)}_1,e^{(1)}_2\}$ and $\{\pi^{(3)},\pi^{(2)}\}$ is 
\[
    \begin{pmatrix}
    a^{(3)}_{11} & a^{(3)}_{12} & a^{(2)}_{11} & a^{(2)}_{12} \\[2pt]
    a^{(3)}_{21} & a^{(3)}_{22} & a^{(2)}_{21} & a^{(2)}_{22} 
    \end{pmatrix}\,, 
\]
and hence $\Delta_1^{\mathrm{con}}$ is defined by the $2 \times 2$ minors of this matrix. 
\end{remark}

\begin{proposition}\label{prop:line_component}
The locus $\Delta^{\mathrm{lin}}$ consists of three irreducible components, each having dimension $8$. 
\end{proposition}
\begin{proof}
If the zero scheme of $f=(f^{(1)},f^{(2)},f^{(3)}) \in H^0(\PP^\bd,E)$ is a line, then one of the coordinate functions of $f$ is irreducible, and the others are reducible but do not share any factor. Thus, three different components exist, each corresponding to which coordinate function is irreducible. If it is $f^{(i)}$, then the corresponding component of $\Delta^{\mathrm{lin}}$ is denoted by $\Delta_i^{\mathrm{lin}}$. Since the locus $\Delta^{\mathrm{lin}}$ is invariant under the permutation on the three variable sets $\pi^{(1)}$, $\pi^{(2)}$, and $\pi^{(3)}$, we only prove one of the components is irreducible of dimension $8$.  

Let $\ku$ be the open subset of points of $\Delta_1^{\mathrm{lin}}$ of the form $[(MN'+M'N, LN, L'M)]$ with $L, L' \in V_1^*\setminus \{\bzero\}, M, M' \in V_2^* \setminus \{\bzero\}$, and $N, N' \in V_3^*\setminus \{\bzero\}$, where both $\{M,M'\}$ and $\{N,N'\}$ are linearly independent so that $MN'+M'N$ is irreducible. Define the map $\Psi\colon \ku\to\PP((V_1^* \otimes V_3^*) \oplus (V_1^* \otimes V_2^*))$ by $\Psi([(MN'+M'N, LN, L'M)]) = [(LN,L'M)]$.
The image of this map is the complete intersection of two quadric hypersurfaces; one is the cone over the Segre embedding of~$\PP(V_1^*) \times \PP(V_3^*)$ in $\PP(V_1^* \otimes V_3^*)$ given by sending $([L],[N])$ to $[LN]$, and the other is the cone over the Segre embedding of~$\PP(V_1^*) \times \PP(V_2^*)$ in~$\PP(V_1^* \otimes V_2^*)$ by sending $([L'],[M])$ to $[L'M]$. So, it is irreducible of dimension $5$. The fiber of $\Psi$ over $[(LN,L'M)]$ is 
\[
    \{[(MN'+M'N,LN,L'M)] \mid M' \in V_2^* \setminus \{\bzero\}, N' \in V_3^* \setminus \{\bzero\}\}\,,
\]
and thus, it can be identified with an open subset of the $3$-plane $\{[(M',N')] \mid M' \in V_2^*, N' \in V_3^*\}$, which implies that $\ku$, and hence its closure, is irreducible and has dimension $8$. 
\end{proof}

\begin{remark}\label{rmk: equations of line component}
This remark concerns the equations for $\Delta_i^{\mathrm{lin}}$ for each $i \in [3]$. We discuss only the case where $i=1$ as the remaining cases can similarly be done. 
With the aim of doing so, we use \eqref{eq: components global section E 222} to describe the algebraic relations among the coefficients of the components of $f=(f^{(1)}, f^{(2)}, f^{(3)})\in H^0(\PP^\bd,E)$ such that $\kz(f)$ is a line. 
    
Recall that if the line is determined by $M \in V_2^*$ and $N \in V_3^*$, then there exist $L, L' \in V_1^*$, $M' \in V_2^*$, and~$N' \in V_3^*$ such that $f^{(1)} = MN'+M'N$, $f^{(2)} = LN$, and $f^{(3)} = L'M$. 
    
The condition that $f^{(2)}$ factors into two linear forms is the same as the condition that $[f^{(2)}]$ lies in the Segre embedding of $\PP(V_1^*) \times \PP(V_3^*)$ in $\PP(V_1^* \otimes V_3^*)$. The latter is equivalent to the condition that $f^{(2)}$ has rank one when considering it as a linear transformation from $V_1$ to $V_3^*$. 
Thus, the condition for $f^{(2)}$ to be a product of linear forms is expressed as the vanishing of the determinant of the matrix of this linear transformation relative to the basis $\{e^{(1)}_1,e^{(1)}_2\}$ for $V_1$ and the basis $\{\pi^{(3)}_1,  \pi^{(3)}_2\}$ for~$V_3^*$:
\[
    Q_1 \coloneqq 
    \begin{vmatrix}
    a^{(2)}_{11} & a^{(2)}_{12} \\[2pt]
    a^{(2)}_{22} & a^{(2)}_{22}
    \end{vmatrix}
    = a^{(2)}_{11}a^{(2)}_{22}-a^{(2)}_{12}a^{(2)}_{21} = 0\,. 
\]
Similarly, the polynomial $f^{(3)}$ factors into linear forms if and only if its coefficients satisfy the following algebraic relation:   
\[
    Q_2\coloneqq 
    \begin{vmatrix}
    a^{(3)}_{11} & a^{(3)}_{12} \\[2pt]
    a^{(3)}_{21} & a^{(3)}_{22}
    \end{vmatrix}
    = a^{(3)}_{11}a^{(3)}_{22}-a^{(3)}_{12}a^{(3)}_{21} = 0\,. 
\]
    
The first coordinate function $f^{(1)}$ can be expressed as a linear combination of $M$ and $N$ if and only if $f^{(1)}(p_M,p_N) = 0$, where $p_M \in V_2$ and $p_N \in V_3$ such that $[p_M]$ is the point of $\PP(V_2)$ defined by $M$ and $[p_N]$ the point of $\PP(V_3)$ defined by $N$. The latter can be translated into the system of polynomial equations
\[
    \begin{cases}
    F_{11} \coloneqq f^{(1)}\left(a^{(2)}_{12},-a^{(2)}_{11},a^{(3)}_{12},-a^{(3)}_{11}\right) = 0\\[4pt]
    F_{12} \coloneqq f^{(1)}\left(a^{(2)}_{12},-a^{(2)}_{11},a^{(3)}_{22},-a^{(3)}_{21}\right) = 0\\[4pt]
    F_{21} \coloneqq f^{(1)}\left(a^{(2)}_{22},-a^{(2)}_{21},a^{(3)}_{12},-a^{(3)}_{11}\right) = 0\\[4pt]
    F_{22} \coloneqq f^{(1)}\left(a^{(2)}_{22},-a^{(2)}_{21},a^{(3)}_{22},-a^{(3)}_{21}\right) = 0\,, 
    \end{cases}
\]
because 
\[
    [N] = 
    \begin{cases} 
    \left[a^{(2)}_{11} \pi^{(2)}_1+a^{(2)}_{12}\pi^{(2)}_2\right] & \text{if $\left(a^{(2)}_{11},a^{(2)}_{12}\right) \neq (0,0)$}\,,\\[4pt]  
    \left[a^{(2)}_{21}\pi^{(2)}_1+a^{(2)}_{22}\pi^{(2)}_2\right] & \text{if $\left(a^{(2)}_{21},a^{(2)}_{22}\right) \neq (0,0)$,}
    \end{cases}
\]
and 
\[
    [M] = 
    \begin{cases}
    \left[a^{(3)}_{11}\pi^{(3)}_1+a^{(3)}_{12}\pi^{(3)}_3\right] & \text{if $\left(a^{(3)}_{11},a^{(3)}_{12}\right) \neq (0,0)$}\,,\\[4pt]
    \left[a^{(3)}_{21}\pi^{(3)}_1+a^{(3)}_{22}\pi^{(3)}_2\right]  & \text{if $\left(a^{(3)}_{21},a^{(3)}_{22}\right) \neq (0,0)$.}
    \end{cases}
\]
Therefore, the polynomials $Q_1$, $Q_2$, $F_{11}$, $F_{12}$, $F_{21}$, and $F_{22}$ cut out $\Delta_1^{\mathrm{lin}}$ set-theoretically.

For each $i,j \in \{1,2\}$, the subscheme $\kz_{ij}$ defined by $Q_1$, $Q_2$, and $F_{ij}$ has dimension~$8$. This means that it contains~$\Delta_1^{\mathrm{lin}}$ as an irreducible component. Furthermore, the degree of $F_{ij}$ is $3$ by definition, and hence $\kz_{ij}$ has degree $12$.
The locus $\kz_{ij}$ contains the union of two irreducible subschemes of degree~$2$; one defined by~$\langle Q_1, Q_2, a^{(2)}_{i1}, a^{(2)}_{j2} \rangle = \langle Q_2, a^{(2)}_{i1}, a^{(2)}_{j2}\rangle$, and the other defined by $\langle Q_1, a^{(3)}_{i1}, a^{(3)}_{j2} \rangle = \langle Q_2, a^{(3)}_{i1}, a^{(3)}_{j2}\rangle$. 
Straightforward calculations show that $\kz_{ij}$ agree on the complement of the closed subsets given by $(a^{(j)}_{i1},a^{(k)}_{j2}) = (0,0)$, $k \in \{2,3\}$. This implies that the residual scheme to $\Delta_1^{\mathrm{lin}}$ in $\kz_{ij}$ is the union of the subschemes of degree $2$, and hence the degree of $\Delta_1^{\mathrm{lin}}$ is $8 = 12-2 \cdot 2$.  
\end{remark}

Let $\mathrm{Sing} \, \Delta$ be the reduced structure on the singular locus of $\Delta$, and define $\mathrm{Sing}^{(i)}\Delta$ by 
\[
\mathrm{Sing}^{(i)}\Delta \coloneqq 
\begin{cases}
\mathrm{Sing}\,\Delta & \text{if $i=1$,} \\
\mathrm{Sing}\,(\mathrm{Sing}^{(i-1)}\Delta) & \text{if $i >1$.} 
\end{cases}
\]
The dimensions and the number of irreducible components of $\mathrm{Sing}^{(i)}\Delta$ are studied in \cite{mathrepo} and summarized in Table~\ref{tab: singular strata}. All the varieties appearing in this table are summarized in Table \ref{tab: description varieties singular strata} and were described in the previous propositions. The identities in the first column were verified using \verb|Macaulay2|.

\begingroup
\renewcommand{\arraystretch}{1.2}
\begin{table}[ht]
\centering
\begin{tabular}{c|c|c|c}
$i$ & $\mathrm{Sing}^{(i)}\Delta$ & $\dim\mathrm{Sing}^{(i)}\Delta$ & \#\{irreducible components\}\\ \hhline{====}
$1$ & $\Delta^{\mathrm{cub}}\cup\Delta^{\mathrm{con}}\cup\Delta^{\mathrm{lin}}$ & $8$ & $7=1+3+3$ \\ \hline 
$2$ & $\Delta^{\mathrm{ll}}\cup\Delta^{\mathrm{cl}}$ & $7$ & $6=3+3$ \\ \hline 
$3$ & $\Delta^{\mathrm{lll}}\cup\Delta^{\mathrm{scr}}$ & $6$ & $6=3+3$ \\ \hline 
$4$ & $\Delta^{\mathrm{3l}}$ & $5$ & $1$ \\ \hline 
$5$ & $\Delta^{\mathrm{qq}}$ & $2$ & $3$ \\ \hline 
\end{tabular}
\cprotect\caption{Singular strata of $\Delta$.}\label{tab: singular strata}
\end{table}
\endgroup

One verifies that the game presented in Example \ref{ex: 222 with infinitely many tmNes} corresponds to a global section $f\in\Delta^{\mathrm{cub}}$. We conclude this section by presenting two games whose associated global sections of $E$ belong to the varieties $\Delta^{\mathrm{con}}$ and $\Delta^{\mathrm{lin}}$.

\begin{example}\label{ex: 222 game with Nash equilibria on a conic}
We construct a three-player game $X=(X^{(1)},X^{(2)},X^{(3)})$ with real payoff tensors such that its Nash equilibrium scheme is a nonsingular plane conic whose real part contains infinitely many points corresponding to totally mixed Nash equilibria. 

For each $i\in[2]$, let $\alpha_i \coloneqq \alpha_{i1}\pi_1^{(2)} + \alpha_{i2}\pi_2^{(2)}$ and $\beta_i \coloneqq \beta_{i1}\pi_1^{(3)} + \beta_{i2}\pi_2^{(3)}$ be linear forms with real coefficients satisfying that $\xi_1 \coloneqq \alpha_{11}\beta_{12}+\alpha_{21}\beta_{22}$, $\xi_2\coloneqq\alpha_{12}\beta_{12}+\alpha_{22}\beta_{22}$, $\eta_1 \coloneqq \alpha_{11}\beta_{11}+\alpha_{21}\beta_{21}$, and $\eta_2\coloneqq\alpha_{12}\beta_{11}+\alpha_{22}\beta_{21}$ are all positive, and let $\gamma \coloneqq \gamma_1\pi_1^{(1)} + \gamma_2\pi_2^{(1)}$ be a linear form with real coefficients satisfying $\gamma_1<0$ and $\gamma_2>0$. Define the global section $f$ of $E$ by 
\[
f = (\alpha_1\beta_1+\alpha_2\beta_2,\alpha_1\gamma,\beta_1\gamma)\,. 
\]
By Proposition~\ref{prop:curve_component}, the zero scheme of $f$ is a nonsingular plane conic. 

If $\pi^{(1)}=\left(\gamma_2/(\gamma_2-\gamma_1),-\gamma_1/(\gamma_2-\gamma_1)\right) \in \Delta_1^\circ$, if $\pi^{(2)}\in\Delta_1^\circ$ satisfies $\alpha_1\beta_1+\alpha_2\beta_2~=~0$, and if 
\[
\pi^{(3)} = \left(\frac{\xi_1\pi_1^{(2)}+\xi_2\pi_2^{(2)}}{(\xi_1+\eta_1)\pi_1^{(2)}+(\xi_2+\eta_2)\pi_2^{(2)}},\frac{\eta_1\pi_1^{(2)}+\eta_2\pi_2^{(2)}}{(\xi_1+\eta_1)\pi_1^{(2)}+(\xi_2+\eta_2)\pi_2^{(2)}}\right)\in\Delta_1^\circ\,,
\]
then straightforward calculations show that $\bpi=(\pi^{(1)},\pi^{(2)},\pi^{(3)})\in (\Delta_1^\circ)^3$ are real solutions to~$f=0$.

Note that $f$ is obtained from the game $X$ with payoff tensors
\begingroup
\setlength{\arraycolsep}{2.5pt}
\begin{align*}
X^{(1)} =  
\begin{bmatrix}[cc:cc]
\xi_3 & \xi_1 & 0 & 0\\
\xi_4 & \xi_2 & 0 & 0
\end{bmatrix},\ 
X^{(2)} =
\begin{bmatrix}[cc:cc]
\beta_{11}\gamma_1 & \beta_{12}\gamma_1 & \beta_{11}\gamma_2 & \beta_{12}\gamma_2\\
0 & 0 & 0 & 0
\end{bmatrix},\ 
X^{(3)} = 
\begin{bmatrix}[cc:cc]
\alpha_{11}\gamma_1 & 0 & \alpha_{11}\gamma_2 & 0\\
\alpha_{12}\gamma_1 & 0 & \alpha_{12}\gamma_2 & 0
\end{bmatrix}.
\end{align*}
\endgroup
If we specify the coefficients of $\alpha_1$, $\alpha_2$, $\beta_1$, $\beta_2$, and $\gamma$ as follows:  
$(\alpha_{11},\alpha_{12},\alpha_{21},\alpha_{22}) = (1,1,2,3)$, $(\beta_{11},\beta_{12},\beta_{21},\beta_{22}) = (1,1,-1,1)$, $(\gamma_1,\gamma_2) = (-1,1)$,  
then the corresponding payoff tables are
\begingroup
\setlength{\arraycolsep}{4pt}
\[
X^{(1)} =  
\begin{bmatrix}[cc:cc]
-1 & 3 & 0 & 0\\
-2 & 4 & 0 & 0
\end{bmatrix}\,,\ 
X^{(2)} =
\begin{bmatrix}[cc:cc]
-1 & -1 & 1 & 1\\
0 & 0 & 0 & 0
\end{bmatrix}\,,\ 
X^{(3)} = 
\begin{bmatrix}[cc:cc]
-1 & 0 & 1 & 0\\
-1 & 0 & 1 & 0
\end{bmatrix}\,.
\]
\endgroup
The set of totally mixed Nash equilibria of $X$ is
\[
    \left\{\left(\left(\frac{1}{2},\frac{1}{2}\right),\pi^{(2)},\left(\frac{3\pi_1^{(2)}+4\pi_2^{(2)}}{4\pi_1^{(2)}+6\pi_2^{(2)}},\frac{\pi_1^{(2)}+2\pi_2^{(2)}}{4\pi_1^{(2)}+6\pi_2^{(2)}}\right)\right)\ \bigg|\ \pi^{(2)}\in\Delta_1^\circ\right\}\,,
\]
which is interpreted as the intersection of the locus defined by
\[
    \pi_1^{(2)}\pi_1^{(3)}+3\pi_1^{(2)}\pi_2^{(3)}-2\pi_2^{(2)}\pi_1^{(3)}+4\pi_2^{(2)}\pi_2^{(3)} = -\pi_1^{(1)}+\pi_2^{(1)} = 0
\]
and the product $(\Delta_1^\circ)^3$ of probability simplices.\hfill$\diamondsuit$
\end{example}

\begin{example}\label{ex: 222 game with Nash equilibria on a line}
We construct a three-player game $X=(X^{(1)},X^{(2)},X^{(3)})$ with real payoff tensors such that its Nash equilibrium scheme is a line whose real part contains infinitely many points corresponding to totally mixed Nash equilibria. 

Let $\alpha_i \coloneqq \alpha_{i1}\pi_1^{(1)} + \alpha_{i2}\pi_2^{(1)}$, $\beta_i \coloneqq \beta_{i1}\pi_1^{(2)} + \beta_{i2}\pi_2^{(2)}$, and  $\gamma_i \coloneqq \gamma_{i1}\pi_1^{(3)} + \gamma_{i2}\pi_2^{(3)}$
for all $i\in[2]$, where~$\alpha_{ij}, \beta_{ij}, \gamma_{ij} \in \R$. Define the global section $f$ of $E$ by 
\[
    f = (\alpha_1\beta_1+\alpha_2\beta_2, \alpha_1\gamma_1, \beta_2\gamma_2)\,,  
\]
which can also be interpreted as the global section of $E$ determined by the game with payoff tensors 
\begin{align*}
X^{(1)} &=  
\begin{bmatrix}[cc:cc]
\beta_{21}\gamma_{21} & \beta_{22}\gamma_{21} & \beta_{21}\gamma_{22} & \beta_{22}\gamma_{22}\\
0 & 0 & 0 & 0
\end{bmatrix},
  \\
X^{(2)} &=
\begin{bmatrix}[cc:cc]
\alpha_{11}\gamma_{11} & 0 & \alpha_{11}\gamma_{12} & 0\\
\alpha_{12}\gamma_{11} & 0 & \alpha_{12}\gamma_{12} & 0
\end{bmatrix}, \\ 
X^{(3)} &= 
\begin{bmatrix}[cc:cc]
\alpha_{11}\beta_{11}+\alpha_{21}\beta_{21} & \alpha_{11}\beta_{12}+\alpha_{21}\beta_{22} & 0 & 0\\
\alpha_{12}\beta_{11}+\alpha_{22}\beta_{21} & \alpha_{12}\beta_{12}+\alpha_{22}\beta_{22} & 0 & 0
\end{bmatrix}\,.
\end{align*}

If $\alpha_{11}<0$ and $\alpha_{12}>0$, then  $\pi^{(1)}=\left(\alpha_{12}/(\alpha_{12}-\alpha_{11}),-\alpha_{11}/(\alpha_{12}-\alpha_{11})\right)\in\Delta_1^\circ$ is a solution to $\alpha_1=0$. Similarly, if $\beta_{21}<0$ and $\beta_{22}>0$, then $\pi^{(2)}=\left(\beta_{22}/(\beta_{22}-\beta_{21}),-\beta_{21}/(\beta_{22}-\beta_{21})\right)\in\Delta_1^\circ$ is a solution to $\beta_2=0$. So,   
\[
    \bpi=\left(\left(\frac{\alpha_{12}}{\alpha_{12}-\alpha_{11}},\frac{-\alpha_{11}}{\alpha_{12}-\alpha_{11}}\right),\left(\frac{\beta_{22}}{\beta_{22}-\beta_{21}},\frac{-\beta_{21}}{\beta_{22}-\beta_{21}}\right),\pi^{(3)}\right)\in(\Delta_1^\circ)^3
\]
is a solution to $f=0$ for every $\pi^{(3)}\in\Delta_1^\circ$, and hence the set of the totally mixed Nash equilibria of $X$ parameterizes a portion of the real line defined by $\alpha_1 = \beta_2 = 0$, which is contained in the real part of $\kz(f)$. 

For example, if $X = (X^{(1)}, X^{(2)}, X^{(3)})$ with the following specific payoff tensors
\[
X^{(1)} =  
\begin{bmatrix}[cc:cc]
-8 & 12 & 10 & -15 \\
 0 & 0 & 0 & 0
\end{bmatrix}\,,\ 
X^{(2)} =
\begin{bmatrix}[cc:cc]
-2 & 0 & 3 & 0\\
2 & 0 & -3 & 0
\end{bmatrix}\,,\ 
X^{(3)} = 
\begin{bmatrix}[cc:cc]
-5 & 5 & 0 & 0\\
-5 & 10 & 0 & 0
\end{bmatrix}\,, 
\]
or if the coefficients of the linear forms $\alpha_i$, $\beta_i$, and $\gamma_i$ are 
$(\alpha_{11},\alpha_{12},\alpha_{21},\alpha_{22}) = (-1,1,2,3)$, $(\beta_{11},\beta_{12},\beta_{21},\beta_{22}) = (1,1,-2,3)$, and $(\gamma_{11},\gamma_{12},\gamma_{21},\gamma_{22}) = (2,-3,4,-5)$, then the set of totally mixed Nash equilibria of $X$ is
\[
    \left\{\left(\left(\frac{1}{2},\frac{1}{2}\right),\left(\frac{3}{5},\frac{2}{5}\right),\pi^{(3)}\right)\ \bigg|\ \pi^{(3)}\in\Delta_1^\circ\right\}\,,
\]
which is interpreted as the intersection of the locus defined by 
\[
    -\pi_1^{(1)}+\pi_2^{(1)} = -2\pi_1^{(2)}+3\pi_2^{(2)} = 0
\]
and the product $(\Delta_1^\circ)^3$ of probability simplices.\hfill$\diamondsuit$
\end{example}

\subsection{The Nash discriminant variety of games of boundary format}\label{subsec: Nash discriminant boundary format}

This subsection concerns games of boundary format (games at the boundary between the balanced and unbalanced games). This subsection aims to study the Nash discriminant variety of such games and to reveal its unexpected properties. To be more specific, we prove that the Nash discriminant variety is reducible. One of the irreducible components is a hypersurface. We also discuss a formula for the degree of this hypersurface component.

\begin{theorem}\label{thm: discriminant of E for boundary format}
Let $\bd=(d_1,\ldots,d_n)$ such that $d_n-1=\sum_{i=1}^{n-1}(d_i-1)$, and let $E$ be the vector bundle on $\PP^\bd$ defined in \eqref{eq: Nash vector bundle}. The discriminant variety of $E$ consists of an irreducible hypersurface $\kd_1$ of degree 
\begin{equation}\label{eq: degree hypersurface component boundary format}
    \frac{(d_n-1)!}{(d_1-1)!\cdots(d_{n-1}-1)!}\left(d_n-1+\sum_{i=1}^{n-1}(d_i-1)(d_n-d_i-1)\right)
\end{equation}
and a variety $\kd_2$ of codimension at least $2$. 
\end{theorem}
\begin{proof}
For each $i \in[n]$ and $\alpha_i \in [d_i]$, let $\ku_{\alpha_i}$ be the open subset of $\PP^{d_i-1}$ defined by $\pi^{(i)}_{\alpha_i} \not=0$, and let $x^{(i)} = (x^{(i)}_1, \ldots, x^{(i)}_{d_i-1})$ be the standard local coordinates on $\ku_{\alpha_i}$ (for example, if $\alpha_i = d_i$, then $x^{(i)}_j = \pi^{(i)}_j/\pi^{(i)}_{d_i}$ for each $j \in [d_i-1]$). For each $\balpha \coloneqq (\alpha_1, \ldots, \alpha_n)$, we write $\ku_{\balpha}$ for the open subset $\prod_{i=1}^n \ku_{\alpha_i}$ of $\PP^\bd$ and $x = (x^{(1)}, \ldots, x^{(n)})$ for its local coordinates. 

The zero scheme $\kz(f)$ of a global section $f$ of $E$ is singular at a point $p \in \PP^\bd$ if and only if for any $\balpha = (\alpha_1, \ldots, \alpha_n) \in \prod_{i=1}^n[d_i]$ such that $p \in \ku_{\balpha}$, the determinant $|\partial {f|}_{\ku_{\balpha}}/\partial x|$ of the Jacobian matrix $\partial {f|}_{\ku_{\balpha}}/\partial x$ of ${f|}_{\ku_{\balpha}}$ vanishes at $p$. We show that $\partial {f|}_{\ku_{\balpha}}/\partial x$ is an anti-upper triangular block matrix whose main antidiagonal blocks consist of two $(d_n-1) \times (d_n-1)$ matrices and that the Jacobian determinant $|\partial {f|}_{\ku_{\balpha}}/\partial x|$ of ${f|}_{\ku_{\balpha}}$ factors into the determinants of these two main antidiagonal blocks, which give rise to the two components of  $\Delta(E)$. 

There exists a unique global section $f^{(i)}$ of $\ko(\bone_i)^{\oplus (d_i-1)}$ for each $i \in [n]$ such that $f = \sum_{i=1}^n f^{(i)}$. Moreover, each direct summand $f^{(i)}$ of $f$ can uniquely be written as the sum of global sections~$f^{(i)}_1, \ldots, f^{(i)}_{d_i-1}$ of~$\ko(\bone_i)$. 

For each $i, j \in [n]$, define the $(d_i-1) \times (d_j-1)$ matrix 
\[
\frac{\partial {f^{(i)}|}_{\ku_{\balpha}}}{\partial x^{(j)}}
\coloneqq 
\left(
\frac{\partial {f^{(i)}_\lambda|}_{\ku_{\balpha}}}{\partial x^{(j)}_\mu}
\right)_{\substack{1\le \lambda \le d_i-1\\[2pt] 1\le \mu \le d_j-1}}\,.   
\]
One can partition $\partial {f|}_{\ku_{\balpha}}/\partial x$ into the $(d_i-1) \times (d_j-1)$ blocks $\partial {f^{(i)}|}_{\ku_{\balpha}}/\partial x^{(j)}$: 
\[
\frac{\partial {f|}_{\ku_{\balpha}}}{\partial x} 
= \left(\frac{\partial {f^{(i)}|}_{\ku_{\balpha}}}{\partial x^{(j)}}\right)_{1\le i,j \le n}\,. 
\]
Since $\partial {f^{(n)}|}_{\ku_{\balpha}}/\partial x^{(n)} = 0$, the Jacobian matrix $\partial {f|}_{\ku_{\balpha}}/\partial x$ can be written as
\begin{equation}\label{eq: block jacobian}
\frac{\partial {f|}_{\ku_{\balpha}}}{\partial x} 
=
\begin{pmatrix}
     \frac{\partial ({f^{(1)}|}_{\ku_{\balpha}},\ldots, {f^{(n-1)}|}_{\ku_{\balpha}})}{\partial (x^{(1)}, \ldots, x^{(n-1)})} & \frac{\partial ({f^{(1)}|}_{\ku_{\balpha}},\ldots, {f^{(n-1)}|}_{\ku_{\balpha}})}{\partial x^{(n)}} \\[8pt]
     \frac{\partial {f^{(n)}|}_{\ku_{\balpha}}}{\partial (x^{(1)}, \ldots, x^{(n-1)})} & 0
\end{pmatrix}\,,
\end{equation}
where each block has size $(d_n-1) \times (d_n-1)$.
Hence, the Jacobian determinant $|\partial {f|}_{\ku_{\balpha}}/\partial x|$ of ${f|}_{\ku_{\balpha}}$ is, up to plus-minus sign, equal to the product of the determinants of the top-right and bottom-left blocks in \eqref{eq: block jacobian}:
\[
    \left|\frac{\partial {f|}_{\ku_{\balpha}}}{\partial x}\right| = \pm\left|\frac{\partial ({f^{(1)}|}_{\ku_{\balpha}},\ldots, {f^{(n-1)}|}_{\ku_{\balpha}})}{\partial x^{(n)}}\right|\cdot\left|\frac{\partial {f^{(n)}|}_{\ku_{\balpha}}}{\partial (x^{(1)}, \ldots, x^{(n-1)})}\right|\,.
\]

First, consider the component $\kd_1$ of $\Delta(E)$ determined by $|\partial {f^{(n)}|}_{\ku_{\balpha}}/\partial (x^{(1)}, \ldots, x^{(n-1)})|$. If $\bd' \coloneqq (d_1, \ldots, d_{n-1})$ and if $E'$ denotes $ \ko_{\PP^{\bd'}}(\bone)^{\oplus (d_n-1)}$, then $H^0(\PP^\bd,\ko(\bone_n)^{\oplus (d_n-1)}) = H^0(\PP^{\bd'},E')$ and the component $\kd_1$ is the cone over the discriminant locus $\Delta(E')$ of $E'$ whose vertex is~$\PP H^0(\PP^\bd, \bigoplus_{i=1}^{n-1} \ko(\bone_i)^{\oplus (d_i-1)})$. In particular, $\deg\Delta(E') = \deg\kd_1$ and $\codim\,\Delta(E') = \codim\,\kd_1$.
Since $E'$ is the direct sum of $d_n-1$ copies of the very ample (and hence $1$-jet ample) line bundle $\ko_{\PP^{\bd'}}(\bone)$, it is $1$-jet ample. In particular, the vector bundle $E'$ is very ample, as well as $1$-jet spanned. Therefore, it follows from \cite[Corollary 2.7]{abo2022ramification} that $\Delta(E')$ is an irreducible hypersurface of degree 
\[
\deg\Delta(E') = \int_{\PP^{\bd'}}(c_1(\omega_{\PP^{\bd'}}) + c_1(E'))c_{d_n-2}(E')+(d_n-1)c_{d_n-1}(E')\,,  
\]
where $\omega_{\PP^{\bd'}}$ denotes the canonical bundle on $\PP^{\bd'}$. 

If $A(\PP^{\bd'}) = \C[h_1, \ldots, h_{n-1}]/\langle h_1^{d_1}, \ldots, h^{d_{n-1}}_{n-1}\rangle$, then the relevant Chern classes of $\omega_{\PP^{\bd'}}$ and $E'$ are expressed as follows: $c_1(\omega_{\PP^{\bd'}}) = \ko_{\PP^{\bd'}}(-\bd')= -\sum_{i=1}^{n-1}d_ih_i$, $c_1(E') = (d_n-1)\sum_{i=1}^{n-1}h_i$, 
\begin{align*}
\begin{split}
    c_{d_n-2}(E') &= (d_n-1)\left(\sum_{i=1}^{n-1}h_i\right)^{d_n-2}\\
    &= \frac{(d_n-1)!}{(d_1-1)!\cdots(d_{n-1}-1)!}\sum_{i=1}^{n-1}(d_i-1)h_1^{d_1-1}\cdots h_i^{d_i-2}\cdots h_{n-1}^{d_{n-1}-1}\,,\\
    c_{d_n-1}(E') &= \frac{(d_n-1)!}{(d_1-1)!\cdots(d_{n-1}-1)!}h_1^{d_1-1}\cdots h_{n-1}^{d_{n-1}-1}\,.
\end{split}
\end{align*}
Therefore, the degree formula, as given above, implies 
\[
    \deg\Delta(E')  = \frac{(d_n-1)!}{(d_1-1)!\cdots(d_{n-1}-1)!}\left(d_n-1+\sum_{i=1}^{n-1}(d_i-1)(d_n-d_i-1)\right)\,.
\]

Next, suppose that $[f] \in \Delta(E) \setminus \kd_1(E)$. Let $f = (f^{(1)},\ldots,f^{(n)})$ be the decomposition of $f$ into global sections of $\ko(\bone_i)^{\oplus (d_i-1)}$. Consider $f^{(n)}$ as a global section of $E'$. The zero scheme~$\kz(f^{(n)})$ is nonsingular of codimension $d_n-1$, because $E'$ is globally generated and $[f] \not\in \kd_1$ (and hence $[f^{(n)}]\not\in \Delta(E')$). Furthermore, it consists of $c_{d_n-1}(E') = (d_n-1)!/((d_1-1)!\cdots (d_{n-1}-1)!)$ distinct closed points.  

The canonical projection from $\PP^\bd$ to $\PP^{\bd'}$ maps $\kz(f)$ onto $\kz(f^{(n)})$. The fiber of the projection restricted to $\kz(f)$ over a closed point of $\kz(f^{(n)})$ is the linear subspace of $\PP^{d_n-1}$ defined by the $d_n-1$ linear forms in $\pi^{(n)}$ obtained from $f^{(1)}, \ldots, f^{(n-1)}$ by evaluating at the closed point. Therefore, the scheme $\kz(f)$ is the disjoint union of $(d_n-1)!/((d_1-1)!\cdots (d_{n-1}-1)!)$ (possibly different dimensional) linear subspaces. This means that $[f] \in \Delta(E)\setminus \kd_1$ if and only if $\kz(f)$ has a positive dimensional component. 

Define 
\[
\Gamma(E) \coloneqq \left\{ (p,[f]) \in \PP^\bd \times (\PP H^0(\PP^\bd,E) \setminus \kd_1) \mid \text{$\kz(f)$ is singular at $p$}\right\}\,, 
\]
and denote the projections from $\Gamma(E)$ to $\PP^\bd$ and $\PP H^0(\PP^\bd,E)$ by $\varpi_1$ and $\varpi_2$ respectively. Since $\kz(f)$ is singular at $p\in \PP^\bd$ if and only if for any $\balpha = (\alpha_1, \ldots, \alpha_n) \in \prod_{i=1}^n [d_i]$ with~$p \in \ku_{\balpha}$,  
\[
({f^{(1)}|}_{\ku_{\balpha}})(p) = \bzero, \, \ldots,  ({f^{(n)}|}_{\ku_{\balpha}})(p) = \bzero, \, \left|
\frac{\partial ({f^{(1)}|}_{\ku_{\balpha}},\ldots, {f^{(n-1)}|}_{\ku_{\balpha}})}{\partial x}(p) \right| = 0\,. 
\]
Thus, the dimension of the fiber of $\varpi_1$ over $p$ is bounded above by $\dim \PP H^0(\PP^\bd,E) - (\dim \PP^\bd+1)$, and hence, for a generic $p \in \PP^\bd$, 
\begin{align*}
\dim \Gamma(E) &= \dim \PP^\bd - \dim \varpi_1^{-1}(p)  \\
&\le \dim \PP^\bd +(\dim \PP H^0(\PP^\bd,E) - (\dim \PP^\bd+1)) \\
&= \dim \PP H^0(\PP^\bd,E) - 1\,. 
\end{align*}
As was shown above, if $[f] \in \Delta(E)\setminus 
\kd_1$, then $\dim \kz(f) \ge 1$. Therefore, 
\[
\dim \Delta(E) \setminus \kd_1 \le \dim \Gamma(E) - \dim \varpi_2^{-1}([f]) 
= \dim \Gamma(E) - \dim \kz(f) \le \dim \PP H^0(\PP^\bd,E) - 2\,. 
\]
Therefore, the codimension of the second component $\kd_2 \coloneqq \overline{\Delta(E)\setminus \kd_1}$ of $\Delta(E)$ is greater than or equal to $2$.
\end{proof}

\begin{example}\label{eq: description of Delta E for 223}
We illustrate the idea of the proof of Theorem~\ref{thm: discriminant of E for boundary format} with a specific format $\bd$. Let $\bd = (2,2,3)$, let $E = \ko(0,1,1)\oplus\ko(1,0,1)\oplus\ko(1,1,0)^{\oplus 2}$, and let $f=(f^{(1)},f^{(2)},f_1^{(3)},f_2^{(3)})$ be a global section of $E$, where
\[
\begin{array}{ll}
    f^{(1)} = \sum_{j=1}^2\sum_{k=1}^3 a_{jk}^{(1)}\pi_j^{(2)}\pi_k^{(3)}\,, &  f^{(2)} = \sum_{i=1}^2\sum_{k=1}^3 a_{ik}^{(2)}\pi_i^{(1)}\pi_k^{(3)}\,,\\[5pt]
    f_1^{(3)} = \sum_{i=1}^2\sum_{j=1}^2 a_{ij}^{(3,1)}\pi_i^{(1)}\pi_j^{(2)}\,, & f_2^{(3)} = \sum_{i=1}^2\sum_{j=1}^2 a_{ij}^{(3,2)}\pi_i^{(1)}\pi_j^{(2)}\,.
\end{array}
\]
If $\balpha = (2,2,3)$ and if $(x^{(1)}_1, x^{(2)}_1, x^{(3)}_1,x^{(3)}_2)$ is the vector of local coordinates on $\ku_{\balpha}$, then the Jacobian matrix of $f$ on $\ku_{\balpha}$ is 
\[
\frac{\partial {f|}_{\ku_{\balpha}}}{\partial x} 
=
\left(
\begin{smallmatrix}
     0 & a_{11}^{(1)}x_1^{(3)}+a_{12}^{(1)}x_2^{(3)}+a_{13}^{(1)} & a_{11}^{(1)}x_1^{(2)}+a_{21}^{(1)} & a_{12}^{(1)}x_1^{(2)}+a_{22}^{(1)} \\[2pt]
     a_{11}^{(2)}x_1^{(3)}+a_{12}^{(2)}x_2^{(3)}+a_{13}^{(2)} & 0 & a_{11}^{(2)}x_1^{(1)}+a_{21}^{(1)} & a_{12}^{(2)}x_1^{(1)}+a_{22}^{(1)} \\[2pt]
     a_{11}^{(3,1)}x_1^{(2)}+a_{12}^{(3,1)} & a_{11}^{(3,1)}x_1^{(1)}+a_{21}^{(3,1)} & 0 & 0 \\[2pt]
     a_{11}^{(3,2)}x_1^{(2)}+a_{12}^{(3,2)} & a_{11}^{(3,2)}x_1^{(1)}+a_{21}^{(3,2)} & 0 & 0
\end{smallmatrix}
\right)\,.
\]

On one hand, eliminating the variables $x_1^{(1)}$ and $x_1^{(2)}$ from the ideal
\[
\left\langle {f_1^{(3)}|}_{\ku_{\balpha}}, {f_2^{(3)}|}_{\ku_{\balpha}},\left|
\begin{smallmatrix}
    a_{11}^{(3,1)}x_1^{(2)}+a_{12}^{(3,1)} & a_{11}^{(3,1)}x_1^{(1)}+a_{21}^{(3,1)} \\[2pt]
    a_{11}^{(3,2)}x_1^{(2)}+a_{12}^{(3,2)} & a_{11}^{(3,2)}x_1^{(1)}+a_{21}^{(3,2)}
\end{smallmatrix}
\right|\right\rangle\,,
\]
gives rise to the equation of the first component $\kd_1$ of $\Delta(E)$, which is an irreducible hypersurface of degree $4$. It coincides with the discriminant of the system of bilinear forms $f_1^{(3)}=f_2^{(3)}=0$ in $\pi^{(1)}$ and $\pi^{(2)}$. On the other hand, eliminating the variables $x_1^{(1)},x_1^{(2)},x_1^{(3)}$, and $x_2^{(3)}$ from the ideal
\[
\left\langle {f^{(1)}|}_{\ku_{\balpha}}, {f^{(2)}|}_{\ku_{\balpha}}, {f_1^{(3)}|}_{\ku_{\balpha}}, {f_2^{(3)}|}_{\ku_{\balpha}}, \left|
\begin{smallmatrix}
    a_{11}^{(1)}x_1^{(2)}+a_{21}^{(1)} & a_{12}^{(1)}x_1^{(2)}+a_{22}^{(1)} \\[2pt]
    a_{11}^{(2)}x_1^{(1)}+a_{21}^{(1)} & a_{12}^{(2)}x_1^{(1)}+a_{22}^{(1)}
\end{smallmatrix}
\right|\right\rangle\,,
\]
we obtain the generators of the ideal of the second component $\kd_2$ of $\Delta(E)$. With the aid of \verb|Macaulay2|, one can verify that $\codim\,\kd_2=2$ and $\deg\kd_2=19$. One can also check that $\kd_2$ is irreducible over the field $\Q$ and its radical ideal is minimally generated by $11$ homogeneous polynomials in  $a_{ij}^{(1)},a_{ij}^{(2)},a_{ij}^{(3,1)}$, and $a_{ij}^{(3,2)}$; three of them have degree $6$, two of them have degree $7$, and the rest of them have degree $8$.\hfill$\diamondsuit$
\end{example}

\begin{example}\label{ex: 223 with one double tmNe}
Consider a three-player game $X=(X^{(1)},X^{(2)},X^{(3)})$ of format $\bd=(2,2,3)$ whose payoff tensors are 
\[
X^{(1)} =  
\begin{bmatrix}[cc:cc:cc]
2 & 2 & 1 & 2 & 3 & 0 \\
1 & 3 & 2 & 3 & 2 & 2 
\end{bmatrix},\ 
X^{(2)} =
\begin{bmatrix}[cc:cc:cc]
3 & 2 & 2 & 1 & 4 & 2 \\
1 & 4 & 1 & 3 & 2 & 3 
\end{bmatrix},\ 
X^{(3)} = 
\begin{bmatrix}[cc:cc:cc]
3 & 1 & 2 & 2 & 1 & 2 \\
2 & 4 & 3 & 4 & 3 & 8 
\end{bmatrix}\,.
\]
The Nash equilibrium scheme $\kz_X$ of $X$ is defined by the ideal $J = J_1+J_2+J_3$ of Definition \ref{def: Nash equilibria scheme}, where $J_1 = \langle \varDelta f_{1,2}^{(1)}\rangle$, $J_2 = \langle \varDelta f_{1,2}^{(2)}\rangle$, and $J_3 = \langle \varDelta f_{1,3}^{(3)}, \varDelta f_{1,3}^{(3)}\rangle$.
More explicitly, it is defined by the system
\[
\begin{cases}
    0 = \varDelta f_{1,2}^{(1)} = \pi_1^{(2)}\pi_1^{(3)}-\pi_1^{(2)}\pi_2^{(3)}+\pi_1^{(2)}\pi_3^{(3)}-\pi_2^{(2)}\pi_1^{(3)}-\pi_2^{(2)}\pi_2^{(3)}-2\pi_2^{(2)}\pi_3^{(3)}\\ 
    0 = \varDelta f_{1,2}^{(2)} = \pi_1^{(1)}\pi_1^{(3)}+\pi_1^{(1)}\pi_2^{(3)}+2\pi_1^{(1)}\pi_3^{(3)}-3\pi_2^{(1)}\pi_1^{(3)}-2\pi_2^{(1)}\pi_2^{(3)}-\pi_2^{(1)}\pi_3^{(3)}\\
    0 = \varDelta f_{1,2}^{(3)} = \pi_1^{(1)}\pi_1^{(2)}-\pi_1^{(1)}\pi_2^{(2)}-\pi_2^{(1)}\pi_1^{(2)}\\
    0 = \varDelta f_{1,3}^{(3)} = 2\pi_1^{(1)}\pi_1^{(2)}-\pi_1^{(1)}\pi_2^{(2)}-\pi_2^{(1)}\pi_1^{(2)}-4\pi_2^{(1)}\pi_2^{(2)}\,.
\end{cases} 
\]
The point
\[
\bpi = 
\left(
\left(\frac{2}{3},\frac{1}{3}\right),\left(\frac{2}{3},\frac{1}{3}\right),\left(\frac{3}{5},\frac{1}{5},\frac{1}{5}\right)\right)\in \Delta_1 \times \Delta_1 \times \Delta_2
\]
is a solution to the system mentioned above. 
With the aid of \verb|Macaulay2|, we verified that $\kz_X$ is a nonreduced point of multiplicity $2$ supported at $[\bpi] \in \PP^\bd$.\hfill$\diamondsuit$
\end{example}

\subsection{The Nash resultant variety}\label{subsec: Nash resultant}
Theorem~\ref{thm: number tmNe generic game} shows that if $X=(X^{(1)},\ldots,X^{(n)}) \in V^{\oplus n}$ is generic, then $\kz_X=\emptyset$ if and only if the format of $X$ is ``beyond boundary'', i.e., $d_n-1>\sum_{i=1}^{n-1}(d_i-1)$. However, there exist $X  \in V^{\oplus n}$ whose Nash equilibrium schemes are not empty. In this subsection, we study the locus of such $X$. 

\begin{definition}\label{def: Nash resultant variety}
Let $n\ge 2$ and let $\bd=(d_1,\ldots,d_n)\in\Z_{\mge[2]}^n$. If $d_n-1>\sum_{i=1}^{n-1}(d_i-1)$, then we call the subset of $\PP V^{\oplus n}$
\[
\kr(\bd) \coloneqq \left\{[X] \in \PP V^{\oplus n} \mid \kz_X \neq \emptyset\right\}
\]
the \emph{Nash resultant variety}. 
\end{definition}

\begin{proposition}\label{prop: Nash resultant variety two players}
Let $\bd=(d_1,d_2)\in\Z_{\mge[2]}^2$ with $d_1 < d_2$. The Nash resultant variety $\kr(\bd)$ is irreducible of codimension $d_2-d_1$ and degree $\binom{d_2-1}{d_1-1}$. 
\end{proposition}
\begin{proof}
Let $X = (X^{(1)},X^{(2)}) \in V_1\oplus V_2$. As was shown in Proposition~\ref{prop: discriminant_two_players}, the Nash equilibrium scheme~$\kz_X$ of $X$ is defined by two homogeneous systems of linear equations, the second of which has the coefficient matrix 
\[
A \coloneqq 
\begin{pmatrix}
x_{11}^{(2)}-x_{12}^{(2)} & \cdots & x_{d_11}^{(2)}-x_{d_12}^{(2)} \\
\vdots & & \vdots \\
x_{11}^{(2)}-x_{1d_2}^{(2)} & \cdots & x_{d_11}^{(2)}-x_{d_1d_2}^{(2)} 
\end{pmatrix}\,. 
\] 
Since $d_1 < d_2$, this system is overdetermined, and hence $\kz_X\neq\emptyset$ precisely when $\mathrm{rank}\,A\le d_1-1$. In particular, the Nash resultant variety $\kr(\bd)$ is a determinantal variety, and its codimension and degree are $\codim\,\kr(\bd)=d_2-d_1$ and $\deg\kr(\bd)=\binom{d_2-1}{d_1-1}$ respectively.
\end{proof}

The following theorem extends Proposition \ref{prop: Nash resultant variety two players} to a more general tensor of beyond boundary format. 

\begin{theorem}\label{thm: codim degree Nash resultant variety}
Let $n\ge 2$ and let $\bd=(d_1,\ldots,d_n)\in\Z_{\mge[2]}^n$. If $d_n-1>\sum_{i=1}^{n-1}(d_i-1)$, then the Nash resultant variety $\kr(\bd)$ is irreducible, and its codimension and degree are
\[
\codim\,\kr(\bd) = d_n-1-\sum_{i=1}^{n-1}(d_i-1)
\]
and   
\begin{align*}
\deg\kr(\bd) &= \binom{d_n-1}{\codim\,\kr(\bd)}\binom{d_n-\codim\,\kr(\bd)-1}{d_1-1,\ldots,d_{n-1}-1} \\
&= \frac{(d_n-1)!}{(d_1-1)! \cdots (d_{n-1}-1)! (d_n-1-\sum_{i=1}^{n-1}(d_i-1))!}
\end{align*}
respectively. 
\end{theorem}
\begin{proof}
Let $E$ be the vector bundle on $\PP^\bd$ as given in \eqref{eq: Nash vector bundle}. If $\kr(E)$ denotes the resultant variety of $E$, i.e., the set of global sections whose zero schemes are not empty:
\[
\kr(E) = \{[f] \in \PP H^0(\PP^\bd,E) \mid \kz(f)\neq \emptyset\}\,,
\]
then $\kr(\bd)$ is, after a suitable linear change of coordinates, a cone over $\kr(E)$, and hence it suffices to show the irreducibility of $\kr(E)$ as well as to find the codimension and degree of $\kr(E)$. 

Given $f=(f^{(1)},\ldots,f^{(n)}) \in H^0(\PP^\bd,E)$, where $f^{(i)}\in H^0(\PP^\bd,\ko(\bone_i)^{\oplus(d_i-1)})$ for all $i\in[n]$, the proof of Theorem \ref{thm: number tmNe generic game} indicates that $\kz(f) = \emptyset$ if and only if $\kz(f^{(n)}) =\emptyset$. 
If $\bd'\coloneqq(d_1,\ldots,d_{n-1})$ and $F\coloneqq \ko_{\PP^{\bd'}}(\bone)^{\oplus(d_n-1)}$, then $H^0(\PP^\bd,\ko(\bone_n)^{\oplus(d_n-1)}) = H^0(\PP^{\bd'},F)$.
Since $\kr(E)$ is the cone over the resultant variety of $F$ 
\[
\kr(F) = \{[f^{(n)}] \in \PP H^0(\PP^{\bd'}, F) \mid \kz(f^{(n)})\neq \emptyset\}\,, 
\]
it suffices to show that $\kr(F)$ is irreducible and to compute its codimension and degree. These were already studied in \cite[Section 3 in Chapter 3]{GKZ}. For the sake of the reader's convenience, we outline the proofs used in \cite{GKZ} below. 

Let $\Gamma(F)\coloneqq\{(p,[f^{(n)}]) \in \PP^{\bd'} \times \PP H^0(\PP^{\bd'},F) \mid \kz(f^{(n)}) \neq \emptyset \}$, and for each $i \in \{1,2\}$, denote the projection from $\PP^{\bd'} \times \PP H^0(\PP^{\bd'},F)$ to its $i$th factor by $q_i$. The fiber $q_1^{-1}(p)$ of $q_1$ over any point $p$ of $\PP^{\bd'}$ is identified with a linear subspace of $\PP H^0(\PP^{\bd'},F)$ defined by $d_n-1$ linearly independent linear forms obtained by evaluating $f^{(n)}$ at $p$. This means that $q_1^{-1}(p)$ has codimension $d_n-1$, and hence $\Gamma(F)$ has a projective bundle structure with the morphism $q_1\colon \Gamma(E) \to \PP^{\bd'}$. In particular, it is irreducible and has $\dim \PP^{\bd'}+\dim \PP H^0(\PP^{\bd'},F)-(d_n-1) = \dim \PP H^0(\PP^{\bd'},F)-(d_n-1-\sum_{i=1}^{n-1}(d_i-1))$. Furthermore, since $F$ is very ample, if $[f^{(n)}] \in \kr(F)$ is generic, then $Z(f^{(n)})$ consists of a closed point. Therefore, the mprphism $q_2\colon \Gamma(F) \to \kr(F)$ is birational, and hence $\kr(F)$ is irreducible of codimension $d_n-1-\sum_{i=1}^{n-1}(d_i-1)$. 

Regarding the degree of $\kr(F)$, the proof of \cite[Chapter 3, Theorem 3.10]{GKZ} shows that, if $d_n-1 \geq \sum_{i=1}^{n-1}(d_i-1)+1$, then its degree is 
\begin{align*}
\deg\kr(F) &= \int_{\PP^{\bd'}} c_{d_n-c-1}(F) \\
&= \binom{d_n-1}{d_n-c-1}\binom{d_n-c-1}{d_1-1, d_2-1, \ldots, d_{n-1}-1} \\
&= \frac{(d_n-1)!}{(d_1-1)!(d_2-1)! \cdots (d_{n-1}-1)!\,c!}\,, 
\end{align*}
which completes the proof.
\end{proof}

\begin{remark}\label{rmk: determinantal representation Nash resultant boundary format}
Let $\bd = (d_1,\ldots, d_n) \in \Z_{\mge[2]}^n$. 
If $d_n-1 = \sum_{i=1}^{n-1}(d_i-1)+1$, then Theorem \ref{thm: codim degree Nash resultant variety} implies that the Nash resultant variety $\kr(\bd)$ is a hypersurface. We call the polynomial defining~$\kr(\bd)$ the \emph{Nash resultant}. 

As the proof of Theorem~\ref{thm: codim degree Nash resultant variety} suggests, the problem of finding the Nash resultant is equivalent to the problem of finding the resultant of the system of the $(d_n-1)$ multilinear forms 
\[
\varDelta f_{1,k}^{(n)}  =   f_1^{(n)}-f_k^{(n)}=\sum_{\bj_{-n} \in I_{-n}} (x_{(1,\bj_{-n})}^{(n)}-x_{(k,\bj_{-n})}^{(n)}) \, \pi_{\bj_{-n}} = 0\,.
\]
Thus, setting $y_{j_{-n}}\coloneqq x_{(1,\bj_{-n})}^{(n)}-x_{(k,\bj_{-n})}^{(n)}$, one can reduce the problem of finding the Nash resultant to the problem of finding the system of the generic $(d_n-1)$ multilinear forms, or equivalently, the hyperdeterminant of the boundary format $(d_1, \ldots, d_{n-1}, d_n-1)$. Furthermore, it follows from \cite[Chapter 14.3.B]{GKZ} that the Nash resultant can be expressed as the determinant of a matrix of order $(d_n-1)!/\prod_{i=1}^{n-1} (d_i-1)!$. 
To be more precise, for each $i \in [n]$, let $e_i$ be the nonnegative integer defined by 
\[
e_i \coloneqq
\begin{cases}
0 & \text{if $i=1$} \\
\sum_{j=1}^{i-1} (d_j-1) & \text{otherwise,}
\end{cases}
\]
and let $\mathrm{Sym}^{e_i}(V_i^*)$ be the $e_i$th symmetric product of the dual space $V_i^*$ of $V_i$. If $\partial_X$ denotes the linear transformation 
\begin{equation}\label{eq: map partial X}
\partial_X\colon \left(\bigotimes_{i=1}^{n-1} \mathrm{Sym}^{e_i} V_i\right)^{\!\!d_n-1} \longrightarrow \bigotimes_{i=1}^{n-1} \mathrm{Sym}^{e_i+1}V_i
\end{equation}
by $\partial_X(G_1, \ldots, G_{d_n-1}) \coloneqq \sum_{k=1}^{d_n-1} \varDelta f_{1,k}^{(n)}\, G_i$,  
then the determinant of $\partial_X$ is a polynomial defining~$\kr(\bd)$. 
\end{remark}

\begin{example}\label{ex: Nash resultant in the case 224}
We discuss an example to illustrate the procedure of finding the Nash resultant given in Remark \ref{rmk: determinantal representation Nash resultant boundary format} in more detail. 

Let $n=3$, let $\bd=(2,2,4)$, and let $X = (X^{(1)},X^{(2)}, X^{(3)}) \in V^{\oplus 3}$. By Theorem~\ref{thm: codim degree Nash resultant variety}, the Nash resultant variety $\kr(\bd)$ is a hypersurface in $V^{\oplus n}$ of degree $6$. The Nash resultant depends only on the variables of entries of $X^{(3)}$.

Note that $e_1 = 0$ and $e_2 = d_1-1 = 1$, so the map $\partial_X$ in \eqref{eq: map partial X} is the linear transformation from $(V_2^*)^{\oplus 3}$ to $V_1^*\otimes\mathrm{Sym}^2(V_2^*)$ that sends a triple $(G_1,G_2,G_3)$ of linear forms in $\pi_1^{(2)}$ and $\pi_2^{(2)}$ to $\sum_{i=1}^3 \varDelta f_{1,k}^{(3)} \, G_i$, where 
\[
\varDelta f_{1,k}^{(3)} = f_1^{(3)}-f_k^{(3)} = \sum_{i=1}^2 \sum_{j=1}^2 (x_{ij 1}^{(3)}-x_{ij k}^{(3)})\, \pi_{i}^{(1)}\pi_{j}^{(2)}
\]
for each $k \in \{2,3,4\}$. The matrix representation of $\partial_X$ relative to the standard bases for~$V_2^{\oplus 3}$ and $V_1 \otimes \mathrm{Sym}^2 (V_2^*)$ is 
\[
\begin{pmatrix}
x_{111}^{(3)}-x_{112}^{(3)} & 0 & x_{111}^{(3)}-x_{113}^{(3)} & 0 & x_{111}^{(3)}-x_{114}^{(3)}  & 0  \\[2pt]
x_{121}^{(3)}-x_{122}^{(3)} & x_{111}^{(3)}-x_{112}^{(3)} & x_{121}^{(3)}-x_{123}^{(3)} & x_{111}^{(3)}-x_{113}^{(3)} & x_{121}^{(3)}-x_{124}^{(3)} & x_{121}^{(3)}-x_{114}^{(3)} \\[2pt]
0 & x_{121}^{(3)}-x_{122}^{(3)} & 0 & x_{121}^{(3)}-x_{123}^{(3)} & 0 & x_{121}^{(3)}-x_{124}^{(3)} \\[2pt]
x_{211}^{(3)}-x_{212}^{(3)} & 0 & x_{211}^{(3)}-x_{213}^{(3)} & 0 & x_{211}^{(3)}-x_{214}^{(3)}  & 0 \\[2pt]
x_{221}^{(3)}-x_{222}^{(3)} & x_{211}^{(3)}-x_{212}^{(3)} & x_{221}^{(3)}-x_{223}^{(3)} & x_{211}^{(3)}-x_{213}^{(3)} & x_{221}^{(3)}-x_{224}^{(3)} & x_{211}^{(3)}-x_{214}^{(3)} \\[2pt]
0 & x_{221}^{(3)}-x_{222}^{(3)} & 0 & x_{221}^{(3)}-x_{223}^{(3)} & 0 & x_{221}^{(3)}-x_{224}^{(3)}  
\end{pmatrix}\,.
\]
The Nash resultant hypersurface is defined by the determinant of this $6 \times 6$ matrix. This polynomial has 960 nonzero terms in the 16 entries of $X^{(3)}$. \hfill$\diamondsuit$
\end{example}

\section*{Acknowledgements}

This work is partially supported by the Thematic Research Programme {\em``Tensors: geometry, complexity and quantum entanglement''}, University of Warsaw, Excellence Initiative - Research University and the Simons Foundation Award No. 663281 granted to the Institute of Mathematics of the Polish Academy of Sciences for the years 2021-2023. We are thankful for the support and excellent working conditions during the semester ``Algebraic Geometry with Applications to Tensors and Secants'' (AGATES). HA would like to express gratitude to the Max Planck Institute for Mathematics in the Sciences for its generous hospitality and inspiring research environment during his visit. IP expresses her sincere gratitude to KTH Royal Institute of Technology in Stockholm for its hospitality and stimulating discussions during her visit. The work of LS was partially supported by a KTH grant from the Verg Foundation and Brummer \& Partners MathDataLab. We are grateful to Antonio Lerario for fruitful discussions on Thom's Isotopy Lemma. The authors also thank the referee and Claus Hertling for carefully reading the paper and for their valuable suggestions and comments.

\bibliographystyle{alpha}
\bibliography{biblio}

\begin{thebibliography}{LGMSY24}

\bibitem[Abo20]{abo2020discriminant}
H.~Abo.
\newblock On the discriminant locus of a rank {$n-1$} vector bundle on
  $\mathbb{P}^{n-1}$.
\newblock {\em Port. Math.}, 77(3-4):299--343, 2020.

\bibitem[ACGH85]{arbarello1985geometry}
E.~Arbarello, M.~Cornalba, P.~A. Griffiths, and J.~Harris.
\newblock {\em Geometry of algebraic curves. {V}ol. {I}}, volume 267 of {\em
  Grundlehren der mathematischen Wissenschaften}.
\newblock Springer-Verlag, New York, 1985.

\bibitem[ALS22]{abo2022ramification}
H.~Abo, R.~Lazarsfeld, and G.~G. Smith.
\newblock Ramification and discriminants of vector bundles and a quick proof of
  {B}ogomolov's theorem.
\newblock {\em Rend. Istit. Mat. Univ. Trieste}, 54:Art. No. 6, 15, 2022.

\bibitem[APS25]{mathrepo}
H.~Abo, I.~Portakal, and L.~Sodomaco.
\newblock Mathrepo - a vector bundle approach to nash equilibria, 2025.
\newblock \url{https://mathrepo.mis.mpg.de/vector-bundle-nash-equilibria/},
  Accessed: 05.03.2025.

\bibitem[Ber75]{bernstein1975number}
D.~N. Bernstein.
\newblock The number of roots of a system of equations.
\newblock {\em Funkcional. Anal. i Prilo\v zen.}, 9(3):1--4, 1975.

\bibitem[BHP24]{brandenburg2024combinatorics}
M.-C. Brandenburg, B.~Hollering, and I.~Portakal.
\newblock Combinatorics of correlated equilibria.
\newblock {\em Exp. Math.}, pages 1--21, 2024.

\bibitem[BKS50]{bohnenblust1950solutions}
H.~F. Bohnenblust, S.~Karlin, and L.~S. Shapley.
\newblock Solutions of discrete, two-person games.
\newblock In {\em Contributions to the {T}heory of {G}ames}, volume no. 24 of
  {\em Ann. of Math. Stud.}, pages 51--72. Princeton Univ. Press, Princeton,
  NJ, 1950.

\bibitem[Bub79]{bubelis1979equilibria}
V.~Bubelis.
\newblock On equilibria in finite games.
\newblock {\em Internat. J. Game Theory}, 8(2):65--79, 1979.

\bibitem[CPR74]{chin1974structure}
H.~H. Chin, T.~Parthasarathy, and T.~E.~S. Raghavan.
\newblock Structure of equilibria in {$N$}-person non-cooperative games.
\newblock {\em Internat. J. Game Theory}, 3:1--19, 1974.

\bibitem[CS08]{conitzer2008complexity}
V.~Conitzer and T.~Sandholm.
\newblock New complexity results about {N}ash equilibria.
\newblock {\em Games Econom. Behav.}, 63(2):621--641, 2008.

\bibitem[Dat03]{datta2003universality}
R.~S. Datta.
\newblock Universality of {N}ash equilibria.
\newblock {\em Math. Oper. Res.}, 28(3):424--432, 2003.

\bibitem[DGP09]{daskalakis2009complexity}
Constantinos Daskalakis, Paul~W. Goldberg, and Christos~H. Papadimitriou.
\newblock The complexity of computing a {N}ash equilibrium.
\newblock {\em SIAM J. Comput.}, 39(1):195--259, 2009.

\bibitem[Ein82]{ein1982some}
L.~Ein.
\newblock Some stable vector bundles on {$\PP^4$} and {$\PP^5$}.
\newblock {\em J. Reine Angew. Math.}, 337:142--153, 1982.

\bibitem[Emi16]{emiris2016compact}
I.~Z. Emiris.
\newblock Compact formulae in sparse elimination [extended abstract].
\newblock In {\em Proceedings of the 2016 {ACM} {I}nternational {S}ymposium on
  {S}ymbolic and {A}lgebraic {C}omputation}, pages 1--4. ACM, New York, 2016.

\bibitem[EV16]{emiris2016discriminants}
I.~Z. Emiris and R.~Vidunas.
\newblock Discriminants of multilinear systems.
\newblock {\em \arxiv{1607.01496}}, 2016.

\bibitem[EZ16]{ekhad2016number}
S.~B. Ekhad and D.~Zeilberger.
\newblock On the number of singular vector tuples of hyper-cubical tensors.
\newblock {\em The Personal Journal of Shalosh B. Ekhad and Doron Zeilberger},
  2016.

\bibitem[FO14]{FO}
S.~Friedland and G.~Ottaviani.
\newblock The number of singular vector tuples and uniqueness of best rank-one
  approximation of tensors.
\newblock {\em Found. Comput. Math.}, 14(6):1209--1242, 2014.

\bibitem[Ful98]{fulton1998intersection}
W.~Fulton.
\newblock {\em Intersection theory}, volume~2 of {\em Ergebnisse der Mathematik
  und ihrer Grenzgebiete. 3. Folge.}
\newblock Springer-Verlag, Berlin, second edition, 1998.

\bibitem[Gee02]{geertsen2002degeneracy}
J.~A. Geertsen.
\newblock Degeneracy loci of vector bundle maps and ampleness.
\newblock {\em Math. Scand.}, 90(1):13--34, 2002.

\bibitem[GH94]{griffiths1994principles}
P.~Griffiths and J.~Harris.
\newblock {\em Principles of algebraic geometry}.
\newblock Wiley Classics Library. John Wiley \& Sons, Inc., New York, 1994.
\newblock Reprint of the 1978 original.

\bibitem[GKZ94]{GKZ}
I.~M. Gel'fand, M.~M. Kapranov, and A.~V. Zelevinsky.
\newblock {\em Discriminants, resultants, and multidimensional determinants}.
\newblock Mathematics: Theory \& Applications. Birkh\"{a}user Boston, Inc.,
  Boston, MA, 1994.

\bibitem[GS97]{grayson1997macaulay2}
D.~Grayson and M.~Stillman.
\newblock Macaulay 2--a system for computation in algebraic geometry and
  commutative algebra, 1997.

\bibitem[GZ89]{gilboa1989nash}
I.~Gilboa and E.~Zemel.
\newblock Nash and correlated equilibria: some complexity considerations.
\newblock {\em Games Econom. Behav.}, 1(1):80--93, 1989.

\bibitem[Har73]{harsanyi1973oddness}
J.~C. Harsanyi.
\newblock Oddness of the number of equilibrium points: a new proof.
\newblock {\em Internat. J. Game Theory}, 2:235--250, 1973.

\bibitem[Har77]{hartshorne1977algebraic}
R.~Hartshorne.
\newblock {\em Algebraic geometry}.
\newblock Graduate Texts in Mathematics, No. 52. Springer-Verlag, New
  York-Heidelberg, 1977.

\bibitem[JPS09]{jeronimo2009parametric}
G.~Jeronimo, D.~Perrucci, and J.~Sabia.
\newblock A parametric representation of totally mixed {N}ash equilibria.
\newblock {\em Comput. Math. Appl.}, 58(6):1126--1141, 2009.

\bibitem[JvS22]{jahani2022automated}
S.~Jahani and B.~von Stengel.
\newblock Automated equilibrium analysis of $2\times 2\times 2$ games.
\newblock In {\em Algorithmic Game Theory: 15th International Symposium, SAGT
  2022, Colchester, UK, September 12--15, 2022, Proceedings}, pages 223--237,
  Berlin, Heidelberg, 2022. Springer-Verlag.

\bibitem[Kap45]{kaplansky1945contribution}
I.~Kaplansky.
\newblock A contribution to von {N}eumann's theory of games.
\newblock {\em Ann. of Math. (2)}, 46:474--479, 1945.

\bibitem[KNP25]{kidambi2025elliptic}
A.~Kidambi, E.~Neuhaus, and I.~Portakal.
\newblock Elliptic curves in game theory.
\newblock {\em \arxiv{2501.14612}}, 2025.

\bibitem[Kou76]{kouchnirenko1976polyedres}
A.~G. Kouchnirenko.
\newblock Poly\`edres de {N}ewton et nombres de {M}ilnor.
\newblock {\em Invent. Math.}, 32(1):1--31, 1976.

\bibitem[Kre81]{kreps1981finite}
V.~L. Kreps.
\newblock Finite {$N$}-person noncooperative games with unique equilibrium
  points.
\newblock {\em Internat. J. Game Theory}, 10(3-4):125--129, 1981.

\bibitem[Ler20]{lerario2023topology}
A.~Lerario.
\newblock Lectures on metric algebraic geometry, 2020.
\newblock
  \url{https://drive.google.com/file/d/1A6UzYuv1OjucRscwZOQ4mDakKSfk_c77/view}.

\bibitem[LGMSY24]{lopez2024real}
J.~Lopez~Garcia, K.~Maluccio, F.~Sottile, and T.~Yahl.
\newblock Real solutions to systems of polynomial equations in {M}acaulay2.
\newblock {\em J. Softw. Algebra Geom.}, 14(1):87--95, 2024.

\bibitem[Lim05]{lim2005singular}
L.~H. Lim.
\newblock Singular values and eigenvalues of tensors: a variational approach.
\newblock In {\em 1st IEEE International Workshop on Computational Advances in
  Multi-Sensor Adaptive Processing, 2005.}, pages 129--132, 2005.

\bibitem[Mac15]{MacMahon1915combinatory}
P.~A. MacMahon.
\newblock {\em Combinatory Analysis}, volume~1.
\newblock Cambridge University Press, 1915.

\bibitem[MM96]{mckelvey1996computation}
R.~D. McKelvey and A.~McLennan.
\newblock Computation of equilibria in finite games.
\newblock In {\em Handbook of Computational Economics}, volume~1, chapter~02,
  pages 87--142. Elsevier, 1 edition, 1996.

\bibitem[MM97]{mckelvey1997maximal}
R.~D. McKelvey and A.~McLennan.
\newblock The maximal number of regular totally mixed {N}ash equilibria.
\newblock {\em J. Econom. Theory}, 72(2):411--425, 1997.

\bibitem[MNS25]{muller2025multilinear}
T.~Muller, V.~Nanda, and A.~Seigal.
\newblock Multilinear hyperquiver representations.
\newblock {\em Found. Comput. Math.}, pages 1--43, 2025.

\bibitem[Nas50]{nash1950equilibrium}
J.~F. Nash, Jr.
\newblock Equilibrium points in {$n$}-person games.
\newblock {\em Proc. Nat. Acad. Sci. U.S.A.}, 36:48--49, 1950.

\bibitem[OSV21]{ottaviani2021asymptotics}
G.~Ottaviani, L.~Sodomaco, and E.~Ventura.
\newblock Asymptotics of degrees and {ED} degrees of {S}egre products.
\newblock {\em Adv. in Appl. Math.}, 130:Paper No. 102242, 36, 2021.

\bibitem[PS22]{portakal2022geometry}
I.~Portakal and B.~Sturmfels.
\newblock Geometry of dependency equilibria.
\newblock {\em Rend. Istit. Mat. Univ. Trieste}, 54:Art. No. 5, 26, 2022.

\bibitem[PW24]{portakal2024dependency}
I.~Portakal and D.~Windisch.
\newblock Dependency equilibria: Boundary cases and their real algebraic
  geometry.
\newblock {\em \arxiv{2405.19054}}, 2024.

\bibitem[Rag70]{raghavan1970completely}
T.~E.~S. Raghavan.
\newblock Completely mixed strategies in bimatrix games.
\newblock {\em J. Lond. Math. Soc.}, 2(Part\_4):709--712, 1970.

\bibitem[RW08]{raichev2008asymptotics}
A.~Raichev and M.~C. Wilson.
\newblock Asymptotics of coefficients of multivariate generating functions:
  improvements for smooth points.
\newblock {\em Electron. J. Combin.}, 15(1):Research Paper 89, 17, 2008.

\bibitem[Sch14]{schneider2014convex}
R.~Schneider.
\newblock {\em Convex bodies: the {B}runn-{M}inkowski theory}, volume 151 of
  {\em Encyclopedia of Mathematics and its Applications}.
\newblock Cambridge University Press, Cambridge, expanded edition, 2014.

\bibitem[Sel75]{selten1975reexamination}
R.~Selten.
\newblock Reexamination of the perfectness concept for equilibrium points in
  extensive games.
\newblock {\em Internat. J. Game Theory}, 4:25--55, 1975.

\bibitem[Slo]{oeis}
N.~Sloane.
\newblock The on-line encyclopedia of integer sequences.
\newblock Available at \href{http://oeis.org}{oeis.org}.

\bibitem[SSV23]{shahidi2021degrees}
Z.~Shahidi, L.~Sodomaco, and E.~Ventura.
\newblock Degrees of {K}alman varieties of tensors.
\newblock {\em J. Symbolic Comput.}, 114:74--98, 2023.

\bibitem[Stu02]{sturmfels2002solving}
B.~Sturmfels.
\newblock {\em Solving systems of polynomial equations}, volume~97 of {\em CBMS
  Regional Conference Series in Mathematics}.
\newblock Conference Board of the Mathematical Sciences, Washington, DC; by the
  American Mathematical Society, Providence, RI, 2002.

\bibitem[Vid17]{vidunas2017counting}
R.~Vidunas.
\newblock Counting derangements and {N}ash equilibria.
\newblock {\em Ann. Comb.}, 21(1):131--152, 2017.

\bibitem[Wil71]{wilson1971computing}
R.~Wilson.
\newblock Computing equilibria of n-person games.
\newblock {\em SIAM J. Appl. Math.}, 21(1):80--87, 1971.

\end{thebibliography}

\end{document}